\documentclass[english]{article}
\usepackage{geometry}
\usepackage{color}
\usepackage{verbatim}
\usepackage{url}
\usepackage{bm}
\usepackage{amsmath}
\usepackage{amsthm}
\usepackage{amssymb}

\makeatletter

\numberwithin{equation}{section}
\numberwithin{figure}{section}
\theoremstyle{plain}
\newtheorem{thm}{\protect\theoremname}[section]
\theoremstyle{remark}
\newtheorem{rem}[thm]{\protect\remarkname}
\theoremstyle{plain}
\newtheorem{lem}[thm]{\protect\lemmaname}
\theoremstyle{plain}
\newtheorem{prop}[thm]{\protect\propositionname}

\usepackage{newtxtext}
\allowdisplaybreaks

\makeatother

\usepackage{babel}
\providecommand{\lemmaname}{Lemma}
\providecommand{\propositionname}{Proposition}
\providecommand{\remarkname}{Remark}
\providecommand{\theoremname}{Theorem}

\begin{document}
\title{Rate of convergence towards equations of Hartree type\\ for mixture
condensates with factorized initial data}
\author{Jinyeop Lee\thanks{\protect\url{jinyeoplee@kaist.ac.kr}}}
\date{Department of Mathematical Sciences, KAIST, Daejeon, Korea}
\maketitle
\begin{abstract}
We consider a system of $p$ components of bosons, each of which consists
of $N_{1},N_{2},\dots,N_{p}$ particles, respectively. The bosons
are in three dimensions with interactions via a generalized interaction
potential which includes the Coulomb interaction. We set the initial
condition to describe a mixture condensate, i.e., a tensor product
of factorized states. We show that the difference between the many-body
Schrödinger evolution in the mean-field regime and the corresponding
$p$-particle dynamics due to a system of Hartree equation is $O(N^{-1})$
where $N=\sum_{q=1}^{p}N_{q}$.
\end{abstract}

\section{Introduction}

Bose-Einstein condensation (BEC), which shows macroscopic quantum
phenomena, has been explored widely by both theorists and experimentalists.
People have been observing the dynamical evolution of the mixture
condensates \cite{hall2008multi,malomed2008multi}. The first experimental
realization of BEC was obtained with a single component element, typically
$^{87}\mathrm{Rb}$ \cite{hall1998dynamics,hall1998measurements,matthews1998dynamical,myatt1997production}.
Moreover, BEC may consist of different components of bosons. Two-component
BEC, such as $^{41}\mathrm{K}$-$^{87}\mathrm{Rb}$ \cite{modugno2001bose},
$^{41}\mathrm{K}$-$^{85}\mathrm{Rb}$ \cite{modugno2002two}, $^{39}\mathrm{K}$-$^{85}\mathrm{Rb}$
\cite{mancini2004observation} and $^{85}\mathrm{K}$-$^{87}\mathrm{Rb}$
\cite{PappWieman2006Observation} $^{87}\mathrm{Rb}$-$^{133}\mathrm{Cs}$\cite{Molony2014Creation},
was also successfully reported. Furthermore, in \cite{Molony2014Creation},
Gonzalez observed an absorption images of triply degenerate quantum
gases of $^{41}\mathrm{K}$-$^{40}\mathrm{K}$-$^{6}\mathrm{Li}$
\cite{santiago2012linak}. If multiple components of bosons are interacting,
one may consider both intra-component and inter-component interactions.

We are interested in the time-evolution of a condensate of mixture
of gases with $p$ components. Here, the number of components, $p$,
can be any positive integer. We assume that the gas is in
a state of complete condensation in each component, i.e., fully factorized.
We claim that under suitable conditions on the density and the interactions
between particles, under which the condensed system with $p$ components
of bosons, is condensed at time $t\geq0$. Additionally, the orbitals
of each component are similar to the solutions of the system of equations
(\ref{sysHartree}).

In order to consider a mixture condensate with $N$ particles consisting
of $p$ components of bosons in three dimensions, we first label $q=1,\dots,p$
for each component and let the number of particles of $q$-th type,
$N_{q}$, for $q=1,\dots,p$. Thus we have
\[
N=\sum_{q=1}^{p}N_{q}.
\]
The system is in the Hilbert space $\text{\ensuremath{\mathcal{H}}}$
such that
\[
\mathcal{H}=\bigotimes_{q=1}^{p}L_{s}^{2}(\mathbb{R}^{3N_{q}}),
\]
where $L_{s}^{2}(\mathbb{R}^{3N_{q}})$ denotes the symmetric subspace
of $L^{2}(\mathbb{R}^{3N_{q}})$.

To introduce our Hamiltonian, we define an \emph{embedding operator}
$E^{(q)}$ as follows: For an operator $\mathcal{O}$ acting on $L_{s}^{2}(\mathbb{R}^{3N_{q}})$,
for $q=1,\dots,p$,
\[
E^{(q)}\mathcal{O}=(\bigotimes_{r<q}I)\otimes\underset{\text{ }}{\mathcal{O}}\otimes(\bigotimes_{r>q}I)
\]
where $I$ is the identity operator on $L_{s}^{2}(\mathbb{R}^{3N_{q}})$.
Note that (i) if $q=1$, then $\bigotimes_{r<q}I=1$ and (ii) if $q=p$,
then $\bigotimes_{r>q}I=1$. Thus, $E^{(q)}\mathcal{O}$ can be understood
as an operator which is acting on the Hilbert space $\mathcal{H}$.
Then $E^{(q)}\mathcal{O}$ acts only on the $q$-th component state
with $\mathcal{O}$. For other sectors, it behaves like the identity
operator. In other words, one can write
\begin{align*}
E^{(1)}\mathcal{O} & =\mathcal{O}\otimes I\otimes\dots\otimes I,\\
E^{(2)}\mathcal{O} & =I\otimes\mathcal{O}\otimes I\otimes\dots\otimes I,\\
 & \vdots\\
E^{(q)}\mathcal{O} & =I\otimes\dots\otimes I\otimes\underset{\text{the \ensuremath{q}-th type }}{\underbrace{\mathcal{O}}}\otimes I\otimes\dots\otimes I,\\
 & \vdots\\
E^{(p)}\mathcal{O} & =I\otimes\dots\otimes I\otimes\mathcal{O}.
\end{align*}
It describes an operator that is only acting on the $q$-th component
particle.

Now, the mean-field Hamiltonian of the system has the form
\[
H_{p}=\sum_{q=1}^{p}E^{(q)}h_{q}+\sum_{q<r}\mathcal{V}_{qr}
\]
where, for $q=1,2,\dots,p$, the $q$-th particle Hamiltonian $h_{p}$
is given by
\begin{equation}
h_{q}=\sum_{j=1}^{N_{q}}(-\Delta_{x_{qj}})+\frac{1}{N_{q}}\sum_{1\leq j<k\leq N_{p}}V_{qq}(x_{qj}-x_{qk})\label{eq:one_particle_Ham}
\end{equation}
and the interaction between the $q$-th particle and the $r$-th particle
$\mathcal{V}_{qr}$
\begin{equation}
\mathcal{V}_{qr}=\frac{1}{N}\sum_{j=1}^{N_{q}}\sum_{k=1}^{N_{r}}V_{qr}(x_{qj}-x_{rk}).\label{eq:inter_type_interaction}
\end{equation}
For each $q,r=1,2,\dots,p$, we assume that each interacting potential
$V_{qr}$ satisfies the operator inequality
\[
V_{qr}^{2}\leq K(1-\Delta)
\]
on $L^{2}(\mathbb{R}^{3})$ for some constant $K>0$. We let $c_{qr}$
be the proportion of the particle number of $r$-th component among
all particles, i.e. $N_{r}=c_{qr}N$ for each $q=1,2,\dots,p$.

Note that in front of the interaction potentials, we have different
scaling limits $N_{p}^{-1}$ and $N^{-1}$ for intra- and inter- component
interactions. We justify these scaling by the following: we have $N=\sum_{q=1}^{p}N_{q}$
kinetic terms and $\frac{1}{2} \binom{N}{2}=\frac{N(N-1)}{2}$ potential
terms for the Hamiltonian. If we follow the sailings of the above,
the potential energy is 
\[
O(N)=\sum_{q=1}^{p}\frac{1}{N_{q}}\binom{N_{q}}{2}+\frac{1}{N}\frac{N(N-1)}{2}.
\]
This is a generalization of the scalings given in \cite[Section 1]{1811.04984}
of a mixture of two-component. A justification for our choice is provided
in \cite[Section 4]{Michelangeli2017Mean-field}.

We consider the $p$-component condensate state $\psi_{\mathbf{N}}\in\mathcal{H}$
of the form
\[
\psi_{\mathbf{N}}=\bigotimes_{q=1}^{p}(\prod_{j=1}^{N_{q}}u_{q})
\]
where $\mathbf{N}=(N_{1},\dots,N_{p})$ and $u_{q}\in H^{1}(\mathbb{R}^{3})$
with $\|u_{q}\|_{L^{2}}=1$.

We expect, as $N_{q}\to\infty$ for all $q=1,\dots,p$, the condensation
remains approximately at time $t>0$. In other words, we hope to have
\[
e^{-\mathrm{i}\mathcal{H}_{\mathbf{N}}\psi_{\mathbf{N}}}:=\psi_{\mathbf{N},t}\simeq\bigotimes_{q=1}^{p}(\prod_{j=1}^{N_{q}}u_{q,t})
\]
where $(u_{q,t})_{q=1}^{p}$ is the solution of the system of equations
of Hartree type
\begin{equation}
i\partial_{t}u_{q,t}=-\Delta u_{q,t}+(V_{qq}*|u_{q,t}|^{2})u_{q,t}+\sum_{r=1}^{p}c_{qr}(V_{qr}*|u_{q,t}|^{2})u_{q,t}\label{sysHartree}
\end{equation}
with $u_{q,0}=u_{q}$ for all $q=1,\dots,p$.

We show that the difference between the many-body Schrödinger evolution
in the mean-field regime and the corresponding equations of Hartree
type is $O(N^{-1})$ where $N=\sum_{q=1}^{p}N_{q}$. To precisely
describe the meaning of the approximate factorization, we let $\mathbf{I}=(1,1,\dots,1)$.
Then, we introduce the following $\mathbf{I}$-particle density tensor
\begin{equation}
\gamma_{\mathbf{N},t}^{\mathbf{I}}=\operatorname{Tr}_{\mathbf{N}-\mathbf{I}}\,|\psi_{\mathbf{N},t}\rangle\langle\psi_{\mathbf{N},t}|,\label{eq:def_double_partial_trace}
\end{equation}
where $|\psi_{\mathbf{N},t}\rangle\langle\psi_{\mathbf{N},t}|$ denotes
the orthogonal projection onto the state $\psi_{\mathbf{N},t}$ and
$\operatorname{Tr}_{\mathbf{N}-\mathbf{I}}$ denotes the trace over
$(N_{q}-1)$ variables corresponding to the $q$-th factor of $\mathcal{H}$
for all $q=1,\dots,p$. The operator $\gamma_{\mathbf{N},t}^{\mathbf{I}}$
is a non-negative trace class operator on $\bigotimes_{q=1}^{p}L^{2}(\mathbb{R}^{3})$
with integral kernel 
\begin{equation}
\begin{split}\gamma_{\mathbf{N},t}^{\mathbf{I}}(\mathbf{x};\mathbf{x}')\;=\; & \int_{\prod_{q=1}^{p}\mathbb{R}^{3(N_{q}-1)}}\mathrm{d}\mu\,\overline{\psi_{\mathbf{N},t}}\psi_{\mathbf{N},t}\end{split}
\label{eq:double_trace-KERNEL}
\end{equation}
where $\mathbf{x}=(x_{11},x_{21},\dots,x_{p1})$, $\mathbf{x}'=(x_{11}',x_{21}',\dots,x_{p1}')$,
and $\mathrm{d}\mu$ denotes the infinitesimal product measure for
all components of particle except the first particle of each components,
i.e. 
\[
\mathrm{d}\mu=\prod_{q=1}^{p}\prod_{n_{q}=2}^{N_{q}}\mathrm{d}x_{qn_{q}}.
\]
Note that if $\|\psi_{\mathbf{N},t}\|_{2}=1$, we have $\operatorname{Tr}\gamma_{\mathbf{N},t}^{\mathbf{I}}=1$.

Thus, to show that the state is in approximately factorized state,
it is to show that
\begin{equation}
\lim_{N_{q}\to\infty}\operatorname{Tr}\,\left|\gamma_{\mathbf{N},t}^{\mathbf{I}}-|\mathbf{u}_{t}\rangle\langle\mathbf{u}_{t}|\right|\,=0\label{eq:BEC-2component}
\end{equation}
where
\[
\mathbf{u}_{t}=\bigotimes_{q=1}^{p}u_{q,t}.
\]
It can be rephrased as $\gamma_{\mathbf{N},t}^{\mathbf{I}}\simeq|\mathbf{u}_{t}\rangle\langle\mathbf{u}_{t}|$
when $N_{q}\to\infty$ for all $q=1,\dots,p$. The limit above describes
large systems, i.e. thermodynamic limit. In fact, all of the eigenvalues
of $\gamma_{\mathbf{N},t}^{\mathbf{I}}$ are non-negative and $\operatorname{Tr}_{\mathbf{N}}\,\gamma_{\mathbf{N},t}^{\mathbf{I}}=1$.

Historically, the cases with single-component, $p=1$, were studied
intensively. First, Spohn \cite{Spohn1980} proved the convergence,
$\operatorname{Tr}|\gamma_{N,t}^{(1)}-|\varphi_{t}\rangle\langle\varphi_{t}||\to0$
as $N\to\infty$, for bounded potential. Moreover, it was extended
by Ginibre and Velo\textcolor{red}{{} }\cite{Ginibre1979_1,Ginibre1979_2}.
Using the BBGKY hierarchy technique, Erd\H{o}s and Yau \cite{ErdosYau2001}
extended the result in \cite{Spohn1980} to singular potential (including
the Coulomb case). Rodnianski and Schlein in \cite{Rodnianski2009},
based on coherent state approach, obtained the trace norm difference
\begin{equation}
\operatorname{Tr}\left|\gamma_{N,t}^{(1)}-|\varphi_{t}\rangle\langle\varphi_{t}|\right|\label{eq:trace_norm_bound}
\end{equation}
which is bounded by $O(N^{-1/2})$ as well as it goes to zero. Knowles
and Pickl \cite{KnowlesPickl2010} considered more singular interaction
potentials and obtained similar estimates on the rate of convergence.
The proof in \cite{KnowlesPickl2010} is based on the use of projection
operators in the $N$-particle space $L_{s}^{2}(\mathbb{R}^{3N})$,
and covers a large class of possibly time-dependent external potentials
even though it does not gives the same rate of convergence. Grillakis,
Machedon, and Margetis \cite{Grillakis2010,Grillakis2011} observed
the second-order corrections. The optimal rate of convergence $O(N^{-1})$
was proved by Chen, Lee, and Schlein in \cite{ChenLeeSchlein2011}
for the Coulomb case. It is known that the rate of convergence $O(N^{-1})$
is optimal in as shown in e.g. \cite{Grillakis13,Grillakis17}. It
was later extended in \cite{ChenLeeLee2018} to cover the case $V\in L^{2}+L^{\infty}$.
We also remark that Hott \cite{Hott2018} pointed out that the same
rate of convergence can be obtained in larger space than $H^{1}(\mathbb{R}^{3})$
for $V(x)=\lambda|x|^{-\gamma}$ with $1<\gamma<3/2$. Using time
decay estimate and Strichartz estimate, the time dependence of the
trace norm bound was investigated in  \cite{Lee2019}. A similar approach
can also be applied to many-body semi-relativistic Schrödinger equations
which describe a boson star. Lee \cite{Lee2013} provided the optimal
rate of convergence $O(N^{-1})$ for Coulomb interaction. Following
the approach presented in this article, it is believed that one can
obtain a corresponding bound for the semi-relativistic case by exploiting
the properties of the mean-field solution. On the other hand, Gross-Pitaevskii
dynamics for BEC has also been heavily studied, for example in \cite{Boccato2018CompleteBEC,Boccato1812.03086,Boccato2017,Erdos2006Derivation-of-the-GP,Erdos2010Derivation-of-the-GP}.

The corresponding result for two-component BEC, i.e., $p=2$, has
also been established well, e.g. \cite{Anapolitanos2017,MichelangeliOlgiati2017Gross,Michelangeli2017Mean-field,Olgiati2017}.
In \cite{MichelangeliOlgiati2017Gross}, Michelangeli and Olgiati
derived effective evolution equations rigorously from the many-body
Schrödinger equation using counting method developed in \cite{KnowlesPickl2010,Pickl2010,Pickl2011,Pickl2015}.
A similar result was obtained by Anapolitanos, Hott, and Hundertmark
\cite{Anapolitanos2017}. Recently, de Oliveira and Michelangeli \cite{1811.04984}
obtained explicit bound $O(N^{-1/2})$ on the rate of the convergence
in the the Coulomb type interaction, assuming a coherent initial state.
In this paper, our goal is to improve the bound of the rate of convergence
$O(N^{-1})$ for general $p\geq2$, assuming not a coherent initial
condition but a factorized initial condition. The following main theorem
gives the desired goal:
\begin{thm}
\label{thm:Main_Theorem}Suppose that all the interaction potentials
$V_{qr}$ satisfy the operator inequality 
\begin{equation}
V_{qr}^{2}\le K(1-\Delta)\label{eq:assumption_V}
\end{equation}
on $L^{2}(\mathbb{R}^{3})$ for some constant $K>0$ for all $q,r=1,\dots,p$.
Suppose that each $N_{q}$ is a positive integer, $N=\sum_{q=1}^{p}N_{q}$,
and $c_{q}$ is real number that satisfies 
\begin{equation}
c_{q}=N_{q}/N\label{eq:definition-of-cq}
\end{equation}
for some constant $D>0$ where $c_{q}\neq0$ with $\sum_{q=1}^{p}c_{q}=1$.
For $\mathbf{N}=(N_{1},\dots,N_{p})$ and $\mathbf{I}=(1,\dots,1)$,
let $\gamma_{\mathbf{N},t}^{\mathbf{I}}$ be the reduced density operator
associated to the solution 
\[
\psi_{t}=e^{-it\mathcal{H}_{\mathbf{N}}}\psi
\]
with $\psi=\bigotimes_{q=1}^{p}u_{q}^{\otimes N_{q}}=:\mathbf{u}^{\mathbf{N}}$
to the Schrödinger equation. Then,
\begin{equation}
\operatorname{Tr}\,\Big|\,\gamma_{\mathbf{N},t}^{\mathbf{I}}-|\mathbf{u}_{t}\rangle\langle\mathbf{u}_{t}\,\Big|\le C\left((\min_{q}c_{q})^{-1}p^{Cp}e^{Kt}\right)\cdot\frac{1}{N}\label{ineqcoherent}
\end{equation}
for all $t\ge0$, where $C$ and $K$ are positive constants that
do not depend on $N_{q}$, and $u_{q,t}$ are the solutions of the
system of equations (\ref{sysHartree}) with initial datum $u_{q}$.
\end{thm}

\begin{rem}
The number of components $p$ may dependent on $N$. Letting $p=\left\lfloor \log\log N\right\rfloor $
where $\left\lfloor x\right\rfloor $ denotes the largest integer
less than or equal to $x$, we get
\[
\operatorname{Tr}\,\Big|\,\gamma_{\mathbf{N},t}^{\mathbf{I}}-|\mathbf{u}_{t}\rangle\langle\mathbf{u}_{t}|\,\Big|\le Ce^{Kt}\cdot\frac{(\log N)^{C}}{N}
\]
for some constant $C$. This result shows that BEC evolves asymptotically
through the Hartree dynamics even when the number of components is
sufficiently large (formally infinite as $N\to\infty$). This case,
we need to have much more particles than the components of bosons.
Here, if we have $p$-component BEC for large $p$, then we need to
have the order of $\exp(\exp(p))$ particles for each component.
\end{rem}

\begin{rem}
Using result i \cite{Lee2019}, with appropriate initial data, one
can generalize interaction potentials $V_{qr}$ appeared in (\ref{eq:one_particle_Ham})
and (\ref{eq:inter_type_interaction}) to $V(x)=\lambda e^{-\mu|x|}/|x|^{\gamma}$
with $\lambda\in\mathbb{R}$, $\mu\geq0$, and $0<\gamma<3/2$. Moreover,
if all the potentials are changed into the same component, the time
dependence can be obtained similarly as in \cite{Lee2019}. Using
result in \cite{Lee2013}, one can have a similar result for semi-relativistic
case. Thus we expect to have a similar result for boson star.
\end{rem}

\begin{rem}
\textcolor{black}{For $p\geq2$, we let $c_{q}=N_{q}/N$ in (\ref{eq:definition-of-cq}).
Instead of this, we can let 
\[
\left|\frac{N_{r}}{N}-c_{qr}\right|\le\frac{C}{N^{2/p}}
\]
to have the same rate of convergence. Similar assumption was appeared
in \cite{1811.04984} for $p=2$ case. See Remark \ref{rem:|NrN-Cqr|}
for more detail.}
\end{rem}

\begin{rem}
We omit the information about the mass $m_{q}$ of each component.
If one have mass $m_{q}$ factors 
\[
h_{q}=\sum_{j=1}^{N_{q}}\left(-\frac{\Delta_{x_{qj}}}{m_{q}}\right)+\frac{1}{N_{q}}\sum_{1\leq j<k\leq N_{p}}V_{qq}(x_{qj}-x_{qk})
\]
by scaling, by scaling to the solution $u_{q}$, one can move the
information about mass $m_{q}$ from the kinetic terms 
\[
\sum_{j=1}^{N_{q}}\left(-\frac{\Delta_{x_{qj}}}{m_{q}}\right)
\]
to potential terms
\[
\frac{1}{N_{q}}\sum_{1\leq j<k\leq N_{p}}V_{qq}(x_{qj}-x_{qk}).
\]
Hence through out this paper, we let $m_{q}=1/2$ for all $q=1,\dots,p$.
\end{rem}

We derive a system of effective equations of Hartree type by generalizing
the methods in Fock space developed by Hepp \cite{Hepp1974}, Ginibre
and Velo \cite{Ginibre1979_1,Ginibre1979_2}, and Rodnianski and Schlein
\cite{Rodnianski2009} (see also \cite{Benedikter2015}). A review
of these methods for single-component Bose gases can be found in \cite{Benedikter2016}.
The methods in Fock space have been studied for dynamical properties of single-component condensates \cite{Benedikter2015,Boccato2017,Brenneck2019,ChenLee2011,ChenLeeSchlein2011,LewinNamSchlein2015,Rodnianski2009}.
By generalizing previous results, our result holds 
for factorized states as initial data (Theorem \ref{thm:Main_Theorem})
with Coulomb interaction.

The main subject of this paper is $p$-component BEC for given integer
$p$. Because of this extension, the analysis natually needs to be
extended from the single component Fock space $\mathcal{F}$ to the
$p$-tensor product Fock spaces $\mathcal{F}^{\otimes p}$. In this
paper, we introduce new and compact notations which is appropriate
to deal with generalized $p$-component BEC.

To deal with $p$-component BEC, an operator $J$, which is introduced
to bound the trace norm, acting on $\mathcal{H}$ should be understood.
Using tensor algebra, for finite $p$, an Hilbert-Schmidt operator
$J$ can be interpreted as an infinite sum of simple tensors, i.e.,
$\sum_{\mathfrak{k}=1}^{\infty}(\bigotimes_{q=1}^{p}J_{q}^{\mathfrak{k}})$.
Using this and the definition of Hilbert-Schmidt norm, one may recover
the Hilbert-Schmidt norm of the operator $J$ such that $\|J\|_{\mathrm{HS}}^{2}=\sum_{\mathfrak{k}=1}^{\infty}\|\bigotimes_{q=1}^{p}J_{q}^{\mathfrak{k}}\|_{\mathrm{HS}}^{2}.$
The detailed explanation will be given in Section \ref{subsec:Review-on-Tensor}.

This paper is organized as follows: First, in Section \ref{sec:preliminaries},
we introduce tensor algebra to describe operators on bosonic Fock
space for mixture condensation. After introducing tensor algebra we
introduce bosonic Fock space for mixture condensation, we omit the
definitions and properties of Fock space for the non-mixture cases
but we cite previous articles to avoid lengthy paper. Using tensor
algebra introduced in Section \ref{subsec:Review-on-Tensor}, it is
possible to generalize the definitions and properties of Fock space
in Section \ref{subsec:Fock-space-formalism-mixture}. It also helps
us to generalize the techniques provided in the previous works to
the $p$-component case. We prove our theorem in Section \ref{sec:Proof-main}.
The main strategy is to embed our state into the Fock space and compare
the time evolution of our stated and coherent state. For that, we
argue that the creation and annihilation operator are similar to the
square root of the number operator, which was suggested in \cite{Rodnianski2009}
for non-mixture case. Using the argument, we provide the proof of
Theorem \ref{thm:Main_Theorem}; this crucial bound is stated in Proposition
\ref{prop:Et1}, which will be proved in Section \ref{sec:Pf-of-Props}.
In Section \ref{sec:Comparison-dynamics}, we prepare lemmas describing
comparison dynamics to prove Proposition \ref{prop:Et1} in Section
\ref{sec:Pf-of-Props}. The lemmas are similar to the lemmas appeared
in previous works, e.g. see \cite{ChenLee2011,ChenLeeSchlein2011,Lee2013},
for the non-mixture case. In Section \ref{sec:Lemmas}, we provide
lemmas to back up the lemmas in Sections \ref{sec:Comparison-dynamics}
and \ref{sec:Pf-of-Props}.

\subsection*{Notional Remark}

For two elements $a,b$ of an inner product space $A$, we denote
by $\langle a,b\rangle_{A}$ the inner product of $a$ and $b$. Similarly,
we denote the norm of $a$ by $\|a\|_{A}$. We denote the standard
$L^{p}$-space in $d$-dimension as $L^{p}(\mathbb{R}^{d},\mathrm{d}x)$
or $L^{p}(\mathbb{R}^{d})$. Let $f\in L^{p}(\mathbb{R}^{d})$, we
use $\|f\|_{p}=\|f\|_{L^{p}(\mathbb{R}^{d})}$ for the standard $L^{p}$-norm
of $f$. We also let $\|f\|_{H^{p}}=\|f\|_{H^{p}(\mathbb{R}^{d})}$
for the standard Sobolev norm of $f\in H^{p}(\mathbb{R}^{d})$. We
denote $\|J\|_{\mathrm{op}}$ as the operator norm of an operator
$J$ and $\|J\|_{\mathrm{HS}}$ as the Hilbert-Schmidt norm of an
operator $J$. In inequalities we have constants $C$ which may be
different line by line. We basically use bold faced roman alphabet
to denote vectors in $\mathbb{R}^{p}$ or elements in $\mathcal{F}^{\otimes p}$
and we use bold faced italic alphabet to denote the product of all
component s of a vector in $\mathbb{R}^{p}$. For example, we denote
\[
\mathbf{N}=(N_{1},\dots,N_{p})\in\mathbb{N}^{p}\qquad\text{and}\qquad\boldsymbol{N}=\prod_{q=1}^{p}N_{q}.
\]
The detailed description will be shown in the article.

\section{Preliminaries\label{sec:preliminaries}}

\subsection{Review on Tensor Algebra\label{subsec:Review-on-Tensor}}

This section is to review tensor algebra to understand following sections
which describes the Fock space of mixture states. Moreover, to apply
the idea to bound the trace norm, we need to introduce the lemmas
given in this section.

Let $H_{1}$ and $H_{2}$ be Hilbert spaces endowed with scalar products
$\langle\cdot,\cdot\rangle_{H_{1}}$ and $\langle\cdot,\cdot\rangle_{H_{2}}$
respectively. Then \emph{inner product} of tensor product is defined
by
\[
\langle f_{1}\otimes f_{2},g_{1}\otimes g_{2}\rangle_{H_{1}\otimes H_{2}}=\langle f_{1},g_{1}\rangle_{H_{1}}\langle f_{2},g_{2}\rangle_{H_{2}}
\]
for each $f_{1},f_{2}\in H_{1}$ and $g_{1},g_{2}\in H_{2}$.

This leads us to a standard \emph{norm} $\|\cdot\|_{H_{1}\otimes H_{2}}$
derived by the inner product such that
\[
\|f\otimes g\|_{H_{1}\otimes H_{2}}=\langle f\otimes g,f\otimes g\rangle_{H_{1}\otimes H_{2}}^{1/2}=\langle f,f\rangle_{H_{1}}^{1/2}\langle g,g\rangle_{H_{2}}^{1/2}=\|f\|_{H_{1}}\|g\|_{H_{2}}.
\]
For Hilbert spaces $H_{q}$ for $q=1,\dots,p$, we generalize above
properties by induction for $\bigotimes_{q=1}^{p}H_{p}$.

The following result can be understood by reading \cite[Chapter II §4.4 Proposition 4]{bourbaki1998algebra}
and \cite[Section XIII.9]{ReedSimonIV1975}
\begin{lem}
\label{lem:tensor-operator-decomposition}Let $H$ be a tensor product
of two Hilbert spaces $H_{1}$ and $H_{2}$, i.e. $H=H_{1}\otimes H_{2}$.
Then given a Hilbert--Schmidt operator $J$ on $H$, it is a limit
of linear combination of simple tensors $J_{1}^{\mathfrak{k}}\otimes J_{2}^{\mathfrak{k}}$
for $J_{1}^{\mathfrak{k}}\in L(H_{1})$ and $J_{2}^{\mathfrak{k}}\in L(H_{2})$
for $\mathfrak{k}\in\mathbb{N}$, i.e., 
\[
J=\sum_{\mathfrak{k}=1}^{\infty}J_{1}^{\mathfrak{k}}\otimes J_{2}^{\mathfrak{k}}.
\]
\end{lem}

\begin{proof}
We will show that there exists a sequence in $L(H_{1})\otimes L(H_{2})$
weakly convergent to $J$. Let $\{e_{i}\otimes e_{j}\}_{i,j=1}^{\infty}$
be an orthogonal basis of $H_{1}\otimes H_{2}$. For any $n$ define
the truncated operator $J_{n}$ such that
\begin{equation}
J_{n}:=\sum_{i,j,i',j'=1}^{n}|e_{i'}\otimes e_{j'}\rangle\langle e_{i'}\otimes e_{j'}|J|e_{i}\otimes e_{j}\rangle\langle e_{i}\otimes e_{j}|.\label{eq:tensor-op-representation}
\end{equation}
Then
\begin{align*}
J_{n} & :=\sum_{i,j,i',j'=1}^{n}\langle e_{i'}\otimes e_{j'},Je_{i}\otimes e_{j}\rangle|e_{i'}\otimes e_{j'}\rangle\langle e_{i}\otimes e_{j}|\\
 & =\sum_{i,j,i',j'=1}^{n}\langle e_{i'}\otimes e_{j'},Je_{i}\otimes e_{j}\rangle\left(|e_{i'}\rangle\langle e_{i}|\otimes|e_{j'}\rangle\langle e_{j}|\right)\in L(H_{1})\otimes L(H_{2}).
\end{align*}
By rearranging $i,j,i',$ and $j'$ by $\mathfrak{k}$ and letting
\[
J_{1}^{\mathfrak{k}}=\langle e_{i'}\otimes e_{j'},Je_{i}\otimes e_{j}\rangle|e_{i'}\rangle\langle e_{i}|\quad\text{and}\quad J_{2}^{\mathfrak{k}}=|e_{j'}\rangle\langle e_{j}|,
\]
we get
\[
J_{n}=\sum_{\mathfrak{k}=1}^{\infty}J_{1}^{\mathfrak{k}}\otimes J_{2}^{\mathfrak{k}}.
\]
For more detail, see \cite[Chapter II §4.4 Proposition 4]{bourbaki1998algebra}.

For any $a\otimes b,c\otimes d\in H_{1}\otimes H_{2}$,
\begin{align*}
\langle a\otimes b,J_{n}c\otimes d\rangle & =\sum_{i,j,i',j'=1}^{n}\langle a\otimes b,e_{i'}\otimes e_{j'}\rangle\langle e_{i'}\otimes e_{j'},Je_{i}\otimes e_{j}\rangle\langle e_{i}\otimes e_{j},c\otimes d\rangle\\
 & =\left\langle \sum_{i',j'=1}^{n}\langle a\otimes b,e_{i'}\otimes e_{j'}\rangle e_{i'}\otimes e_{j'},\sum_{i,j=1}^{n}J\langle e_{i}\otimes e_{j},c\otimes d\rangle e_{i}\otimes e_{j}\right\rangle .
\end{align*}
Thus, as $n\to\infty,$ the boundedness of $J$ leads that
\[
\langle a\otimes b,J_{n}c\otimes d\rangle\to\langle a\otimes b,Jc\otimes d\rangle.
\]
For more detaeil, see \cite[Section XIII.9]{ReedSimonIV1975}.

Since the linear combination of $\{e_{i}\otimes e_{j}\}_{i,j=1}^{\infty}$
is dense in $H\otimes H$, we get the conclusion.
\end{proof}
\begin{prop}
A bounded linear Hilbert-Schmidt operator $J$ defined on $\mathcal{F}^{\otimes p}$
can be represented as
\[
J=\sum_{\mathfrak{k}=1}^{\infty}\left(\bigotimes_{q=1}^{p}J_{q}^{\mathfrak{k}}\right)=\sum_{\mathfrak{k}=1}^{\infty}\mathbf{J}^{\mathfrak{k}}
\]
such that $J\left(\bigotimes_{q=1}^{p}\psi_{q}\right)=\sum_{\mathfrak{k}=1}^{\infty}\left(\bigotimes_{q=1}^{p}J_{q}^{\mathfrak{k}}\right)\left(\bigotimes_{q=1}^{p}\psi_{q}\right)=\sum_{\mathfrak{k}=1}^{\infty}\left(\bigotimes_{q=1}^{p}J_{q}^{\mathfrak{k}}\psi_{q}\right)$
for $\psi_{q}\in\mathcal{F}$. Moreover, 
\begin{equation}
\|J\|_{\mathrm{HS}}^{2}=\sum_{\mathfrak{k}=1}^{\infty}\|\mathbf{J}^{\mathfrak{k}}\|_{\mathrm{HS}}^{2}.\label{eq:op-norm-identity}
\end{equation}
\end{prop}

\begin{proof}
First we consider the case with $p=2$. The proof for the linear combination
is straight from Lemma \ref{lem:tensor-operator-decomposition}. It
is left to show the operator norm identity (\ref{eq:op-norm-identity}).
Using (\ref{eq:tensor-op-representation}), by defining
\[
J_{i'i}\otimes J_{j'j}:=\langle e_{i'}\otimes e_{j'}|J|e_{i}\otimes e_{j}\rangle\left(|e_{i'}\rangle\langle e_{i}|\otimes|e_{j'}\rangle\langle e_{j}|\right),
\]
it is possible to represent
\[
J=\sum_{i,j,i',j'}J_{i'i}\otimes J_{j'j}.
\]
Then noting that $(J_{i'i}\otimes J_{j'j})^{*}$ and $(J_{i'i}\otimes J_{j'j})$
are positive operators,
\begin{align*}
\|J\|_{\mathrm{HS}}^{2} & =\|J^{*}J\|_{\mathrm{HS}}=\|(\sum_{i,j,i',j'}J_{i'i}\otimes J_{j'j})^{*}(\sum_{k,\ell,k',\ell'}J_{k'k}\otimes J_{\ell'\ell})\|_{\mathrm{HS}}\\
 & =\|\sum_{i,j,i',j'}(J_{i'i}\otimes J_{j'j})^{*}(J_{i'i}\otimes J_{j'j})\|_{\mathrm{HS}}=\sum_{i,j,i',j'}\|(J_{i'i}\otimes J_{j'j})^{*}(J_{i'i}\otimes J_{j'j})\|_{\mathrm{HS}}\\
 & =\sum_{i,j,i',j'}\|(J_{i'i}\otimes J_{j'j})\|_{\mathrm{HS}}^{2}.
\end{align*}
By induction we get the result for any positive integer $p\geq2$.
\end{proof}

\subsection{Fock space formalism for mixture condensation with $p$-component
of bosons} \label{subsec:Fock-space-formalism-mixture}

To avoid lengthy paper, we just cite the single-component Fock space
formalism for previous works (for example \cite{ChenLeeSchlein2011,Rodnianski2009,Lee2019,ChenLeeLee2018}).

We now consider $p$-copies of the bosonic Fock space $\mathcal{F}=\bigoplus_{n=0}^{\infty}L^{2}(\mathbb{R}^{3})^{\otimes_{s}n}$.
In our model, the state space for the $p$-component BEC is in
\[
\mathcal{F}^{\otimes p}.
\]
Here, we are using the standard tensor product of $p$-Hilbert spaces.
We will often use the standard construction of tensor product of operators,
as described in \cite{ReedSimonI1975}.

We introduce an \emph{(extended) embedding operator} $\mathcal{E}^{(q)}$
as follows: For an operator $\mathcal{O}$ acting on $\mathcal{F}$,
\[
\mathcal{E}^{(q)}\mathcal{O}=(\bigotimes_{r<q}I)\otimes\mathcal{O}\otimes(\bigotimes_{r>q}I)
\]
for $q=1,\dots,p$. Note that (i) if $q=1$, then $\bigotimes_{r<q}I=1$
and (ii) if $q=p$, then $\bigotimes_{r>q}I=1$. In other words,
\[
\mathcal{E}^{(q)}\mathcal{O}=I\otimes\dots\otimes\underset{\text{the \ensuremath{q}-th type }}{\underbrace{\mathcal{O}}}\otimes\dots\otimes I
\]
For example, one can embed an operator $\mathcal{O}$ on $\mathcal{F}$
to an operator on $\mathcal{F}^{\otimes p}$ by the following ways:
\[
\mathcal{E}^{(1)}\mathcal{O}=\mathcal{O}\otimes I\otimes\dots\otimes I,
\]
\[
\mathcal{E}^{(2)}\mathcal{O}=I\otimes\mathcal{O}\otimes I\otimes\dots\otimes I,
\]
or
\[
\mathcal{E}^{(p)}\mathcal{O}=I\otimes\dots\otimes I\otimes\mathcal{O}.
\]
Similar to $E^{(q)}$, the extended embedding operator $\mathcal{E}^{(q)}$
can be understood as it embeds an operator $\mathcal{O}$ on $\mathcal{F}$
to an operator on $\mathcal{F}^{\otimes p}$. It describes an operator
that is only acting on the $q$-th component particle.

For $f\in L^{2}(\mathbb{R}^{3})$, we define creation and annihilation
operators on each factor of $\mathcal{F}^{\otimes p}$ by using the
embedding operator and the standard creation and annihilation operators
($a^{*}(f)$ and $a(f)$ respectively) on $\mathcal{F}$ defined in
\cite{ChenLeeLee2018,ChenLeeSchlein2011,Lee2019,Rodnianski2009},
\[
a^{(q)*}(f)=\mathcal{E}^{(q)}a^{*}(f)\qquad\text{and}\qquad a^{(q)}(f)=\mathcal{E}^{(q)}a(f).
\]
For $\mathbf{f}=(f_{1},f_{2},\dots,f_{p})$ with $f_{q}\in L^{2}(\mathbb{R}^{3})$
for $q=1,2,\dots,p$, we define the product of all creation and annihilation
operators for all $q=1,\dots,p$,
\[
\boldsymbol{a}^{*}(\mathbf{f}):=\prod_{q=1}^{p}a^{(q)*}(f_{p})\qquad\text{and}\qquad\boldsymbol{a}^{*}(\mathbf{f}):=\prod_{q=1}^{p}a^{(q)}(f_{p}).
\]
The corresponding operator valued distributions are given by 
\begin{equation}
\boldsymbol{a}_{\mathbf{x}}^{*}:=\prod_{q=1}^{p}\mathcal{E}^{(q)}a_{x_{q}}^{*}\qquad\text{and}\qquad\boldsymbol{a}_{\mathbf{x}}:=\prod_{q=1}^{p}\mathcal{E}^{(q)}a_{x_{q}}.\label{eq:creation}
\end{equation}
It is useful to consider self-adjoint operators
\[
\Phi(\mathbf{f})=\bigotimes_{q=1}^{p}\phi(f_{q})
\]
and
\[
\boldsymbol{\phi}^{(Q)}(\mathbf{f})=\bigotimes_{q\in Q}\phi(f_{q})
\]
for $Q\subset\{1,\dots,p\}$.

There are several commutation relations for the operators $a^{(q)*}(f)$
and $a^{(q)}(f)$. The only commutators that are not equal to zero
are 
\[
[a^{(q)}(f),a^{(r)*}(g)]=\delta_{qr}\langle f,g\rangle.
\]
Consequently,
\[
[a_{x}^{(q)},a_{y}^{(r)*}]=\delta_{qr}\delta(x-y).
\]
Using the operator-valued distributions, we let $\mathcal{N}^{(q)}$
to denote the number of particles of component $q$, i.e.
\[
\mathcal{N}^{(q)}=\mathcal{E}^{(q)}\int dx\,a_{x}^{*}a_{x}.
\]
Then, we define
\[
\mathcal{N}_{\mathrm{total}}:=\sum_{q=1}^{p}\mathcal{E}^{(q)}\mathcal{N}.
\]
It leads us to have such identity
\[
\mathcal{N}_{\mathrm{total}}+p\mathcal{I}=\sum_{q=1}^{p}\mathcal{E}^{(q)}\mathcal{N}+p\bigotimes_{q=1}^{p}1=\sum_{q=1}^{p}\mathcal{E}^{(q)}(\mathcal{N}+1)
\]
where $\mathcal{I}$ is the identity operator which can be written
as
\[
\mathcal{I}=\bigotimes_{q=1}^{p}1.
\]

\begin{rem}
We want to remark that $\mathcal{N}_{\mathrm{total}}+p\mathcal{I}$
is a generalization of $\mathcal{N}+1$ given in previous works for
example in \cite{Rodnianski2009} and $(\mathcal{N}+1)\otimes I+I\otimes(\mathcal{N}+1)$
in \cite{1811.04984}. For each cases, if we put $p=1$ or $p=2$,
it matches the definition.
\end{rem}

We also generalize projection operator $P_{n}$. For each non-negative
integers $n_{q}$, let $\mathbf{n}=(n_{1},\dots,n_{p})$. We introduce
the projection operator $\mathcal{P}_{\mathbf{n}}$ onto the $\mathbf{n}$-particle
sector of the Fock space $\mathcal{F}^{\otimes p}$, 
\[
\mathcal{P}_{\mathbf{n}}(\psi):=\bigotimes_{q=1}^{p}P_{n_{q}}\psi_{q}
\]
for $\psi=\bigotimes_{q=1}^{p}\psi_{q}=(\psi_{1}^{(0)},\psi_{1}^{(1)},\dots)\otimes(\psi_{2}^{(0)},\psi_{2}^{(1)},\dots)\otimes\dots\in\mathcal{F}$.
For simplicity, with slight abuse of notation, we will use $\bigotimes_{q=1}^{p}\psi_{q}^{(n_{q})}$
(or more compactly $\psi^{(\mathbf{n})}$) to denote $\mathcal{P}_{\mathbf{n}}(\psi)$.

For this space we define inner product of $\psi=\bigotimes_{q=1}^{p}\psi_{q}$
and $\phi=\bigotimes_{q=1}^{p}\phi_{q}$ by
\[
\langle\psi,\phi\rangle_{\mathcal{F}^{\otimes p}}=\langle\bigotimes_{q=1}^{p}\psi_{q},\bigotimes_{q=1}^{p}\phi_{q}\rangle_{\mathcal{F}^{\otimes p}}=\prod_{q=1}^{p}\langle\psi_{q},\phi_{q}\rangle_{\mathcal{F}}=\prod_{q=1}^{p}\sum_{n_{q}\geq1}\langle\psi_{q}^{(n_{q})},\phi_{q}^{(n_{q})}\rangle_{L^{2}(\mathbb{R}^{3n_{q}})}.
\]
This induces natural norm on $\mathcal{F}^{\otimes p}$ which will
be denoted by $\|\cdot\|_{\mathcal{F}^{\otimes p}}$.

In general, for every Hilbert-Schmidt operator $J$ defined on the
$\mathbf{I}$-particle sector, we have a representation 
\[
J=\sum_{\mathfrak{k}=1}^{\infty}(\bigotimes_{q=1}^{p}J_{q}^{\mathfrak{k}})=\sum_{\mathfrak{k}=1}^{\infty}(\mathbf{J}^{\mathfrak{k}}).
\]
The second quantization $d\Gamma(J)$ of $J$ is the operator on $\mathcal{F}^{\otimes p}$
whose action on the $\mathbf{n}$-particle sector is given by 
\[
\sum_{\mathfrak{k}=1}^{\infty}\bigotimes_{q=1}^{p}\left(d\Gamma(J_{q}^{\mathfrak{k}})\psi\right)^{(n_{q})}=\sum_{\mathfrak{k}=1}^{\infty}\bigotimes_{q=1}^{p}\left(\sum_{j=1}^{n_{q}}J_{q,j}^{\mathfrak{k}}\psi_{q}^{(n_{q})}\right)
\]
where $J_{qj}^{\mathfrak{n}}=1\otimes1\otimes\cdots\otimes1\otimes J_{q}^{\mathfrak{n}}\otimes1\otimes\cdots$
are the operators acting only on the $j$-th particle of $q$-th component.
If $J$ has a kernel $J(\mathbf{x};\mathbf{x}')$, then $d\Gamma(J)$
can be expressed as $\sum_{\mathfrak{k}=1}^{\infty}\bigotimes_{q=1}^{p}d\Gamma(J_{q}^{\mathfrak{k}})$
due to the following:

Since $J$ is a bounded linear operator, it can be written in the
tensor product form such that
\[
J(\mathbf{x};\mathbf{x}')=\sum_{\mathfrak{k}=1}^{\infty}\bigotimes_{q=1}^{p}J_{q}^{\mathfrak{k}}(x_{q};x_{q}').
\]
Then 
\[
d\Gamma(J)=\int\mathrm{d}\mathbf{x}\mathrm{d}\mathbf{x}'\ J(\mathbf{x};\mathbf{x}')\prod_{q=1}^{p}\mathcal{E}^{(q)}a_{x_{q}}^{*}a_{x_{q}'}
\]
which also can be expressed
\[
d\Gamma(\sum_{\mathfrak{k}=1}^{\infty}\mathbf{J}^{\mathfrak{k}})=\int\mathrm{d}\mathbf{x}\mathrm{d}\mathbf{x}'\ \sum_{\mathfrak{n}=1}^{\infty}\mathbf{J}^{\mathfrak{k}}(\mathbf{x};\mathbf{x}')\boldsymbol{a}_{\mathbf{x}}^{*}\boldsymbol{a}_{\mathbf{x}'}
\]
for each $q=1,\dots,p$. We have
\begin{align*}
d\Gamma(\bigotimes_{q=1}^{p}J_{q}) & =\int\mathrm{d}\mathbf{x}\mathrm{d}\mathbf{x}'\ \sum_{\mathfrak{k}=1}^{\infty}\mathbf{J}^{\mathfrak{k}}(\mathbf{x};\mathbf{x}')\boldsymbol{a}_{\mathbf{x}}^{*}\boldsymbol{a}_{\mathbf{x}'}.\\
 & =\sum_{\mathfrak{k}=1}^{\infty}\int\mathrm{d}\mathbf{x}\mathrm{d}\mathbf{x}'\ \mathbf{J}^{\mathfrak{k}}(\mathbf{x};\mathbf{x}')\boldsymbol{a}_{\mathbf{x}}^{*}\boldsymbol{a}_{\mathbf{x}'}\\
 & =\sum_{\mathfrak{k}=1}^{\infty}\bigotimes_{q=1}^{p}\int\mathrm{d}x_{q}\mathrm{d}x_{q}'\ J_{q}^{\mathfrak{k}}(x_{q};x_{q}')a_{x}^{*}a_{x'}\\
 & =\sum_{\mathfrak{k}=1}^{\infty}\bigotimes_{q=1}^{p}d\Gamma(J_{q}^{\mathfrak{k}}).
\end{align*}

The creation and annihilation operators for $q$-th component are
bounded with respect to $\mathcal{E}^{(q)}\mathcal{N}^{1/2}$. The
precise statement is given in the following lemma. The proof of the
lemma follows from the corresponding well-known result for creation
and annihilation operators on each factor of $\mathcal{F}\otimes\mathcal{F}$
(see \cite{Rodnianski2009} for a proof).
\begin{lem}
\label{lem:relbN} Let $\mathbf{f}:=\bigotimes_{q=1}^{p}f_{q}\in L^{2}(\mathbb{R}^{3})^{\otimes p}$.
For any $\psi\in\mathcal{F}^{\otimes p}$, we have 
\begin{alignat*}{1}
\|a^{(q)*}(f_{q})\psi\|_{\mathcal{F}^{\otimes p}} & \le\|f_{q}\|_{2}\,\|\mathcal{E}^{(q)}(\mathcal{N}+1)^{1/2}\psi\|_{\mathcal{F}^{\otimes p}}\quad\text{and}\\
\|a^{(q)}(f_{q})\psi\|_{\mathcal{F}^{\otimes p}} & \le\|f_{q}\|_{2}\,\|\mathcal{E}^{(q)}\mathcal{N}^{1/2}\psi\|_{\mathcal{F}^{\otimes p}}.
\end{alignat*}
\end{lem}

\begin{lem}
\label{lem:bound-a-a*-Phi-dGamma}For $\alpha_{q}>0$, let $D\left(\mathcal{N}^{\mathbf{\alpha}}\right)=\{\mathcal{N}^{\mathbf{\alpha}}\psi\in\mathcal{F}^{\otimes p}:\sum_{n_{q}\geq1}\prod_{q=1}^{p}n_{q}^{2\alpha_{q}}\|\psi_{q}^{(n_{q})}\|_{2}^{2}<\infty\}$
denote the domain of the operator $\mathcal{N}^{\mathbf{\alpha}}$.
For any $\mathbf{f}=\bigotimes_{q=1}^{p}f_{q}\in\bigotimes_{q=1}^{p}L^{2}(\mathbb{R}^{3},\mathrm{d}x_{q})$
and any $\psi\in D(\bigotimes_{q=1}^{p}\mathcal{N}^{1/2})$, we have
\[
\begin{split}\|\boldsymbol{a}(\mathbf{f})\psi\|_{\mathcal{F}^{\otimes p}} & \leq(\prod_{q=1}^{p}\|f_{q}\|_{2})\,\|(\bigotimes_{q=1}^{p}\mathcal{N}^{1/2})\psi\|_{\mathcal{F}^{\otimes p}},\\
\|\boldsymbol{a}^{*}(\mathbf{f})\psi\|_{\mathcal{F}^{\otimes p}} & \leq(\prod_{q=1}^{p}\|f_{q}\|_{2})\,\|(\bigotimes_{q=1}^{p}(\mathcal{N}+1)^{1/2})\psi\|_{\mathcal{F}^{\otimes p}},
\end{split}
\]
and
\begin{equation}
\begin{split}\|\Phi(\mathbf{f})\psi\|_{\mathcal{F}^{\otimes p}} & \leq2(\prod_{q=1}^{p}\|f_{q}\|_{2})\|(\bigotimes_{q=1}^{p}(\mathcal{N}+1)^{1/2})\psi\|_{\mathcal{F}^{\otimes p}}\\
 & \,\leq2(\prod_{q=1}^{p}\|f_{q}\|_{2})\|(\mathcal{N}_{\mathrm{tot}}+1)^{p/2}\psi\|_{\mathcal{F}^{\otimes p}}.
\end{split}
\label{eq:bd-a-1}
\end{equation}
Moreover, for any Hilbert--Schmidt $p$-particle operator $J$ on
$L^{2}(\mathbb{R}^{3},dx)^{\otimes p}$ and for every $\psi\in D((\mathcal{N}^{1/2}\otimes\mathcal{N}^{1/2}))$,
we find 
\begin{equation}
\|d\Gamma(J)\psi\|\leq\|J\|_{\mathrm{HS}}\|\mathcal{N}^{\otimes p}\psi\|\,.\label{eq:J-bd-1}
\end{equation}
\end{lem}

\begin{proof}
These bounds are standard. A proof of (\ref{eq:bd-a-1}) can be obtained
using Lemma (\ref{lem:relbN}). As for (\ref{eq:J-bd-1}), it is to
observe that
\begin{align*}
\|d\Gamma(J)\psi\|_{\mathcal{F}^{\otimes p}}^{2} & =\sum_{\mathfrak{k}=1}^{\infty}\prod_{q=1}^{p}\langle d\Gamma(J_{q}^{\mathfrak{k}})\psi_{q},d\Gamma(J_{q}^{\mathfrak{k}})\psi_{q}\rangle_{\mathcal{F}}\\
 & =\sum_{\mathfrak{k}=1}^{\infty}\prod_{q=1}^{p}\sum_{i_{q},j_{q}=1}^{n_{q}}\langle J_{q,i_{q}}^{\mathfrak{k}}\psi_{q}^{(n_{q})},J_{q,j_{q}}^{\mathfrak{k}}\psi_{q}^{(n_{q})}\rangle\\
 & \leq\sum_{\mathfrak{k}=1}^{\infty}\prod_{q=1}^{p}\sum_{n_{q}\geq1}n_{q}^{2}\|J_{q}^{\mathfrak{k}}\|_{\mathrm{HS}}^{2}\|\psi_{q}^{(n_{q})}\|_{2}^{2}.
\end{align*}
By the definition of $\mathcal{N}$ on $\mathcal{F}$,
\begin{align*}
\|d\Gamma(J)\psi\|_{\mathcal{F}^{\otimes p}}^{2} & \leq\sum_{\mathfrak{k}=1}^{\infty}\prod_{q=1}^{p}\|J_{q}^{\mathfrak{k}}\|_{\mathrm{HS}}^{2}\|\mathcal{E}^{(q)}\mathcal{N}\psi\|_{\mathcal{F}^{\otimes p}}^{2}\\
 & \leq\sum_{\mathfrak{k}=1}^{\infty}\|\mathbf{J}^{\mathfrak{k}}\|_{\mathrm{HS}}^{2}\|\mathcal{N}^{\otimes p}\psi\|_{\mathcal{F}^{\otimes p}}^{2}.\\
 & =\|J\|_{\mathrm{HS}}^{2}\|\mathcal{N}^{\otimes p}\psi\|_{\mathcal{F}^{\otimes p}}^{2}
\end{align*}
which completes the proof.
\end{proof}
For $Q\subset\{1,\dots,p\}$, let us denote 
\[
d\Gamma^{(Q)}(J):=\int\mathrm{d}\mathbf{x}\mathrm{d}\mathbf{x}'\ J(\mathbf{x};\mathbf{x}')\prod_{q\in Q}\mathcal{E}^{(q)}a_{x_{q}}^{*}a_{x_{q}'}\,.
\]
We let $J\mathbf{u}_{t}=(J_{1}u_{1,t},\dots,J_{p}u_{p,t})$ and $J^{(Q)}=\prod_{q\in Q}\mathcal{E}^{(q)}J^{q}$.
Then using the same idea proving Lemma \ref{lem:bound-a-a*-Phi-dGamma},
we have
\begin{lem}
\label{lem:dGammaQphiQ}Let $p$ be a positive integer and let $Q$
be a subset of positive integer such that $Q\subset\{1,\dots,p\}$.
Then for $\psi\in\mathcal{F}^{\otimes p}$ and for any nonnegative
integer $j$, we have
\[
\left\Vert (\mathcal{N}_{\mathrm{total}}+p\mathcal{I})^{j}d\Gamma^{(Q)}(J)\psi\right\Vert _{\mathcal{F}^{\otimes p}}\leq\|J^{(Q)}\|_{\mathrm{HS}}\left\Vert (\mathcal{N}_{\mathrm{total}}+p\mathcal{I})^{j+|Q|}\psi\right\Vert _{\mathcal{F}^{\otimes p}}
\]
and
\[
\left\Vert (\mathcal{N}_{\mathrm{total}}+p\mathcal{I})^{j}\boldsymbol{\phi}^{(Q)}(J\mathbf{u}_{t})\psi\right\Vert _{\mathcal{F}^{\otimes p}}\leq\|J^{(Q)}\|_{\mathrm{HS}}\left\Vert (\mathcal{N}_{\mathrm{total}}+p\mathcal{I})^{j+|Q|}\psi\right\Vert _{\mathcal{F}^{\otimes p}}.
\]
\end{lem}

\textbf{The Hamiltonian.} We now define the extended Hamiltonian $\mathcal{H}_{\mathbf{N}}$
acting on $\mathcal{F}^{\otimes p}$. First, we observe that the set
of all finite sums $\sum_{\mathbf{n}\in\mathbb{Z}_{\geq0}^{p}}\psi_{q}^{(\mathbf{n})}$
is dense in $\mathcal{F}^{\otimes p}$. Note that the set of all finite
linear combinations of products $\psi^{(\mathbf{n})}=\bigotimes_{q=1}^{p}\psi_{q}^{(n_{q})}$,
denoted by $\mathcal{D}$, is dense in $\mathcal{F}^{\otimes p}$.
We define the Hamiltonian $\mathcal{H}_{\mathbf{N}}$ acting on products
$\psi^{(\mathbf{n})}$ by $\mathcal{H}_{\mathbf{N}}\,\psi^{(\mathbf{n})}=\mathcal{H}_{\mathbf{N}}^{\mathbf{n}}\,\psi^{(\mathbf{n})}$
where 
\[
\mathcal{H}_{\mathbf{N}}^{\mathbf{n}}=\sum_{q=1}^{p}\mathcal{E}^{(q)}h_{q}^{(n_{q})}+\sum_{q\neq r}V_{qr}^{(n_{q},n_{r})}
\]
with 
\[
h_{q}^{(n_{q})}=\sum_{j=1}^{n_{q}}-\Delta_{x_{j}}+\frac{1}{N_{q}}\sum_{1\le i<j\le n_{q}}V_{qq}(x_{i}-x_{j})
\]
for $q=1,\dots,p$, and 
\[
V_{qr}^{(n_{q},n_{r})}=\frac{1}{N}\sum_{i=1}^{n_{q}}\sum_{j=1}^{n_{r}}V_{qr}(x_{i}-y_{j})
\]
for $q,r=1,\dots,p$.

We then extend $\mathcal{H}_{\mathbf{N}}$ to $\mathcal{D}$ by linearity.
With the assumption/hypothesis on $V_{qr}$ in Theorem \ref{thm:Main_Theorem},
the Hamiltonian $\mathcal{H}_{\mathbf{N}}$ on $\mathcal{D}$ gives
rise to a self-adjoint operator \cite{ReedSimonII1975}.

Using the operator-valued distributions $a_{x}$ and $a_{x}^{*}$,
the Hamiltonian $\mathcal{H}_{\mathbf{N}}$ can be written as 
\[
\mathcal{H}_{\mathbf{N}}=\sum_{q=1}^{p}\mathcal{H}_{N_{q}}+\sum_{q\neq r}\mathcal{V}_{qr}
\]
where 
\[
\mathcal{H}_{q}=\int\mathrm{d}x\,\nabla_{x}a_{x}^{(q)*}\nabla_{x}a_{x}^{(q)}+\frac{1}{N_{q}}\int\int\mathrm{d}x\mathrm{d}z\,V_{qq}(x-z)a_{x}^{(q)*}a_{z}^{(q)*}a_{z}^{(q)}a_{x}^{(q)},
\]
and 
\[
\mathcal{V}_{qr}=\frac{1}{N}\int\int\mathrm{d}x\mathrm{d}y\,V_{qr}(x-y)a_{x}^{(q)*}a_{y}^{(r)*}a_{y}^{(r)}a_{x}^{(q)}.
\]

The Hamiltonian $\mathcal{H}_{\mathbf{N}}$ conserves the number of
particles in each factor of $\mathcal{F}^{\otimes p}$. In fact, it
is simple to verify that, for $q=1,\dots,p$, 
\[
[\mathcal{H}_{q},\mathcal{E}^{(q)}\mathcal{N}]=[\mathcal{H}_{q},\mathcal{E}^{(r)}\mathcal{N}]=[\mathcal{V}_{qr},\mathcal{E}^{(q)}\mathcal{N}]=[\mathcal{V}_{qr},\mathcal{E}^{(r)}\mathcal{N}]=0.
\]
Moreover we have
\[
[\mathcal{H}_{q},\mathcal{N}^{\otimes p}]=[\mathcal{V}_{qr},\mathcal{N}^{\otimes p}]=0.
\]
Furthermore, for fixed $\mathbf{N}=(N_{1},\dots,N_{p})$, the subspace
\[
\mathcal{S}_{\mathbf{N}}=\text{span}\{\psi_{q}^{(\mathbf{N})}:=\bigotimes_{q=1}^{p}\psi_{q}^{(N_{q})}\,|\,\psi_{q}^{(N_{q})}\in L_{s}^{2}(\mathbb{R}^{3N_{q}})\text{ for }q=1,\dots,p\}\subset\mathcal{F}^{\otimes p}
\]
is invariant by $\mathcal{H}_{\mathbf{N}}$, and the Hamiltonian $\mathcal{H}_{\mathbf{N}}$
restricted to $\mathcal{S}_{\mathbf{N}}$ is equal to $H_{\mathbf{N}}$.
Therefore, for initial data in $\mathcal{S}_{\mathbf{N}}$, the time
evolution generated by $\mathcal{H}_{\mathbf{N}}$ reduces to the
time evolution generated by $H_{\mathbf{N}}$.

\textbf{The reduced density operator.} For $\psi\in\mathcal{F}^{\otimes p}$,
we define the reduced density operator $\gamma_{\psi}^{\mathbf{I}}$
as the operator on $L^{2}(\mathbb{R}^{3})^{\otimes p}$ determined
by the kernel 
\begin{equation}
\gamma_{\psi}^{\mathbf{I}}(\mathbf{x};\mathbf{x}')=\frac{1}{\langle\psi,\mathcal{N}^{\otimes p}\psi\rangle}\langle\psi,\prod_{q=1}^{p}a_{x_{q}'}^{(q)*}a_{x_{q}}^{(q)}\psi\rangle.\label{eq:intKer}
\end{equation}
We observe that $\text{Tr}_{\mathbf{I}}\gamma_{\psi}^{\mathbf{I}}=1$.
If $\psi$ is in the subspace $\mathcal{S}_{\mathbf{N}}$ of fixed
number of particles, the above definition is reduced to the definition
of $\gamma_{\mathbf{N},t}^{\mathbf{I}}$ given in (\ref{eq:double_trace-KERNEL}).

For $f_{q}\in L^{2}(\mathbb{R}^{3})$, we define 
\[
\mathcal{W}(\mathbf{f})=\bigotimes_{q=1}^{p}W(f_{q}),
\]
which is an operator on $\mathcal{F}^{\otimes p}$. We set a vacuum
state $\omega=\Omega^{\otimes p}$ . Thus $\mathcal{W}(\mathbf{f})\,\omega$
is a tensor product of coherent states (or simply coherent state).

In the following lemma, we have some important properties of the operator
$\mathcal{W}(\mathbf{f})$ and the coherent state $\mathcal{C}(\mathbf{f})=\mathcal{W}(\mathbf{f})\,\omega$.
These properties follow easily from the corresponding well-known properties
of Weyl operators (see \cite{Rodnianski2009}, for example). The coherent
state $\mathcal{C}(\mathbf{f})$ can also be written in terms of the
Weyl operator as 
\begin{equation}
\mathcal{C}(\mathbf{f})=\mathcal{W}(\mathbf{f})\,\omega=\prod_{q=1}^{p}e^{-(\|f_{q}\|_{2}^{2})/2}\exp(a^{*}(f_{q}))\omega=e^{-(\sum_{q=1}^{p}\|f_{q}\|_{2}^{2})/2}\bigotimes_{q=1}^{p}\left(\sum_{n\geq0}\frac{1}{\sqrt{n!}}f_{q}^{\otimes n}\right).\label{Weyl_f_mixed}
\end{equation}
We collect the useful properties of the Weyl operator and the coherent
states in the following lemma.
\begin{lem}
\label{lem:weyl} Let $\omega=\Omega^{\otimes p}$ and $f_{q}\in L^{2}(\mathbb{R}^{3})$
for all $q=1,\dots,p$.
\begin{enumerate}
\item The operator $\mathcal{W}(\mathbf{f})$ is unitary and 
\[
\mathcal{W}(\mathbf{f})^{*}=\mathcal{W}(\mathbf{f})^{-1}=\mathcal{W}(-\mathbf{f}).
\]
\item We have 
\begin{align*}
\mathcal{W}(\mathbf{f})^{*}a_{x}^{(q)}\mathcal{W}(\mathbf{f}) & =\mathcal{E}^{(q)}\left(a_{x}+f_{q}(x)\right), & \qquad\mathcal{W}(\mathbf{f})^{*}a_{x}^{(q)*}\mathcal{W}(\mathbf{f}) & =\mathcal{E}^{(q)}\left(a_{x}^{*}+\overline{f(x)}\right).
\end{align*}
\item For $\bm{\alpha}=(\alpha_{1},\dots,\alpha_{p})\in\left\{ 0,1\right\} ^{p}$,
we have 
\[
\langle\mathcal{W}(\mathbf{f})\,\omega,\mathcal{N}^{\bm{\alpha}}\mathcal{W}(\mathbf{f})\,\omega\rangle=\prod_{q}\|f_{q}\|^{2\alpha_{q}}=:\|\mathbf{f}\|^{2\bm{\alpha}}.
\]
\end{enumerate}
\end{lem}

\textbf{Schrödinger dynamics.} We introduce Hardamard product $\odot:\mathbb{R}^{d}\times\mathbb{R}^{d}\to\mathbb{R}^{d}$
such that for $\mathbf{u},\mathbf{v}\in\mathbb{R}^{d}$ with $\mathbf{u}=(u_{1},\dots,u_{d})$
and $\mathbf{v}=(v_{1},\dots,v_{d})$, the Hardamard product of $\mathbf{u}$
and $\mathbf{v}$ is given by
\[
\mathbf{u}\odot\mathbf{v}=(u_{1}v_{1},\dots,u_{d}v_{d}).
\]
We denote for $r\in\mathbb{R}$
\[
\mathbf{u}^{\odot r}=(u_{1}^{r},\dots,u_{d}^{r})
\]
and for $\mathbf{r}=(r_{1},\dots,r_{d})\in\mathbb{R}^{d}$
\[
\mathbf{u}^{\odot\mathbf{r}}=(u_{1}^{r_{1}},\dots,u_{d}^{r_{d}}).
\]
For $u_{q}\in H^{1}(\mathbb{R}^{3})$ with $\|u_{q}\|_{L^{2}}=1$
, we have
\[
\mathcal{C}_{\mathbf{N}}(\mathbf{u})=\mathcal{W}(\mathbf{N}^{\odot1/2}\odot\mathbf{u})\,\omega
\]
and denote by $t\mapsto\mathbf{u}_{t}=e^{-\mathrm{i}t\mathcal{H}_{\mathbf{N}}}\mathbf{u}$
the solution to the Schrödinger equation $\mathrm{i}\partial_{t}\mathbf{u}_{t}=\mathcal{H}_{\mathbf{N}}\mathbf{u}_{t}$
with initial condition $\mathbf{u}_{t=0}=\mathbf{u}$. We will study
the family of solutions $\{\mathbf{u}_{t}=(u_{1,t},\dots,u_{p,t})\}$
as $N_{q}$ go to infinity.

Finally, we collect lemmas on the Weyl operator acting on a state
with fixed number of particles, which will be used in Section \ref{sec:Lemmas}.
Define 
\begin{equation}
d_{N}:=\frac{\sqrt{N!}}{N^{N/2}e^{-N/2}}.\label{eq:d_N}
\end{equation}
We note that $C^{-1}N^{1/4}\leq d_{N}\leq CN^{1/4}$ for some constant
$C>0$ independent of $N$, which can be easily checked by using Stirling's
formula. We denote
\[
\boldsymbol{d}_{\boldsymbol{N}}:=\prod_{q=1}^{p}\frac{\sqrt{N_{q}!}}{N_{q}^{N_{q}/2}e^{-N_{q}/2}}
\]
and 
\[
\boldsymbol{d}_{Q}:=\prod_{q\in Q}\frac{\sqrt{N_{q}!}}{N_{q}^{N_{q}/2}e^{-N_{q}/2}}.
\]

\section{Proof of the Main Result\label{sec:Proof-main}}

We now turn to the proof of our main theorem (Theorem \ref{thm:Main_Theorem}).
We first introduce unitary operators and their generators to give
the proof.
\begin{rem}
\label{rem:HSnorm-opnorm}Similar to the argument given in \cite[Remark 1.4]{Rodnianski2009},
we will prove the theorem not for the trace norm but for the Hilbert-Schmidt
norm, i.e.,
\[
\|\gamma_{\mathbf{N},t}^{\mathbf{I}}-|\mathbf{u}_{t}\rangle\langle\mathbf{u}_{t}|\,\|_{\mathrm{HS}}\le Ce^{Kt}\cdot\frac{p^{Cp}}{N}.
\]
Note that $\operatorname{Tr}\gamma_{\mathbf{N},t}^{\mathbf{I}}=1$,
all the eigenvalues of $\gamma_{\mathbf{N},t}^{\mathbf{I}}$ are non-negative,
and $|\mathbf{u}_{t}\rangle\langle\mathbf{u}_{t}|$ is a rank $p$
projection. Then the operator $\gamma_{\mathbf{N},t}^{\mathbf{I}}-|\mathbf{u}_{t}\rangle\langle\mathbf{u}_{t}|$
can have at most $p$ negative eigenvalues so that the set of negative
eigenvalues $\Lambda_{\mathrm{neg}}$ have at most $p$ elements.
Using
\[
0=\operatorname{Tr}\left(\gamma_{\mathbf{N},t}^{\mathbf{I}}-|\mathbf{u}_{t}\rangle\langle\mathbf{u}_{t}|\right)=\sum_{\lambda_{q}^{-}\not\in\Lambda_{\mathrm{neg}}}|\lambda_{q}^{+}|-\sum_{\lambda_{q}^{-}\in\Lambda_{\mathrm{neg}}}|\lambda_{q}^{-}|,
\]
we have
\[
\sum_{\lambda_{q}^{-}\not\in\Lambda_{\mathrm{neg}}}|\lambda_{q}^{+}|=\sum_{\lambda_{q}^{-}\in\Lambda_{\mathrm{neg}}}|\lambda_{q}^{-}|.
\]
Thus one may see that
\begin{align*}
\operatorname{Tr}\left|\gamma_{\mathbf{N},t}^{\mathbf{I}}-|\mathbf{u}_{t}\rangle\langle\mathbf{u}_{t}|\right| & =\sum_{\lambda_{q}^{+}\not\in\Lambda_{\mathrm{neg}}}|\lambda_{q}^{+}|+\sum_{\lambda_{q}^{-}\in\Lambda_{\mathrm{neg}}}|\lambda_{q}^{-}|=2\sum_{\lambda_{q}^{-}\in\Lambda_{\mathrm{neg}}}|\lambda_{q}^{-}|.
\end{align*}
Using Cauchy-Schwarz inequality,
\begin{align*}
\operatorname{Tr}\left|\gamma_{\mathbf{N},t}^{\mathbf{I}}-|\mathbf{u}_{t}\rangle\langle\mathbf{u}_{t}|\right| & =\sum_{\lambda_{q}^{+}\in\Lambda_{\mathrm{pos}}}|\lambda_{q}^{+}|+\sum_{\lambda_{q}^{-}\in\Lambda_{\mathrm{neg}}}|\lambda_{q}^{-}|\leq|\Lambda_{\mathrm{pos}}|^{1/2}(\sum_{\lambda_{q}^{+}\in\Lambda_{\mathrm{pos}}}|\lambda_{q}^{-}|^{2})^{1/2}+|\Lambda_{\mathrm{neg}}|^{1/2}(\sum_{\lambda_{q}^{-}\in\Lambda_{\mathrm{neg}}}|\lambda_{q}^{-}|^{2})^{1/2}\\
 & \leq(|\Lambda_{\mathrm{pos}}|^{1/2}+|\Lambda_{\mathrm{neg}}|^{1/2})\left\Vert \gamma_{\mathbf{N},t}^{\mathbf{I}}-|\mathbf{u}_{t}\rangle\langle\mathbf{u}_{t}|\right\Vert _{\mathrm{HS}}\leq\sqrt{2p}\left\Vert \gamma_{\mathbf{N},t}^{\mathbf{I}}-|\mathbf{u}_{t}\rangle\langle\mathbf{u}_{t}|\right\Vert _{\mathrm{HS}}.
\end{align*}
\end{rem}

\subsection{Unitary operators and their generators}

\textbf{Fluctuation dynamics.} We define $\alpha(t)=\int_{0}^{t}\mathrm{d}s\,\varphi_{\mathbf{N}}(s)$,
where $\varphi_{\mathbf{N}}$ is a real-valued function that we will
choose later. Let $\mathbf{u}_{t}=(u_{1,t},\dots,u_{p,t})$ be the
solutions to the equations of Hartree type (\ref{sysHartree}). For
$t,s\in\mathbb{R}$, we set 
\[
\mathcal{U}(t,s)=e^{\mathrm{i}(\alpha(t)-\alpha(s))}\mathcal{W}(\mathbf{N}^{\odot1/2}\odot\mathbf{u}_{t})^{*}e^{-\mathrm{i}(t-s)\mathcal{H}_{\mathbf{N}}}\mathcal{W}(\mathbf{N}^{\odot1/2}\mathbf{u}_{s}).
\]
We refer to the operator $\mathcal{U}(t,s)$ as the fluctuation dynamics.
We abbreviate 
\[
\mathcal{W}_{t}=\mathcal{W}(\mathbf{N}^{\odot1/2}\mathbf{u}_{t})\qquad\text{and}\qquad\mathcal{U}_{t,s}=\mathcal{U}(t,s).
\]
Thus, we may write $\mathcal{U}_{t,s}=e^{\mathrm{i}(\alpha(t)-\alpha(s))}\mathcal{W}_{t}e^{-\mathrm{i}(t-s)\mathcal{H}_{\mathbf{N}}}\mathcal{W}_{s}$.
We also define 
\[
\omega_{t}=e^{-\mathrm{i}\alpha(t)}\mathcal{U}(t,0)\omega.
\]
Using the above definitions, it is simple to verify that 
\[
\mathcal{C}_{t}=\mathcal{W}_{t}\,\omega_{t}.
\]

For each $s\in\mathbb{R}$, the fluctuation dynamics satisfies the
equation 
\[
\mathrm{i}\partial_{t}\mathcal{U}_{t,s}=\mathcal{L}(t)\mathcal{U}_{t,s}\qquad\text{with}\qquad\mathcal{U}_{s,s}=I
\]
where
\[
\mathcal{W}_{t}^{*}\mathcal{H}_{\mathbf{N}}\mathcal{W}_{t}=:\mathcal{L}_{0}(t)+\mathcal{L}_{1}(t)+\mathcal{L}_{2}(t)+\mathcal{L}_{3}(t)+\mathcal{L}_{4}(t)
\]
with
\[
\begin{split}\mathcal{L}_{0}(t) & :=\sum_{q,r=1}^{p}\sqrt{N_{q}}\left(\frac{N_{r}}{N}-c_{qr}\right)\phi^{(q)}\big((V_{qr}*|u_{rt}|^{2})u_{qt},\\
\mathcal{L}_{1}(t) & :=\sum_{q,r=1}^{p}\frac{\sqrt{N_{q}N_{r}}}{N}\int\mathrm{d}x_{q}\mathrm{d}x_{r}\,V_{qr}(x_{q}-x_{r})(\overline{u_{r,t}}(x_{r})u_{q,t}(x_{q})a_{x_{q}}^{(q)*}a_{x_{r}}^{(r)}+\overline{u_{q,t}}(x_{q})u_{r,t}(x_{r})a_{x_{r}}^{(r)*}a_{x_{q}}^{(q)}),\\
\mathcal{L}_{2}(t) & :=\sum_{q=1}^{p}\Big(\int\mathrm{d}x\,\nabla_{x}a_{x}^{(q)*}\nabla_{x}a_{x}^{(q)}+\int\mathrm{d}x\,(V_{qq}*|u_{t}|^{2})(x)a_{x}^{(q)*}a_{x}^{(q)}+\int\mathrm{d}x\mathrm{d}z\,V_{qq}(x-z)\overline{u_{t}}(z)u_{t}(x)a_{x}^{(q)*}a_{z}^{(q)}\\
 & \quad+\frac{1}{2}\int\mathrm{d}x\mathrm{d}z\,V_{qq}(x-z)(u_{q,t}(z)u_{q,t}(x)a_{x}^{(q)*}a_{z}^{(q)*}+\overline{u_{q,t}}(z)\overline{u_{q,t}}(x)a_{x}^{(q)}a_{z}^{(q)})\Big)\\
 & \quad+\sum_{q,r=1}^{p}\frac{N_{q}}{N}\int\mathrm{d}x\,(V_{qr}*|u_{r,t}|^{2})(x)a_{x}^{(q)*}a_{x}^{(q)}\\
 & \quad=:\sum_{q=1}^{p}\mathcal{L}_{2}^{q}(t)+\mathcal{L}_{2}^{\text{cross}}(t),\\
\mathcal{L}_{3}(t) & =\sum_{q=1}^{p}\Big(\frac{1}{\sqrt{N_{q}}}\int\mathrm{d}x\mathrm{d}z\,V_{qq}(x-z)a_{x}^{(q)*}(u_{q,t}(z)a_{z}^{(q)*}+\overline{u_{q,t}}(z)a_{z}^{(q)})a_{x}^{(q)}\Big)\\
 & \quad+\sum_{q,r=1}^{p}\Big(\frac{\sqrt{N_{q}}}{N}\int\mathrm{d}x\mathrm{d}y\,V_{qr}(x-y)a_{y}^{(r)*}(u_{q,t}(x)a_{x}^{(q)*}+\overline{u_{q,t}}(x)a_{x}^{(q)})a_{y}^{(r)}\Big)\\
 & \quad=:\sum_{q=1}^{p}\mathcal{L}_{3}^{q}(t)+\mathcal{L}_{3}^{\text{cross}}(t),
\end{split}
\]
and
\[
\begin{split}\mathcal{L}_{4}(t) & =\sum_{q=1}^{p}\frac{1}{N_{q}}\int\int\mathrm{d}x\mathrm{d}z\,V_{qq}(x-z)a_{x}^{(q)*}a_{z}^{(q)*}a_{z}^{(q)}a_{x}^{(q)}\\
 & =\sum_{q,r=1}^{p}\frac{1}{N}\int\int\mathrm{d}x\mathrm{d}y\,V_{qr}(x-y)a_{x}^{(q)*}a_{y}^{(r)*}a_{y}^{(r)}a_{x}^{(q)}\\
 & \quad=:\sum_{q=1}^{p}\mathcal{L}_{4}^{q}(t)+\mathcal{L}_{4}^{\text{cross}}(t).
\end{split}
\]

Then 
\begin{equation}
\mathrm{i}\partial_{t}\mathcal{U}\left(t;s\right)=\left(\mathcal{L}_{0}(t)+\mathcal{L}_{2}(t)+\mathcal{L}_{3}(t)+\mathcal{L}_{4}(t)\right)\mathcal{U}\left(t;s\right)\quad\text{and}\quad\mathcal{U}\left(s;s\right)=I\label{eq:def_mathcalU}
\end{equation}
and

\begin{align*}
 & \mathcal{W}^{*}(\mathbf{N}^{\odot1/2}\odot\mathbf{u}_{s})e^{\mathrm{i}\mathcal{H}_{\mathbf{N}}\left(t-s\right)}\left(\prod_{q=1}^{p}\left(a_{x}^{(q)}-\sqrt{N_{q}}u_{q,t}\left(x\right)\right)\right)e^{-\mathrm{i}\mathcal{H}_{\mathbf{N}}\left(t-s\right)}\mathcal{W}(\mathbf{N}^{\odot1/2}\odot\mathbf{u}_{s})\\
 & =\mathcal{U}^{*}\left(t;s\right)\boldsymbol{a}_{x}\mathcal{U}\left(t;s\right).
\end{align*}

\begin{rem}
\label{rem:|NrN-Cqr|}Note that $\mathcal{L}_{0}(t)=0$ because we
define $c_{qr}=N_{r}/N$. Here, however, we put generalized definition
of it to leave a little room for changing condition on $c_{qr}$.
For example, one can assume such that
\[
\left|\frac{N_{r}}{N}-c_{qr}\right|\leq\frac{C}{N}
\]
for some $C$ for all $q,r=1,\dots,p$.

\textcolor{black}{Let $u_{q,t}$ be the solution of the equations
of Hartree type (\ref{sysHartree}). If we perturb $c_{qr}$ by $O(N^{-1})$
and denote it $\widetilde{c}_{qr}$, and let $\widetilde{u}_{q,t}$
be the solution of the perturbed equations of Hartree type.}

\begin{align*}
\frac{\mathrm{d}}{\mathrm{d}t}\|u_{q,t}-\widetilde{u}_{q,t}\|_{2}^{2} & =2\operatorname{Im}\sum_{r=1}^{p}\langle u_{q,s},[c_{qr}(V_{qr}*|u_{q,s}|^{2})-\widetilde{c}_{qr}(V_{qr}*|\widetilde{u}_{q,s}|^{2})]\widetilde{u}_{q,s}\rangle\\
 & =2\operatorname{Im}\sum_{r=1}^{p}\langle u_{q,s},[(c_{qr}-\widetilde{c}_{qr})(V_{qr}*|u_{q,s}|^{2})]\widetilde{u}_{q,s}\rangle\\
 & \qquad+2\operatorname{Im}\sum_{r=1}^{p}\langle u_{q,s},[\widetilde{c}_{qr}V_{qr}*(|u_{q,s}|^{2}-|\widetilde{u}_{q,s}|^{2}](\widetilde{u}_{q,s}-u_{q,s})\rangle
\end{align*}
where, the last line we used that
\[
\operatorname{Im}\sum_{r=1}^{p}\langle u_{q,s},[\widetilde{c}_{qr}V_{qr}*(|u_{q,s}|^{2}-|\widetilde{u}_{q,s}|^{2}]u_{q,s}\rangle=0.
\]
Then, \textcolor{black}{using Cauchy-Schwarz inequality, our assumption
about $V_{qr}$ in (\ref{eq:assumption_V}), and the fact that both
$u_{q,t}$ and $\widetilde{u}_{q,t}$ is bounded $H^{1}$-norm independent
to time $t$,}
\begin{align*}
\left|\frac{\mathrm{d}}{\mathrm{d}t}\|u_{q,t}-\widetilde{u}_{q,t}\|_{2}^{2}\right| & \leq\frac{Cp}{N^{2/p}}+Cp\|u_{q,t}-\widetilde{u}_{q,t}\|_{2}^{2}
\end{align*}
\textcolor{black}{This implies from Grönwall inequality that
\[
\|u_{q,t}-\widetilde{u}_{q,t}\|_{2}\leq\frac{Cp}{N^{1/p}}e^{Kpt}.
\]
Thus the $L^{2}$-norm difference is of $O(N^{-1})$}. Hence we have
\[
\operatorname{Tr}\,\Big|\,|\mathbf{u}_{t}\rangle\langle\mathbf{u}_{t}|-|\widetilde{\mathbf{u}}_{t}\rangle\langle\widetilde{\mathbf{u}}_{t}|\,\Big|\leq2\sqrt{p}\left\Vert \mathbf{u}_{t}-\widetilde{\mathbf{u}}_{t}\right\Vert _{\mathcal{H}}\leq2\sqrt{p}\prod_{q=1}^{p}\left\Vert u_{q,t}-\widetilde{u}_{q,t}\right\Vert _{2}\leq\frac{2\sqrt{p}}{N}.
\]
\end{rem}

Let $\widetilde{\mathcal{L}}=\mathcal{L}_{0}(t)+\mathcal{L}_{2}(t)+\mathcal{L}_{4}(t)$
and define the unitary operator $\widetilde{\mathcal{U}}\left(t;s\right)$
by 
\begin{equation}
\mathrm{i}\partial_{t}\widetilde{\mathcal{U}}\left(t;s\right)=\widetilde{\mathcal{L}}\left(t\right)\widetilde{\mathcal{U}}\left(t;s\right)\quad\text{and }\quad\widetilde{\mathcal{U}}\left(s;s\right)=1.\label{eq:def_mathcaltildeU}
\end{equation}
Since $\widetilde{\mathcal{L}}$ does not change the parity of the
number of particles, 
\begin{equation}
\left\langle \Omega,\widetilde{\mathcal{U}}^{*}\left(t;0\right)\boldsymbol{a}_{y}\widetilde{\mathcal{U}}\left(t;0\right)\Omega\right\rangle =\left\langle \Omega,\widetilde{\mathcal{U}}^{*}\left(t;0\right)\boldsymbol{a}_{x}^{*}\widetilde{\mathcal{U}}\left(t;0\right)\Omega\right\rangle =0\label{eq:Parity_Consevation}
\end{equation}
We refer to \cite[Lemma 8.2]{Lee2013} for a rigorous proof of (\ref{eq:Parity_Consevation}).

\subsection{Proof of Theorem \ref{thm:Main_Theorem}}

To have compact notation, we denote the products over indices from
$1$ to $p$ such that
\[
\frac{1}{\boldsymbol{N}}=\prod_{q=1}^{p}\frac{1}{N_{q}},
\]
\[
\frac{\boldsymbol{d}_{\boldsymbol{N}}}{\boldsymbol{N}}=\prod_{q=1}^{p}\frac{d_{N_{q}}}{N_{q}},
\]
and

\[
\frac{\left(\boldsymbol{a}^{*}(\mathbf{u})\right)^{\boldsymbol{N}}}{\sqrt{\boldsymbol{N}!}}=\prod_{q=1}^{p}\frac{a^{(q)*}(u_{q})^{N_{q}}}{\sqrt{N_{q}!}}.
\]
For any nonempty subset $Q$ of positive integers, we let the products
over all indices in the set $Q$ using the following notation
\[
\sqrt{\boldsymbol{N}_{Q}}=\prod_{q\in Q}\sqrt{N_{q}},
\]

\[
d\boldsymbol{\Gamma}^{(Q)}(J)=\prod_{q\in Q}\left(d\Gamma^{(q)}(J)\right),
\]
\[
\boldsymbol{\phi}^{(Q)}(J_{Q}u_{Q,t})=\prod_{q\in Q}\left(\phi^{(q)}(J_{q}u_{q,t})\right),
\]
and
\[
\left\langle \boldsymbol{u}_{Q,t}\right|\boldsymbol{J}_{Q}\left|\boldsymbol{u}_{Q,t}\right\rangle =\prod_{q\in Q}\left\langle u_{q,t}\right|J_{q}\left|u_{q,t}\right\rangle .
\]
If $Q=\emptyset$, we define all the products above as $1$.
\begin{prop}
\label{prop:Et1} Suppose that the assumptions in Theorem \ref{thm:Main_Theorem}
hold. For a Hilbert--Schmidt operator $J$ on $L^{2}(\mathbb{R}^{3})$,
let
\begin{align*}
E_{t}^{Q,R}(J) & :=\left\langle \boldsymbol{u}_{Q^{c}\cap R^{c},t}\right|\boldsymbol{J}_{Q^{c}\cap R^{c}}\left|\boldsymbol{u}_{Q^{c}\cap R^{c},t}\right\rangle \sqrt{\boldsymbol{N}_{Q^{c}}}\sqrt{\boldsymbol{N}_{R^{c}}}\frac{1}{\boldsymbol{N}}\\
 & \qquad\qquad\qquad\times\left\langle \frac{\left(\boldsymbol{a}^{*}(\mathbf{u})\right)^{\boldsymbol{N}}}{\sqrt{\boldsymbol{N}!}}\omega,\mathcal{W}(\mathbf{N}^{\odot1/2}\odot\mathbf{u}_{s})\prod_{\substack{q\in Q\\
r\in R
}
}\left(d\Gamma^{(q)}(J)\cdot\phi^{(r)}(J_{r}u_{r,t})\right)\mathcal{U}(t)\omega\right\rangle _{\mathcal{F}^{\otimes p}},
\end{align*}
Then, there exist constants $C$ and $K$, depending only on $\sup_{|s|\leq t}\|u_{q,s}\|_{H^{1}}$,
such that 
\[
\left|E_{t}^{Q,R}(J)\right|\leq Cp^{Cp}(\min_{q}c_{q})^{-1/2}\|J\|_{\mathrm{HS}}e^{Kt}\frac{1}{N}
\]
for all subsets $Q$ and $R$ of $\{1,\dots,p\}$ such that $|Q|+|R|\geq1$.
\end{prop}

A proof of Propositions \ref{prop:Et1} will be given later in Section
\ref{sec:Pf-of-Props}. With Propositions \ref{prop:Et1}, we now
prove Theorem \ref{thm:Main_Theorem}.
\begin{proof}[Proof of Theorem \ref{thm:Main_Theorem}]
Recall that 
\[
\gamma_{\mathbf{N},t}^{\mathbf{I}}=\frac{1}{\boldsymbol{N}}\left\langle \frac{\left(\boldsymbol{a}^{*}(\mathbf{u})\right)^{\boldsymbol{N}}}{\sqrt{\boldsymbol{N}!}}\omega,e^{i\mathcal{H}_{\mathbf{N}}t}\mathbf{a}_{x}^{*}\mathbf{a}_{x'}e^{-i\mathcal{H}_{\mathbf{N}}t}\frac{\left(\boldsymbol{a}^{*}(\mathbf{u})\right)^{\boldsymbol{N}}}{\sqrt{\boldsymbol{N}!}}\omega\right\rangle _{\mathcal{F}^{\otimes p}}.
\]
From the definition of the creation operator in (\ref{eq:creation}),
we can easily find that 
\begin{equation}
\bigotimes_{q=1}^{p}\{0,0,\dots,0,u_{q}^{\otimes N_{q}},0,\dots\}=\prod_{q=1}^{p}\frac{\left(a^{(q)*}(u_{q})\right)^{N_{q}}}{\sqrt{N_{q}!}}\omega=:\frac{\left(\boldsymbol{a}^{*}(\mathbf{u})\right)^{\boldsymbol{N}}}{\sqrt{\boldsymbol{N}!}}\omega,\label{eq:coherent_vec}
\end{equation}
where the $u_{q}^{\otimes N_{q}}$ on the left-hand side is in the
$N_{q}$-th sector of the Fock space $\mathcal{F}$ of the $q$-th
component particle. Recall that $P_{N}$ is the projection onto the
$N$-particle sector of the Fock space $\mathcal{F}$. From the property
of Weyl operator, we find that 
\begin{align*}
\frac{\left(\boldsymbol{a}^{*}(\mathbf{u})\right)^{\boldsymbol{N}}}{\sqrt{\boldsymbol{N}!}}\omega & =\left(\bigotimes_{q=1}^{p}d_{N_{q}}\mathcal{P}_{N_{q}}\right)\mathcal{W}(\mathbf{N}^{\odot1/2}\odot\mathbf{u}_{s})\,\omega.
\end{align*}
Since $\mathcal{H}_{\mathbf{N}}$ does not change the number of particles,
we also have that 
\begin{align*}
 & \gamma_{\mathbf{N},t}^{\mathbf{I}}(\mathbf{x};\mathbf{x}')\\
 & =\frac{1}{\boldsymbol{N}}\left\langle \frac{\left(\boldsymbol{a}^{*}(\mathbf{u})\right)^{\boldsymbol{N}}}{\sqrt{\boldsymbol{N}!}}\omega,e^{i\mathcal{H}_{\mathbf{N}}t}\boldsymbol{a}_{x}^{*}\boldsymbol{a}_{x'}e^{-i\mathcal{H}_{\mathbf{N}}t}\frac{\left(\boldsymbol{a}^{*}(\mathbf{u})\right)^{\boldsymbol{N}}}{\sqrt{\boldsymbol{N}!}}\omega\right\rangle _{\mathcal{F}^{\otimes p}}\\
 & =\frac{1}{\boldsymbol{N}}\left\langle \frac{\left(\boldsymbol{a}^{*}(\mathbf{u})\right)^{\boldsymbol{N}}}{\sqrt{\boldsymbol{N}!}}\omega,e^{i\mathcal{H}_{\mathbf{N}}t}\boldsymbol{a}_{x}^{*}\boldsymbol{a}_{x'}e^{-i\mathcal{H}_{\mathbf{N}}t}\left(\bigotimes_{q=1}^{p}d_{N_{q}}\mathcal{P}_{N_{q}}\right)\mathcal{W}(\mathbf{N}^{\odot1/2}\odot\mathbf{u}_{s})\,\omega\right\rangle _{\mathcal{F}^{\otimes p}}\\
 & =\frac{\boldsymbol{d}_{\boldsymbol{N}}}{\boldsymbol{N}}\left\langle \frac{\left(\boldsymbol{a}^{*}(\mathbf{u})\right)^{\boldsymbol{N}}}{\sqrt{\boldsymbol{N}!}}\omega,\mathcal{P}_{\mathbf{N}}e^{i\mathcal{H}_{\mathbf{N}}t}\boldsymbol{a}_{x}^{*}\boldsymbol{a}_{x'}e^{-i\mathcal{H}_{\mathbf{N}}t}\mathcal{W}(\mathbf{N}^{\odot1/2}\odot\mathbf{u}_{s})\,\omega\right\rangle _{\mathcal{F}^{\otimes p}}\\
 & =\frac{\boldsymbol{d}_{\boldsymbol{N}}}{\boldsymbol{N}}\left\langle \frac{\left(\boldsymbol{a}^{*}(\mathbf{u})\right)^{\boldsymbol{N}}}{\sqrt{\boldsymbol{N}!}}\omega,e^{i\mathcal{H}_{\mathbf{N}}t}\boldsymbol{a}_{x}^{*}\boldsymbol{a}_{x'}e^{-i\mathcal{H}_{\mathbf{N}}t}\mathcal{W}(\mathbf{N}^{\odot1/2}\odot\mathbf{u}_{s})\,\omega\right\rangle _{\mathcal{F}^{\otimes p}}
\end{align*}
To simplify it further, we use the relation 
\[
e^{\mathrm{i}\mathcal{H}_{\mathbf{N}}t}\boldsymbol{a}_{x'}e^{-\mathrm{i}\mathcal{H}_{\mathbf{N}}t}=\mathcal{W}(\mathbf{N}^{\odot1/2}\odot\mathbf{u}_{s})\mathcal{U}^{*}(t)\left(\prod_{q=1}^{p}\left(a_{x'}^{(q)}+\sqrt{N_{q}}u_{q,t}(x)\right)\right)\mathcal{U}(t)\mathcal{W}^{*}(\mathbf{N}^{\odot1/2}\odot\mathbf{u}_{s})
\]
(and an analogous result for the creation operator) to obtain that
\begin{align*}
 & \gamma_{\mathbf{N},t}^{\mathbf{I}}(\mathbf{x};\mathbf{x}')\\
 & =\frac{\boldsymbol{d}_{\boldsymbol{N}}}{\boldsymbol{N}}\left\langle \frac{\left(\boldsymbol{a}^{*}(\mathbf{u})\right)^{\boldsymbol{N}}}{\sqrt{\boldsymbol{N}!}}\omega,e^{i\mathcal{H}_{\mathbf{N}}t}\boldsymbol{a}_{x}^{*}\boldsymbol{a}_{x'}e^{-i\mathcal{H}_{\mathbf{N}}t}\mathcal{W}(\mathbf{N}^{\odot1/2}\odot\mathbf{u}_{s})\,\omega\right\rangle _{\mathcal{F}^{\otimes p}}\\
 & =\frac{\boldsymbol{d}_{\boldsymbol{N}}}{\boldsymbol{N}}\left\langle \frac{\left(\boldsymbol{a}^{*}(\mathbf{u})\right)^{\boldsymbol{N}}}{\sqrt{\boldsymbol{N}!}}\omega,\mathcal{W}(\mathbf{N}^{\odot1/2}\odot\mathbf{u}_{s})\mathcal{U}^{*}(t)\left(\prod_{r=1}^{p}\left(a_{x'}^{(r)*}+\sqrt{N_{r}}\,\overline{u_{r,t}(x)}\right)\right)\left(\prod_{q=1}^{p}\left(a_{x'}^{(q)}+\sqrt{N_{q}}u_{q,t}(x)\right)\right)\mathcal{U}(t)\,\omega\right\rangle _{\mathcal{F}^{\otimes p}}
\end{align*}
Thus, 
\begin{align*}
 & \gamma_{\mathbf{N},t}^{\mathbf{I}}(\mathbf{x};\mathbf{x}')-\overline{\mathbf{u}_{t}(\mathbf{x}')}\cdot\mathbf{u}_{t}(\mathbf{x})\\
 & =\sum_{Q,R}F_{t}^{Q,R}(\mathbf{x};\mathbf{x}')
\end{align*}
with
\begin{align*}
F_{t}^{Q,R}(\mathbf{x};\mathbf{x}') & :=\sqrt{\boldsymbol{N}_{Q^{c}}}\,u_{Q^{c},t}(x_{Q^{c}})\sqrt{\boldsymbol{N}_{R^{c}}}\,\overline{u_{R^{c},t}(x_{R^{c}})}\frac{\boldsymbol{d}_{\boldsymbol{N}}}{\boldsymbol{N}}\left\langle \frac{\left(\boldsymbol{a}^{*}(\mathbf{u})\right)^{\boldsymbol{N}}}{\sqrt{\boldsymbol{N}!}}\omega,\mathcal{W}(\mathbf{N}^{\odot1/2}\odot\mathbf{u}_{s})\left(a_{x}^{(Q)*}a_{x'}^{(R)}\right)\mathcal{U}(t)\omega\right\rangle _{\mathcal{F}^{\otimes p}}
\end{align*}
for $Q$ and $R$ be subsets of $\{1,\dots,p\}$ such that $Q\cup R\neq\emptyset$.

Recall the definition of $E_{t}^{Q,R}(J)$ in Proposition \ref{prop:Et1}
and the definition of $\mathbf{I}:=(1,\dots,1)\in\mathbb{Z}^{p}$.
For any $\mathbf{I}$-particle Hilbert-Schmidt operator $J$ on $\bigotimes_{q=1}^{p}L^{2}(\mathbb{R}^{3})$,
we have 
\begin{align*}
 & \operatorname{Tr}(J(\gamma_{\mathbf{N},t}^{\mathbf{I}}-\left|\mathbf{u}_{t}\right\rangle \left\langle \mathbf{u}_{t}\right|)\\
 & =\int\mathrm{d}\mathbf{x}\mathrm{d}\mathbf{x}'J(\mathbf{x}';\mathbf{x})\left(\gamma_{\mathbf{N},t}^{\mathbf{I}}(\mathbf{x}';\mathbf{x})-\overline{\mathbf{u}_{t}(\mathbf{x}')}\mathbf{u}_{t}(\mathbf{x})\right)\\
 & =\int\mathrm{d}\mathbf{x}\mathrm{d}\mathbf{x}'J(\mathbf{x}';\mathbf{x})\left(\sum_{Q,R}F_{t}^{Q,R}(\mathbf{x};\mathbf{x}')\right)\\
 & =:\sum_{Q\cup R\neq\emptyset}E_{t}^{Q,R}(J).
\end{align*}
Thus, from Propositions \ref{prop:Et1}, we find that 
\[
\left\Vert J(\gamma_{\mathbf{N},t}^{\mathbf{I}}-\left|\mathbf{u}_{t}\right\rangle \left\langle \mathbf{u}_{t}\right|)\right\Vert _{\mathrm{HS}}\leq Cp^{Cp}e^{Kt}\|J\|_{\mathrm{HS}}(\min_{q}c)^{-1/2}\frac{1}{N}.
\]
Since the space of compact operators is the dual to that of the trace
class operators, and since $\gamma_{\mathbf{N},t}^{\mathbf{I}}$ and
$\left|\mathbf{u}_{t}\right\rangle \left\langle \mathbf{u}_{t}\right|$
are Hermitian, using Remark \ref{rem:HSnorm-opnorm}
\[
\operatorname{Tr}\left|\gamma_{\mathbf{N},t}^{\mathbf{I}}-\left|\mathbf{u}_{t}\right\rangle \left\langle \mathbf{u}_{t}\right|\right|\leq Cp^{Cp}(\min_{q}c)^{-1/2}e^{Kt}\frac{1}{N}
\]
which concludes the proof of Theorem \ref{thm:Main_Theorem}.
\end{proof}

\section{Comparison dynamics\label{sec:Comparison-dynamics}}

This section is devoted to giving and proving lemmas.

\subsection{Comparison dynamics for $\mathcal{U}$}
\begin{lem}
\label{lem:NjU} Suppose that the assumptions in Theorem \ref{thm:Main_Theorem}
hold. Let $\mathcal{U}\left(t;s\right)$ be the unitary evolution
defined in (\ref{eq:def_mathcalU}). Then for any $\psi\in\mathcal{F}$
and $j\in\mathbb{N}$, there exist constants $C\equiv C(j)$ and $K\equiv K(j)$
such that 
\[
\left\langle \mathcal{U}\left(t;s\right)\psi,\left(\mathcal{N}_{\mathrm{total}}\right)^{j}\mathcal{U}\left(t;s\right)\psi\right\rangle _{\mathcal{F}^{\otimes p}}\leq Cp^{Cp}e^{Kt}\left\langle \psi,\left(\mathcal{N}_{\mathrm{total}}+p\mathcal{I}\right)^{2j+2}\psi\right\rangle _{\mathcal{F}^{\otimes p}}.
\]
\end{lem}

The proof of Lemma \ref{lem:NjU} will be followed by the following
argument. The argument for single component case is developed in \cite{Rodnianski2009}.
For the two-component BEC, see \cite[Proposition 3.2]{1811.04984}.
For $p\geq2$, $p$-component case will be proved by the following
argument which is a generalization of previous works \cite{1811.04984,Rodnianski2009}.

We introduce a truncated time-dependent generator with fixed $\mathbf{M}>0$
as follows:
\[
\mathcal{L}^{(\mathbf{M})}(t)=\mathcal{L}_{0}(t)+\mathcal{L}_{1}(t)+\mathcal{L}_{2}(t)+\mathcal{L}_{3}^{(\mathbf{M})}(t)+\mathcal{L}_{4}
\]
with 
\[
\begin{split}\mathcal{L}_{3}^{(\mathbf{M})}(t) & =\sum_{q=1}^{p}\Big(\frac{1}{\sqrt{N_{q}}}\int\mathrm{d}x\mathrm{d}z\,V_{qq}(x-z)a_{x}^{(q)*}(u_{q,t}(z)\chi_{q}a_{z}^{(q)*}+\overline{u_{q,t}}(z)a_{z}\chi_{q}^{(q)})a_{x}^{(q)}\Big)\\
 & \quad+\sum_{q,r=1}^{p}\Big(\frac{\sqrt{N_{q}}}{N}\int\mathrm{d}x\mathrm{d}y\,V_{qr}(x-y)a_{y}^{(r)*}(u_{q,t}(x)\chi_{q}a_{x}^{(q)*}+\overline{u_{q,t}}(x)a_{x}^{(q)}\chi_{q})a_{y}^{(r)}\Big)
\end{split}
\]
Here $\chi_{q}=\mathcal{E}^{(q)}\chi(\mathcal{N}\leq M_{q})$ where
$\chi$ denotes a characteristic function.

We remark that $\mathbf{M}$ will be chosen later such that $\mathbf{M}=\mathbf{N}$.
Define a unitary operator $\mathcal{U}^{(\mathbf{M})}$ such that
\[
\mathrm{i}\partial_{t}\mathcal{U}^{(\mathbf{M})}(t;s)=\mathcal{L}^{(\mathbf{M})}(t)\mathcal{U}^{(\mathbf{M})}(t;s)\qquad\text{and}\qquad\mathcal{U}^{(\mathbf{M})}(s;s)=1.
\]
The detailed proof will be omitted. The main idea is to follow the
steps using the argument in \cite{Rodnianski2009} and generalized
in \cite{1811.04984}. One can easily generalize further for $p$-component
case.

\subsubsection*{Step 1. Truncation with respect to $\mathcal{N}$ with $\mathbf{M}>0$.}
\begin{lem}
\label{lem:Truncation}Let $t>s$ and $\psi\in\mathcal{F}^{\otimes p}$.
Then
\begin{align*}
 & \left\langle \mathcal{U}^{(\mathbf{M})}\left(t;s\right)\psi,\left(\mathcal{N}_{\mathrm{total}}\right)^{j}\mathcal{U}^{(\mathbf{M})}\left(t;s\right)\psi\right\rangle _{\mathcal{F}^{\otimes p}}\\
 & \leq\left\langle \psi,\left(\mathcal{N}_{\mathrm{total}}+p\mathcal{I}\right)^{2j+2}\psi\right\rangle _{\mathcal{F}^{\otimes p}}Cp^{Cp}\exp\left(K_{j}(1+\sum_{q=1}^{p}\sqrt{M_{q}/N_{q}})|t-s|\right).
\end{align*}
\end{lem}

To prove this lemma, keep in mind that $N_{r}/N_{q}$, $N/N_{q}$,
and $\sqrt{N_{q}N_{r}}/N$ are bounded for any $q,r=1,\dots,p$ because
\[
\frac{N_{r}}{N_{q}}\leq\frac{N}{N_{q}}\leq(\min_{q}c_{q})^{-1}\quad\text{and}\quad\frac{\sqrt{N_{q}N_{r}}}{N}\leq1.
\]

\begin{proof}
Fix any $j\in\mathbb{N}$. Then we want to bound
\[
\frac{\mathrm{d}}{\mathrm{d}t}\left\langle \mathcal{U}^{(\mathbf{M})}\left(t;s\right)\psi,\left(\mathcal{N}_{\mathrm{total}}\right)^{j}\mathcal{U}^{(\mathbf{M})}\left(t;s\right)\psi\right\rangle 
\]
by
\[
C_{j}\left(1+\sum_{q=1}^{p}\sqrt{M_{q}/N_{q}}\right)\left\langle \mathcal{U}^{(\mathbf{M})}\left(t;s\right)\psi,\left(\mathcal{N}_{\mathrm{total}}\right)^{j}\mathcal{U}^{(\mathbf{M})}\left(t;s\right)\psi\right\rangle 
\]
for some $C_{j}$ to use Grönwall inequality.
\begin{align*}
 & \frac{\mathrm{d}}{\mathrm{d}t}\left\langle \mathcal{U}^{(\mathbf{M})}\left(t;s\right)\psi,\left(\mathcal{N}_{\mathrm{total}}\right)^{j}\mathcal{U}^{(\mathbf{M})}\left(t;s\right)\psi\right\rangle \\
 & =\sum_{\mathrm{add}\mathbf{k}=j}\binom{j}{\mathbf{k}}\mathrm{i}\left\langle \mathcal{U}^{(\mathbf{M})}\left(t;s\right)\psi,[\mathcal{L}^{(\mathbf{M})}(t),\bigotimes_{q=1}^{p}(\mathcal{N}+1)^{k_{q}}]\mathcal{U}^{(\mathbf{M})}\left(t;s\right)\psi\right\rangle \\
 & =:\sum_{\mathrm{add}\mathbf{k}=j}\binom{j}{\mathbf{k}}\sum_{\ell=1}^{3}\sum_{q,r=1}^{p}\mathcal{Q}_{\ell,q,r}
\end{align*}
where $\mathbf{k}=(k_{1},\dots,k_{p})$, $\mathrm{add}\mathbf{k}=\sum_{q=1}^{p}k_{q}$,
and the multinomial coefficient
\[
\binom{j}{\mathbf{k}}:=\frac{j!}{k_{1}!\dots k_{p}!}
\]
and 
\begin{align*}
\mathcal{Q}_{1,q,r} & =-\operatorname{Im}\int\mathrm{d}x\mathrm{d}z\,V_{qr}(x-z)u_{t}^{(q)}(z)u_{t}^{(r)}(x)\left\langle \mathcal{U}^{(\mathbf{M})}\left(t;s\right)\psi,[a_{z}^{(q)*}a_{x}^{(r)*},(\mathcal{N}+1)^{\otimes\mathbf{k}}]\mathcal{U}^{(\mathbf{M})}\left(t;s\right)\psi\right\rangle ,\\
\mathcal{Q}_{2,q,r} & =-\operatorname{Im}\int\mathrm{d}x\mathrm{d}z\,V_{qr}(x-z)u_{t}^{(q)}(z)\overline{u_{t}^{(r)}(x)}\left\langle \mathcal{U}^{(\mathbf{M})}\left(t;s\right)\psi,[a_{z}^{(q)*}a_{x}^{(r)},(\mathcal{N}+1)^{\otimes\mathbf{k}}]\mathcal{U}^{(\mathbf{M})}\left(t;s\right)\psi\right\rangle ,\\
\mathcal{Q}_{3,q,r} & =-\operatorname{Im}\int\mathrm{d}x\mathrm{d}z\,V_{qr}(x-z)\overline{u_{t}^{(q)}(z)}\left\langle \mathcal{U}^{(\mathbf{M})}\left(t;s\right)\psi,[a_{x}^{(r)*}a_{z}^{(q)}\chi_{q}a_{x}^{(r)},(\mathcal{N}+1)^{\otimes\mathbf{k}}]\mathcal{U}^{(\mathbf{M})}\left(t;s\right)\psi\right\rangle .
\end{align*}
Here, we denote
\[
(\mathcal{N}+1)^{\otimes\mathbf{k}}=\bigotimes_{q=1}^{p}(\mathcal{N}+1)^{k_{q}}
\]
for $\mathbf{k}=(k_{1},\dots,k_{p})$.

We now want to prove that for all $\ell=1,2,3$,
\begin{equation}
\sum_{q,r=1}^{p}\sum_{\mathrm{add}\mathbf{k}=j}\binom{j}{\mathbf{k}}\mathcal{Q}_{\ell,q,r}\leq C_{j}p^{Cp}\left(1+\sum_{q=1}^{p}\sqrt{M_{q}/N_{q}}\right)\left\langle \mathcal{U}^{(\mathbf{M})}\left(t;s\right)\psi,\left(\mathcal{N}_{\mathrm{total}}\right)^{j}\mathcal{U}^{(\mathbf{M})}\left(t;s\right)\psi\right\rangle .\label{eq:Q-lqr}
\end{equation}
Because it implied form Grönwall inequality, we will have the conclusion.

To have the bound (\ref{eq:Q-lqr}), we need to note that we have
commutators for single component case,
\begin{align*}
[a_{x}^{*},(\mathcal{N}+1)^{j}] & =\sum_{k=0}^{j-1}\binom{j}{k}(-1)^{k}(\mathcal{N}+1)^{k}a_{x}^{*},\\{}
[a_{x},(\mathcal{N}+1)^{j}] & =\sum_{k=0}^{j-1}\binom{j}{k}(\mathcal{N}+1)^{k}a_{x}.
\end{align*}
Consequently,
\begin{align*}
 & [a_{x}^{*}a_{y}^{*},(\mathcal{N}+1)^{j}]\\
 & =\sum_{k=0}^{j-1}\binom{j}{k}(-1)^{k}\left(a_{x}^{*}(\mathcal{N}+1)^{k}a_{y}^{*}+(\mathcal{N}+1)^{k}a_{x}^{*}a_{y}^{*}\right)\mathcal{E}^{(q)}\\
 & =\sum_{k=0}^{j-1}\binom{j}{k}(-1)^{k}\left(\mathcal{N}^{k/2}a_{x}^{*}a_{y}^{*}(\mathcal{N}+2)+(\mathcal{N}+1)^{k/2}a_{x}^{*}a_{y}^{*}(\mathcal{N}+3)^{k/2}\right)
\end{align*}
and
\[
[a_{x},(\mathcal{N}+1)^{j}]=\sum_{k=0}^{j-1}\binom{j}{k}(\mathcal{N}+1)^{k/2}a_{x}\mathcal{N}^{k/2}.
\]
Using these and denoting $\mathbf{j}=(j_{1},\dots,j_{p})$ and$\mathbf{k}_{q}$
a vector such that the $q$-th sector is $k$ and all the other $r$-th
sector is $j_{r}$, we find that
\begin{align*}
[a_{x}^{(q)*},(\mathcal{N}+1)^{\mathbf{\otimes\mathbf{j}}}] & =\sum_{\substack{\mathrm{add}\mathbf{k}_{q}=j\\
k<j
}
}\binom{j}{\mathbf{k}_{q}}(-1)^{k}(\mathcal{N}+1)^{\otimes\mathbf{k}_{q}}\left(a_{x}^{(q)*}\right)\\{}
[a_{x}^{(q)},(\mathcal{N}+1)^{\mathbf{\otimes\mathbf{j}}}] & \sum_{\substack{\mathrm{add}\mathbf{k}_{q}=j\\
k<j
}
}\binom{j}{\mathbf{k}_{q}}(\mathcal{N}+1)^{\otimes\mathbf{k}_{q}}\left(a_{x}^{(q)}\right).
\end{align*}
Consequently,
\begin{align*}
 & [a_{x}^{(q)*}a_{y}^{(q)*},(\mathcal{N}+1)^{\otimes\mathbf{j}}]\\
 & =\sum_{\substack{\mathrm{add}\mathbf{k}_{q}=j\\
k<j
}
}\binom{j}{\mathbf{k}_{q}}(-1)^{k}(\mathcal{N}+1)^{\mathbf{k}_{q}/2}\\
 & \qquad\qquad\times a_{x}^{(q)*}a_{y}^{(q)*}\left(\bigotimes_{r<q}(\mathcal{N}+1)^{j_{r}/2}\otimes(\mathcal{N}+2)^{k/2}\otimes\bigotimes_{r>q}(\mathcal{N}+1)^{j_{r}/2}\right)\\
 & \qquad+\sum_{\substack{j_{1}+\dots+k+\dots+j_{p}=j\\
k<j
}
}\binom{j}{\mathbf{k}_{q}}(-1)^{k}(\mathcal{N}+1)^{\otimes\mathbf{k}_{q}/2}\\
 & \qquad\qquad\times a_{x}^{(q)*}a_{y}^{(q)*}\left(\bigotimes_{r<q}(\mathcal{N}+1)^{j_{r}/2}\otimes(\mathcal{N}+3)^{k/2}\otimes\bigotimes_{r>q}(\mathcal{N}+1)^{j_{r}/2}\right),
\end{align*}
and suppose $q<q'$
\begin{align*}
 & [a_{x}^{(q)*}a_{y}^{(q')*},(\mathcal{N}+1)^{\otimes\mathbf{j}}]\\
 & =\sum_{\substack{\mathrm{add}\mathbf{k}_{q}=j\\
k<j
}
}\binom{j}{\mathbf{k}_{q}}(-1)^{k}(\mathcal{N}+1)^{\otimes\mathbf{k}_{q}/2}\\
 & \qquad\qquad\times a_{x}^{(q)*}a_{y}^{(q')*}\left(\bigotimes_{r<q}(\mathcal{N}+1)^{j_{r}/2}\otimes(\mathcal{N}+2)^{k/2}\otimes\bigotimes_{q<r<q'}(\mathcal{N}+1)^{j_{r}/2}\otimes(\mathcal{N}+2)^{k/2}\otimes\bigotimes_{r>q'}(\mathcal{N}+1)^{j_{r}/2}\right)\\
 & \qquad+\sum_{\substack{\mathrm{add}\mathbf{k}_{q}=j\\
k<j
}
}\binom{j}{\mathbf{k}_{q}}(-1)^{k}\left(\bigotimes_{r<q}(\mathcal{N}+1)^{j_{r}/2}\otimes(\mathcal{N})^{j_{q}/2}\otimes\bigotimes_{q<r<q'}(\mathcal{N}+1)^{j_{r}/2}\otimes(\mathcal{N}+1)^{k/2}\otimes\bigotimes_{r>q'}(\mathcal{N}+1)^{j_{r}/2}\right)\\
 & \qquad\qquad\times a_{x}^{(q)*}a_{y}^{(q')*}\left(\bigotimes_{r<q'}(\mathcal{N}+1)^{j_{r}/2}\otimes(\mathcal{N}+2)^{k/2}\otimes\bigotimes_{r>q'}(\mathcal{N}+1)^{j_{r}/2}\right).
\end{align*}
We now want to estimate
\[
\sum_{\substack{\sum_{r\neq q}j_{r}+k=j\\
k<j
}
}\binom{j}{\mathbf{k}_{q}}|\mathcal{Q}_{\ell,q,q'}|
\]
for $\ell=1,2,3$. To have a bound of this, we need to bound each
$|\mathcal{Q}_{\ell,q,r}|$. Using the above and denoting $\boldsymbol{\ell}_{q}$
a vector such that the $q$-th sector is $\ell$ and all the other
$r$-th sector is $j_{r}$ 
\begin{align*}
\mathcal{Q}_{1,q,q'} & =\sum_{\substack{\mathrm{add}\boldsymbol{\ell}_{q}=j-k\\
\ell<j-k
}
}\binom{j-k}{\boldsymbol{\ell}_{q}}(-1)^{\ell+1}\operatorname{Im}\int\mathrm{d}x\,u_{q,t}(x)\\
 & \quad\times\Bigg\langle a_{x}^{(q)*}\bigotimes_{r<q}(\mathcal{N}+1)^{j_{r}/2}\otimes(\mathcal{N}+1)^{\ell/2}\otimes\bigotimes_{q<r<q'}(\mathcal{N}+1)^{j_{r}/2}\\
 & \qquad\qquad\qquad\qquad\otimes(\mathcal{N}+2)^{k/2}\otimes\bigotimes_{q<r<q'}(\mathcal{N}+1)^{j_{r}/2}\big)\mathcal{U}^{(\mathbf{M})}\left(t;s\right)\psi,\\
 & \quad\qquad\left(a^{(r)*}(V_{qr}(x-\cdot)u_{r,t}(\cdot))\right)\big(\bigotimes_{r<q}(\mathcal{N}+1)^{j_{r}/2}\otimes(\mathcal{N}+1)^{\ell/2}\\
 & \quad\qquad\quad\times\otimes\bigotimes_{q<r<q'}(\mathcal{N}+1)^{j_{r}/2}\otimes(\mathcal{N}+2)^{k/2}\otimes\bigotimes_{q<r<q'}(\mathcal{N}+1)^{j_{r}/2}\big)\mathcal{U}^{(\mathbf{M})}\left(t;s\right)\psi\Bigg\rangle\\
 & \quad+\sum_{\substack{\mathrm{add}\boldsymbol{\ell}_{q}=k\\
\ell<k
}
}\binom{k}{\boldsymbol{\ell}_{q}}(-1)^{\ell+1}\operatorname{Im}\int\mathrm{d}x\,u_{q,t}(x)\\
 & \quad\qquad\times\Bigg\langle a_{x}^{(q)*}\big(\bigotimes_{r<q}(\mathcal{N}+1)^{j_{r}/2}\otimes(\mathcal{N})^{(j-k)/2}\otimes\bigotimes_{q<r<q'}(\mathcal{N}+1)^{j_{r}/2}\\
 & \qquad\qquad\qquad\qquad\qquad\otimes(\mathcal{N}+1)^{\ell/2}\otimes\bigotimes_{q<r<q'}(\mathcal{N}+1)^{j_{r}/2}\big)\mathcal{U}^{(\mathbf{M})}\left(t;s\right)\psi,\\
 & \quad\qquad\qquad\left(a^{(r)*}(V_{qr}(x-\cdot)u_{r,t}(\cdot))\right)\\
 & \quad\qquad\qquad\quad\times\big(\bigotimes_{r<q}(\mathcal{N}+1)^{j_{r}/2}\otimes(\mathcal{N}+1)^{(j-k)/2}\otimes\bigotimes_{q<r<q'}(\mathcal{N}+1)^{j_{r}/2}\\
 & \qquad\qquad\qquad\qquad\qquad\otimes(\mathcal{N}+2)^{\ell/2}\otimes\bigotimes_{q<r<q'}(\mathcal{N}+1)^{j_{r}/2}\big)\mathcal{U}^{(\mathbf{M})}\left(t;s\right)\psi\Bigg\rangle
\end{align*}
Then, using the Cauchy-Schwarz inequality and Lemma (\ref{lem:dGammaQphiQ}),
\[
|\mathcal{Q}_{1,q,q'}|\leq I+II
\]
Thus, using the assumption on $V_{qr}$ in (\ref{eq:assumption_V}),
the first sum is bounded by
\begin{align*}
I & \leq C\sum_{\substack{\mathrm{add}\boldsymbol{\ell}_{q}=j-k\\
\ell<j-k
}
}\binom{j-k}{\boldsymbol{\ell}_{q}}\Bigg(\int\mathrm{d}x\,\Bigg\Vert a_{x}^{(q)*}\big(\bigotimes_{r<q}(\mathcal{N}+1)^{j_{r}/2}\otimes(\mathcal{N}+1)^{\ell/2}\otimes\bigotimes_{q<r<q'}(\mathcal{N}+1)^{j_{r}/2}\\
 & \qquad\qquad\otimes(\mathcal{N}+2)^{k/2}\otimes\bigotimes_{q<r<q'}(\mathcal{N}+1)^{j_{r}/2}\big)\mathcal{U}^{(\mathbf{M})}\left(t;s\right)\psi\Bigg\Vert_{\mathcal{F}^{\otimes p}}\Bigg)^{1/2}\\
 & \quad\times\Bigg\Vert\big(\bigotimes_{r<q}(\mathcal{N}+1)^{j_{r}/2}\otimes(\mathcal{N}+1)^{\ell/2}\otimes\bigotimes_{q<r<q'}(\mathcal{N}+1)^{j_{r}/2}\\
 & \qquad\qquad\otimes(\mathcal{N}+2)^{(k+1)/2}\otimes\bigotimes_{q<r<q'}(\mathcal{N}+1)^{j_{r}/2}\big)\mathcal{U}^{(\mathbf{M})}\left(t;s\right)\psi\Bigg\Vert_{\mathcal{F}^{\otimes p}}.
\end{align*}
The second term is bounded by, also using the assumption on $V_{qr}$
in (\ref{eq:assumption_V}),
\begin{align*}
II & \leq C\sum_{\substack{\mathrm{add}\boldsymbol{\ell}_{q}=k\\
\ell<k
}
}\binom{k}{\boldsymbol{\ell}_{q}}\Bigg(\int\mathrm{d}x\,\Bigg\Vert a_{x}^{(q)*}\big(\bigotimes_{r<q}(\mathcal{N}+1)^{j_{r}/2}\otimes(\mathcal{N})^{(j-k)/2}\otimes\bigotimes_{q<r<q'}(\mathcal{N}+1)^{j_{r}/2}\\
 & \qquad\qquad\otimes(\mathcal{N}+1)^{\ell/2}\otimes\bigotimes_{q<r<q'}(\mathcal{N}+1)^{j_{r}/2}\big)\mathcal{U}^{(\mathbf{M})}\left(t;s\right)\psi\Bigg\Vert_{\mathcal{F}^{\otimes p}}\Bigg)^{1/2}\\
 & \quad\times\Bigg\Vert\big(\bigotimes_{r<q}(\mathcal{N}+1)^{j_{r}/2}\otimes(\mathcal{N})^{(j-k)/2}\otimes\bigotimes_{q<r<q'}(\mathcal{N}+1)^{j_{r}/2}\\
 & \qquad\qquad\otimes(\mathcal{N}+1)^{(\ell+1)/2}\otimes\bigotimes_{q<r<q'}(\mathcal{N}+1)^{j_{r}/2}\big)\mathcal{U}^{(\mathbf{M})}\left(t;s\right)\psi\Bigg\Vert_{\mathcal{F}^{\otimes p}}.
\end{align*}
Applying the Cauchy-Schwarz inequality again and rewriting with multinomial
coefficient, we obtain
\[
\sum_{\substack{\sum_{r\neq q}j_{r}+k=j\\
k<j
}
}\binom{j}{\mathbf{k}_{q}}|\mathcal{Q}_{1,q,r}|\leq C_{j}\left\langle \mathcal{U}^{(\mathbf{M})}\left(t;s\right)\psi,\left(\mathcal{N}_{\mathrm{total}}\right)^{j}\mathcal{U}^{(\mathbf{M})}\left(t;s\right)\psi\right\rangle .
\]
Similar bound can be obtained for 
\[
\sum_{\substack{\sum_{r\neq q}j_{r}+k=j\\
k<j
}
}\binom{j}{\mathbf{k}_{q}}|\mathcal{Q}_{2,q,r}|.
\]
This upper bound contributes to the first term on the right hand side
of (\ref{eq:Q-lqr}).

Now it remains to bound 
\[
\sum_{\substack{\sum_{r\neq q}j_{r}+k=j\\
k<j
}
}\binom{j}{\mathbf{k}_{q}}|\mathcal{Q}_{3,q,r}|
\]
Noting that $[\chi_{q},a_{x}^{(r)}]=0$,
\[
[\chi_{q},(\mathcal{N}+1)^{\otimes\mathbf{j}}]=0
\]
and
\[
\|a^{(r)}(V_{qr}(x-\cdot)u_{r,t}(\cdot))\chi_{r}\|\leq C\sqrt{M_{r}}
\]
with similar approach to bound $|\mathcal{Q}_{1,q,q'}|$, we can bound
that
\[
\frac{1}{\sqrt{N_{r}}}\sum_{\substack{\sum_{r\neq q}j_{r}+k=j\\
k<j
}
}\binom{j}{\mathbf{k}_{q}}|\mathcal{Q}_{1,q,r}|\leq C_{j}\sqrt{M_{r}/N_{r}}\left\langle \mathcal{U}^{(\mathbf{M})}\left(t;s\right)\psi,\left(\mathcal{N}_{\mathrm{total}}\right)^{j}\mathcal{U}^{(\mathbf{M})}\left(t;s\right)\psi\right\rangle .
\]
\end{proof}

\subsubsection*{Step 2. Weak bounds on the $\mathcal{U}$ dynamics}
\begin{lem}
\label{lem:Weeakbounds}For arbitrary $t,s\in\mathbb{R}$ and $\psi\in\mathcal{F}^{\otimes p}$,
we have
\begin{equation}
\left\langle \psi,\mathcal{U}^{*}(t,s)(\mathcal{N}_{\mathrm{total}})\mathcal{U}(t,s)\psi\right\rangle \leq Cp^{Cp}\left\langle \psi,(\mathcal{N}_{\mathrm{total}}+N+1)\psi\right\rangle \label{eq:UstarNtotU}
\end{equation}
Moreover, for every $j\in\mathbb{N}$,
\begin{align}
\left\langle \psi,\mathcal{U}^{*}(t,s)(\mathcal{N}_{\mathrm{total}})^{2j}\mathcal{U}(t,s)\psi\right\rangle  & \leq C_{j}p^{Cp}\left\langle \psi,(\mathcal{N}_{\mathrm{total}}+N)^{2j}\psi\right\rangle \label{eq:UstartNtotsaqurtejU}\\
\left\langle \psi,\mathcal{U}^{*}(t,s)(\mathcal{N}_{\mathrm{total}})^{2j+1}\mathcal{U}(t,s)\psi\right\rangle  & \leq C_{j}p^{Cp}\left\langle \psi,(\mathcal{N}_{\mathrm{total}}+N)^{2j+1}(\mathcal{N}_{\mathrm{total}}+1)\psi\right\rangle \label{eq:UstartNtotsaqurtej+1U}
\end{align}
for appropriate constant $C_{j}$.
\end{lem}

\begin{proof}
Using Lemma \ref{lem:weyl}, one may calculate
\[
\mathcal{U}^{*}(t;s)\left(\mathcal{E}^{(q)}\mathcal{N}\right)\mathcal{U}^{*}(t;s)=\mathcal{W}_{s}^{*}((\mathcal{E}^{(q)}\mathcal{N})-\sqrt{N_{q}}\,e^{\mathrm{i}(t-s)\mathcal{H}_{\mathbf{N}}}\,\phi^{(q)}(u_{q,t})\,e^{-\mathrm{i}(t-s)\mathcal{H}_{\mathbf{N}}}+N_{q})\mathcal{W}_{s}.
\]
Then using the Cauchy-Schwarz inequality, Lemma \ref{lem:dGammaQphiQ},
and Lemma \ref{lem:weyl}, we have
\[
\left\langle \psi,\mathcal{U}^{*}(t;s)\left(\mathcal{E}^{(q)}\mathcal{N}\right)\mathcal{U}^{*}(t;s)\psi\right\rangle _{\mathcal{F}^{\otimes p}}\leq C\left\langle \psi,\left(\mathcal{E}^{(q)}\mathcal{N}+N_{q}+1\right)\psi\right\rangle _{\mathcal{F}^{\otimes p}}.
\]
Therefore, we have 
\[
\left\langle \psi,\mathcal{U}(t,s)^{*}(\mathcal{N}_{\mathrm{total}})\mathcal{U}(t,s)\psi\right\rangle \leq C\left\langle \psi,(\mathcal{N}_{\mathrm{total}}+N+1)\psi\right\rangle 
\]
which is (\ref{eq:UstarNtotU}).

To show (\ref{eq:UstartNtotsaqurtejU}), we use induction on $j$.
Define
\[
Z_{t,s}:=\mathcal{N}_{\mathrm{total}}+e^{\mathrm{i}(t-s)\mathcal{H}_{\mathbf{N}}}\left(\sum_{q=1}^{p}\sqrt{N_{q}}\phi^{(q)}(u_{q,t})\right)e^{-\mathrm{i}(t-s)\mathcal{H}_{\mathbf{N}}}+N.
\]
Then
\[
\mathcal{U}(t,s)^{*}\mathcal{N}_{\mathrm{total}}\mathcal{U}(t,s)=\mathcal{W}_{s}^{*}Z_{t,s}\mathcal{W}_{s}.
\]
We argue that
\begin{equation}
Z_{t,s}^{2}\leq C(\mathcal{N}_{\mathrm{total}}+N)^{2}\label{eq:ind1}
\end{equation}
for some constant $C$. By Cauchy-Schwarz inequality, we have
\[
\left\langle \psi,\mathcal{Z}_{t,s}^{2}\psi\right\rangle \leq C(A+B)
\]
with
\begin{align*}
A & =\left\langle \psi,(\mathcal{N}_{\mathrm{total}}+N)^{2}\psi\right\rangle \\
B & =\left\langle \psi,e^{\mathrm{i}(t-s)\mathcal{H}_{\mathbf{N}}}\left(\sum_{q=1}^{p}\sqrt{N_{q}}\phi^{(q)}(u_{q,t})\right)^{2}e^{-\mathrm{i}(t-s)\mathcal{H}_{\mathbf{N}}}\psi\right\rangle .
\end{align*}
Using the Cauchy-Schwarz inequality and Lemma \ref{lem:dGammaQphiQ},
we obtain
\[
B\leq C\left\langle \psi,(\mathcal{N}_{\mathrm{total}}+N)^{2}\psi\right\rangle .
\]
These estimates for $A$ and $B$ prove our argument.

\begin{equation}
\mathrm{ad}_{Z}^{m}(\mathcal{N}_{\mathrm{total}})\mathrm{ad}_{Z}^{m}(\mathcal{N}_{\mathrm{total}})\leq C(\mathcal{N}_{\mathrm{total}}+N)^{2}\label{eq:ind2}
\end{equation}
for all $m\in\mathbb{N}$ by following proof is given by performing
straightforward calculation.

By induction, it follows that, for all $j\in\mathbb{N}$, there are
constants $C_{j}$ and $D_{j}$ such that 
\begin{equation}
Z_{t,s}^{j-1}(\mathcal{N}_{\mathrm{total}}+N)^{2}Z_{t,s}^{j-1}\le C_{j}(\mathcal{N}_{\mathrm{total}}+N)^{2j},\label{Z1}
\end{equation}
\begin{equation}
Z_{t,s}^{2j}\le D_{j}(\mathcal{N}_{\mathrm{total}}+N)^{2j}.\label{Z2}
\end{equation}
In fact, this is already proved for $j=1$ by (\ref{eq:ind1}) and
(\ref{eq:ind2}). Suppose that (\ref{Z1}) and (\ref{Z2}) hold true
for all $j<k$. We will prove the estimates for $j=k$.

For (\ref{Z1}), we have 
\[
\begin{split} & Z_{t,s}^{k-1}(\mathcal{N}_{\mathrm{total}}+N)^{2}Z_{t,s}^{k-1}\\
 & =(\mathcal{N}_{\mathrm{total}}+N)Z_{t,s}^{k-1}(\mathcal{N}_{\mathrm{total}}+N)Z_{t,s}^{k-1}\\
 & \quad+[Z_{t,s}^{k-1},\mathcal{N}_{\mathrm{total}}+N](\mathcal{N}_{\mathrm{total}}+N)Z_{t,s}^{k-1}\\
 & \le Cp^{Cp}\Big((\mathcal{N}_{\mathrm{total}}+N)Z_{t,s}^{2k-2}(\mathcal{N}_{\mathrm{total}}+N)\\
 & \quad+Z_{t,s}^{k-1}(\mathcal{N}_{\mathrm{total}}+N)^{2}Z_{t,s}^{k-1}\\
 & \quad+[Z_{t,s}^{k-1},\mathcal{N}_{\mathrm{total}}+N]\,[Z_{t,s}^{k-1},\mathcal{N}_{\mathrm{total}}+N]^{*}\Big).
\end{split}
\]
Thus 
\[
\begin{split} & Z_{t,s}^{k-1}(\mathcal{N}_{\mathrm{total}}+N)^{2}Z_{t,s}^{k-1}\\
 & \le Cp^{Cp}\Big((\mathcal{N}_{\mathrm{total}}+N)Z_{t,s}^{2k-2}(\mathcal{N}_{\mathrm{total}}+N)\\
 & \quad+[Z_{t,s}^{k-1},\mathcal{N}_{\mathrm{total}}+N]\,[Z_{t,s}^{k-1},\mathcal{N}_{\mathrm{total}}+N]^{*}\Big)\\
 & \le C_{k}p^{Cp}\Big((\mathcal{N}_{\mathrm{total}}+N)^{2k},
\end{split}
\]
as desired. Here, we used the operator inequality 
\[
\begin{split} & [Z_{t,s}^{k-1},\mathcal{N}_{\mathrm{total}}+N]\,[Z_{t,s}^{k-1},\mathcal{N}_{\mathrm{total}}+N]^{*}\\
 & \le(\mathcal{N}_{\mathrm{total}}+N)^{2k},
\end{split}
\]
which follows using the commutator expansion 
\[
[Z_{t,s}^{k-1},\mathcal{N}_{\mathrm{total}}]=\sum_{m=0}^{k-2}\binom{k-1}{m}Z_{t,s}^{m}\mathrm{ad}_{Z_{t,s}}^{k-1-m}(\mathcal{N}_{\mathrm{total}}).
\]

For (\ref{Z2}), using (\ref{eq:ind1}) and (\ref{Z1}), we obtain
\[
\begin{split}Z_{t,s}^{2k} & =Z_{t,s}^{k-1}Z_{t,s}^{2}Z_{t,s}^{k-1}\\
 & \le Cp^{Cp}Z_{t,s}^{k-1}(\mathcal{N}_{\mathrm{total}}+N)^{2}Z_{t,s}^{k-1}\\
 & \le Cp^{Cp}C_{j}(\mathcal{N}_{\mathrm{total}}+N)^{2k},
\end{split}
\]
as desired. Therefore, we have proved (\ref{Z1}) and (\ref{Z2}).

Let us finish the proof of Lemma (\ref{lem:Weeakbounds}). First we
observe that, similarly as we estimated $Z_{t,s}^{2j}$, we can prove
that 
\[
(\mathcal{N}_{\mathrm{total}}+\sum_{q=1}^{p}\sqrt{N_{q}}(\phi^{(q)}(u_{q,t}))+2N)^{2j}\le C_{j}p^{2p}(\mathcal{N}_{\mathrm{total}}+N)^{2j}
\]
for all $j\in\mathbb{N}$. Hence 
\[
\begin{split} & \langle\psi,\mathcal{U}(t,s)^{*}(\mathcal{N}_{\mathrm{total}})^{2j}\mathcal{U}(t,s)\psi\rangle\\
 & =\langle\mathcal{W}_{t}\psi,Z_{t,s}^{2j}\mathcal{W}_{t}\psi\rangle\\
 & \le C_{j}p^{Cp}\langle\mathcal{W}_{t}\psi,(\mathcal{N}_{\mathrm{total}}+N)^{2j}\mathcal{W}_{t}\psi\rangle\\
 & =C_{j}p^{Cp}\langle\psi,(\mathcal{N}_{\mathrm{total}}+\sum_{q=1}^{p}\sqrt{N_{q}}(\phi^{(q)}(u_{t}))+2N)^{2j}\psi\rangle\\
 & \le C_{j}p^{Cp}\langle\psi,(\mathcal{N}_{\mathrm{total}}+N)^{2j}\psi\rangle.
\end{split}
\]
This proves (\ref{eq:UstartNtotsaqurtejU}).

Then, using Cauchy-Schwarz inequality, we have
\begin{align*}
 & \left\langle \psi,\mathcal{U}^{*}(t,s)(\mathcal{N}_{\mathrm{total}})^{2j+1}\mathcal{U}(t,s)\psi\right\rangle \\
 & =\left\langle \frac{(\mathcal{N}_{\mathrm{total}})^{j+1}}{\sqrt{N}}\mathcal{U}(t,s)\psi,\sqrt{N}(\mathcal{N}_{\mathrm{total}})^{j}\mathcal{U}(t,s)\psi\right\rangle \\
 & \leq\frac{1}{N}\left\langle \psi,\mathcal{U}^{*}(t,s)(\mathcal{N}_{\mathrm{total}})^{2j+2}\mathcal{U}(t,s)\psi\right\rangle +N\left\langle \psi,\mathcal{U}^{*}(t,s)(\mathcal{N}_{\mathrm{total}})^{2j}\mathcal{U}(t,s)\psi\right\rangle \\
 & \leq\frac{C_{j+1}p^{Cp}}{N}\left\langle \psi,\mathcal{U}^{*}(t,s)(\mathcal{N}_{\mathrm{total}}+N)^{2j+2}\mathcal{U}(t,s)\psi\right\rangle +NC_{j}\left\langle \psi,\mathcal{U}^{*}(t,s)(\mathcal{N}_{\mathrm{total}}+N)^{2j}\mathcal{U}(t,s)\psi\right\rangle \\
 & \leq C_{j}p^{Cp}\left\langle \psi,(\mathcal{N}_{\mathrm{total}}+N)^{2j+1}(\mathcal{N}_{\mathrm{total}}+1)\psi\right\rangle .
\end{align*}
Thus we get (\ref{eq:UstartNtotsaqurtej+1U}).
\end{proof}

\subsubsection*{Step 3. Comparison between $\mathcal{U}$ and $\mathcal{U}^{(M)}$
dynamics.}
\begin{lem}
\label{lem:Comparision-U-UM}For every $j\in\mathbb{N}$, there exists
constant $C_{j}$ and $K_{j}$ such that
\begin{equation}
\begin{aligned} & \left|\left\langle \mathcal{U}\left(t;s\right)\psi,\left(\mathcal{N}_{\mathrm{total}}\right)^{j}\left(\mathcal{U}\left(t;s\right)-\mathcal{U}^{(\mathbf{M})}\left(t;s\right)\right)\psi\right\rangle _{\mathcal{F}^{\otimes p}}\right|\\
 & \qquad\leq C_{j}p^{Cp}\left[\sum_{q=1}^{p}\left(\frac{N}{M_{q}}\right)^{j}\right]\left\Vert \left(\mathcal{N}_{\mathrm{total}}+p\mathcal{I}\right)^{j+1}\psi\right\Vert ^{2}\frac{\exp\left(K_{j}\left(1+\sum_{q=1}^{p}\sqrt{M_{q}/N_{q}}\right)|t-s|\right)}{1+\sum_{q=1}^{p}\sqrt{M_{q}/N_{q}}}
\end{aligned}
\label{eq:Comparison-U-UM-1}
\end{equation}
and
\begin{equation}
\begin{aligned} & \left|\left\langle \mathcal{U}^{(\mathbf{M})}\left(t;s\right)\psi,\left(\mathcal{N}_{\mathrm{total}}\right)^{j}\left(\mathcal{U}\left(t;s\right)-\mathcal{U}^{(\mathbf{M})}\left(t;s\right)\right)\psi\right\rangle _{\mathcal{F}^{\otimes p}}\right|\\
 & \qquad\leq C_{j}p^{Cp}\left[\sum_{q=1}^{p}\frac{1}{M_{q}^{j}}\right]\left\Vert \left(\mathcal{N}_{\mathrm{total}}+p\mathcal{I}\right)^{j+1}\psi\right\Vert ^{2}\frac{\exp\left(K_{j}\left(1+\sum_{q=1}^{p}\sqrt{M_{q}/N_{q}}\right)|t-s|\right)}{1+\sum_{q=1}^{p}\sqrt{M_{q}/N_{q}}}
\end{aligned}
\label{eq:Comparison-U-UM-2}
\end{equation}
for all $\psi\in\mathcal{F}^{\otimes p}$ and $t>s$.
\end{lem}

\begin{proof}
The proof of this lemma is also followed by Lemma 5.4 in \cite{1811.04984}.
To prove (\ref{eq:Comparison-U-UM-1}), we use
\[
\mathcal{U}\left(t;s\right)-\mathcal{U}^{(\mathbf{M})}\left(t;s\right)=-\mathrm{i}\int_{s}^{t}\mathrm{d}r\,\mathcal{U}\left(t;s\right)\left(\mathcal{L}(r)-\mathcal{L}^{(\mathbf{M})}(r)\right)\mathcal{U}^{(\mathbf{M})}\left(t;s\right).
\]
Note that
\begin{align*}
 & \mathcal{L}(t)-\mathcal{L}^{(\mathbf{M})}(t)\\
 & =\sum_{q=1}^{p}\Big(\frac{1}{\sqrt{N_{q}}}\int\mathrm{d}x\mathrm{d}z\,V_{qq}(x-z)a_{x}^{(q)*}(u_{q,t}(z)\chi_{q}^{c}a_{z}^{(q)*}+\overline{u_{q,t}}(z)a_{z}^{(q)}\chi_{q}^{c})a_{x}^{(r)}\Big)\\
 & \quad+\sum_{q,r=1}^{p}\Big(\frac{\sqrt{N_{q}}}{N}\int\mathrm{d}x\mathrm{d}y\,V_{qr}(x-y)a_{y}^{(r)*}(u_{q,t}(x)\chi_{q}^{c}a_{x}^{(q)*}+\overline{u_{q,t}}(x)a_{x}^{(q)}\chi_{q}^{c})a_{y}^{(r)}\Big)
\end{align*}
where $\chi_{q}^{c}=\mathcal{E}^{(q)}\chi(\mathcal{N}>M_{q}).$

Using these formulae, we write
\[
\left\langle \mathcal{U}\left(t;s\right)\psi,\left(\mathcal{N}_{\mathrm{total}}\right)^{j}\left(\mathcal{U}\left(t;s\right)-\mathcal{U}^{(\mathbf{M})}\left(t;s\right)\right)\psi\right\rangle _{\mathcal{F}^{\otimes p}}=\sum_{q=1}^{p}(w_{1,q}+w_{2,q})+\sum_{q,r=1}^{p}(w_{3,q,r}+w_{4,q,r})
\]
where
\begin{align*}
w_{1,q} & =\frac{-\mathrm{i}}{\sqrt{N_{q}}}\int_{s}^{t}\mathrm{d}r\int\mathrm{d}x\langle a_{x}^{(q)}\mathcal{U}^{*}\left(t;s\right)(\mathcal{N}_{\mathrm{total}})^{j}\,\mathcal{U}\left(t;s\right)\psi,\\
 & \qquad\qquad\qquad\qquad\qquad\qquad a^{(q)}(V_{qq}(x-\cdot)u_{q,t}(\cdot))qa_{x}^{(q)}\mathcal{U}^{(\mathbf{M})}\left(t;s\right)\psi\rangle,\\
w_{2,q} & =\frac{-\mathrm{i}}{\sqrt{N_{q}}}\int_{s}^{t}\mathrm{d}r\int\mathrm{d}x\langle a_{x}^{(q)}\mathcal{U}^{*}\left(t;s\right)(\mathcal{N}_{\mathrm{total}})^{j}\,\mathcal{U}\left(t;s\right)\psi,\\
 & \qquad\qquad\qquad\qquad\qquad\qquad\chi_{q}^{c}a^{(q)*}(V_{qq}(x-\cdot)u_{q,t}(\cdot))a_{x}^{(q)}\mathcal{U}^{(\mathbf{M})}\left(t;s\right)\psi\rangle,\\
w_{3,q,r} & =\frac{-\mathrm{i}\sqrt{N_{q}}}{N}.\int_{s}^{t}\mathrm{d}r\int\mathrm{d}y\langle a_{y}^{(r)}\mathcal{U}^{*}\left(t;s\right)(\mathcal{N}_{\mathrm{total}})^{j}\,\mathcal{U}\left(t;s\right)\psi,\\
 & \qquad\qquad\qquad\qquad a^{(r)}(V_{qr}(y-\cdot)u_{r,t}(\cdot))\chi_{r}^{c}a_{y}^{(q)}\mathcal{U}^{(\mathbf{M})}\left(t;s\right)\psi\rangle,\\
w_{4,q,r} & =\frac{-\mathrm{i}\sqrt{N_{q}}}{N}\int_{s}^{t}\mathrm{d}r\int\mathrm{d}y\langle a_{y}^{(q)}\,\mathcal{U}^{*}\left(t;s\right)(\mathcal{N}_{\mathrm{total}})^{j}\,\mathcal{U}\left(t;s\right)\psi,\\
 & \qquad\qquad\qquad\qquad\chi_{r}^{c}a^{(r)*}(V_{qr}(y-\cdot)u_{r,t}(\cdot))a_{y}^{(q)}\mathcal{U}^{(\mathbf{M})}\left(t;s\right)\psi\rangle.
\end{align*}

Now we bound each $w_{1,q}$, $w_{2,q}$, $w_{3,q,r}$, and $w_{4,q.r}$.
First,
\begin{align*}
 & |w_{1,q}|\\
 & \leq\frac{1}{\sqrt{N_{q}}}\int_{s}^{t}\mathrm{d}r\int\mathrm{d}x\left\Vert a_{x}^{(q)}\mathcal{U}^{*}\left(t;s\right)(\mathcal{N}_{\mathrm{total}})^{j}\,\mathcal{U}\left(t;s\right)\psi\right\Vert _{\mathcal{F}^{\otimes p}}\\
 & \qquad\qquad\times\left\Vert a_{x}^{(q)}(V_{qq}(x-\cdot)u_{q,t}(\cdot))\chi_{q}^{c}a_{x}^{(q)}\mathcal{U}^{(\mathbf{M})}\left(t;s\right)\psi\right\Vert _{\mathcal{F}^{\otimes p}}\\
 & \leq\frac{1}{\sqrt{N_{q}}}\sup_{x}\|V_{qq}(x-\cdot)u_{q,t}(\cdot)\|_{2}\int_{s}^{t}\mathrm{d}r\int\mathrm{d}x\left\Vert a_{x}^{(q)}\mathcal{U}^{*}\left(t;s\right)(\mathcal{N}_{\mathrm{total}})^{j}\,\mathcal{U}\left(t;s\right)\psi\right\Vert _{\mathcal{F}^{\otimes p}}\\
 & \qquad\qquad\times\left\Vert a_{x}^{(q)}\left(\mathcal{E}^{(q)}(\mathcal{N}-1)^{1/2}\right)\chi_{q}^{c}(M_{q}+1)\mathcal{U}^{(\mathbf{M})}\left(t;s\right)\psi\right\Vert _{\mathcal{F}^{\otimes p}}\\
 & \leq\frac{C}{\sqrt{N_{q}}}\int_{s}^{t}\mathrm{d}r\left\Vert \left(\mathcal{E}^{(q)}\mathcal{N}^{1/2}\right)\mathcal{U}^{*}\left(t;s\right)(\mathcal{N}_{\mathrm{total}})^{j}\,\mathcal{U}\left(t;s\right)\psi\right\Vert _{\mathcal{F}^{\otimes p}}\left\Vert \left(\mathcal{E}^{(q)}\mathcal{N}\right)\chi_{q}^{c}\mathcal{U}^{(\mathbf{M})}\left(t;s\right)\psi\right\Vert _{\mathcal{F}^{\otimes p}}.
\end{align*}
Using Lemma \ref{lem:NjU}, \ref{lem:Truncation}, and \ref{lem:Weeakbounds},
one may have
\begin{align*}
 & |w_{1,q}|\\
 & \leq C_{j}\left(\frac{N}{M_{q}}\right)^{j}\left(\frac{N}{N_{q}}\right)^{1/2}\|(\mathcal{N}_{\mathrm{total}}+p\mathcal{I})^{j+1}\psi\|\int_{s}^{t}\mathrm{d}r\left\langle \mathcal{U}^{(\mathbf{M})}\left(t;s\right)\psi,(\mathcal{N}_{\mathrm{total}})^{2j+2}\mathcal{U}^{(\mathbf{M})}\left(t;s\right)\psi\right\rangle \\
 & \leq C_{j}\left(\frac{N}{M_{q}}\right)^{j}\|(\mathcal{N}_{\mathrm{total}}+p\mathcal{I})^{j+1}\psi\|^{2}\int_{s}^{t}\mathrm{d}r\,\exp\left(K_{j}\left(1+\sum_{q=1}^{p}\sqrt{M_{q}/N_{q}}\right)(r-s)\right)\\
 & \leq C_{j}\left(\frac{N}{M_{q}}\right)^{j}\|(\mathcal{N}_{\mathrm{total}}+p\mathcal{I})^{j+1}\psi\|^{2}\frac{\exp\left(K_{j}\left(1+\sum_{q=1}^{p}\sqrt{M_{q}/N_{q}}\right)(t-s)\right)}{1+\sum_{q=1}^{p}\sqrt{M_{q}/N_{q}}}
\end{align*}

Similarly, one can estimate $w_{2,q,r}$, $w_{3,q,r}$, and $w_{4,q,r}$
and obtain similar bounds. This proves (\ref{eq:Comparison-U-UM-1}).
One can also obtain (\ref{eq:Comparison-U-UM-2}) using similar technique
as in the proof of (\ref{eq:Comparison-U-UM-1}).
\end{proof}
\begin{proof}[Proof of Lemma \ref{lem:NjU}]
 From Lemmas \ref{lem:Truncation}-\ref{lem:Comparision-U-UM} with
the choice $\mathbf{M}=\mathbf{N}$,
\begin{align*}
\left\langle \mathcal{U}\left(t;s\right)\psi,\left(\mathcal{N}_{\mathrm{total}}\right)^{j}\mathcal{U}\left(t;s\right)\psi\right\rangle _{\mathcal{F}^{\otimes p}} & =\left\langle \mathcal{U}\left(t;s\right)\psi,\left(\mathcal{N}_{\mathrm{total}}\right)^{j}\left(\mathcal{U}\left(t;s\right)-\mathcal{U}^{(\mathbf{M})}\left(t;s\right)\right)\psi\right\rangle _{\mathcal{F}^{\otimes p}}\\
 & \qquad+\left\langle \left(\mathcal{U}\left(t;s\right)-\mathcal{U}^{(\mathbf{M})}\left(t;s\right)\right)\psi,\left(\mathcal{N}_{\mathrm{total}}\right)^{j}\mathcal{U}\left(t;s\right)\psi\right\rangle _{\mathcal{F}^{\otimes p}}\\
 & \qquad+\left\langle \mathcal{U}^{(\mathbf{M})}\left(t;s\right)\psi,\left(\mathcal{N}_{\mathrm{total}}\right)^{j}\mathcal{U}^{(\mathbf{M})}\left(t;s\right)\psi\right\rangle _{\mathcal{F}^{\otimes p}}\\
 & \leq Cp^{Cp}e^{Kt}\left\langle \psi,\left(\mathcal{N}_{\mathrm{total}}+p\mathcal{I}\right)^{2j+2}\psi\right\rangle _{\mathcal{F}^{\otimes p}}.
\end{align*}
This proves the desired lemma.
\end{proof}

\subsection{Comparison dynamics for others}
\begin{lem}
\label{lem:tildeNj} Suppose that the assumptions in Theorem \ref{thm:Main_Theorem}
hold. Let $\widetilde{\mathcal{U}}$ be the unitary operator defined
in (\ref{eq:def_mathcaltildeU}). Then, for any $\psi\in\mathcal{F}$
and $j\in\mathbb{N}$, there exist constants $C\equiv C(j)$ and $K\equiv K(j)$
such that 
\[
\left\langle \widetilde{\mathcal{U}}\left(t;s\right)\psi,\left(\mathcal{N}_{\mathrm{total}}\right)^{j}\widetilde{\mathcal{U}}\left(t;s\right)\psi\right\rangle _{\mathcal{F}^{\otimes p}}\leq Cp^{Cp}e^{K|t-s|}\left\langle \psi,\left(\mathcal{N}_{\mathrm{total}}+p\mathcal{I}\right)^{j}\psi\right\rangle _{\mathcal{F}^{\otimes p}}.
\]
\end{lem}

\begin{proof}
This will follow and use the proof of \cite{ChenLee2011}.

Let $\widetilde{\psi}=\widetilde{\mathcal{U}}(t;s)\psi$. We have
\begin{align*}
\frac{\mathrm{d}}{\mathrm{d}t}\left\langle \widetilde{\psi},\mathcal{N}_{\mathrm{total}}^{j}\widetilde{\psi}\right\rangle _{\mathcal{F}^{\otimes p}} & =\left\langle \widetilde{\psi},[\mathrm{i}(\mathcal{L}_{0}(t)+\mathcal{L}_{2}(t)+\mathcal{L}_{4}(t)),\mathcal{N}_{\mathrm{total}}^{j}]\widetilde{\psi}\right\rangle _{\mathcal{F}^{\otimes p}}\\
 & =\left\langle \widetilde{\psi},[\mathrm{i}(\mathcal{L}_{2}^{1}(t)+\mathcal{L}_{4}^{1}(t)),\mathcal{N}_{\mathrm{total}}^{j}]\widetilde{\psi}\right\rangle _{\mathcal{F}^{\otimes p}}\\
 & \qquad+\left\langle \widetilde{\psi},[\mathrm{i}(\mathcal{L}_{2}^{2}(t)+\mathcal{L}_{4}^{2}(t)),\mathcal{N}_{\mathrm{total}}^{j}]\widetilde{\psi}\right\rangle _{\mathcal{F}^{\otimes p}}\\
 & \qquad+\left\langle \widetilde{\psi},[\mathrm{i}(\mathcal{L}_{2}^{\mathrm{cross}}(t)+\mathcal{L}_{4}^{\mathrm{cross}}(t)),\mathcal{N}_{\mathrm{total}}^{j}\widetilde{\psi}\right\rangle _{\mathcal{F}^{\otimes p}}.
\end{align*}
To bound the list three terms, we investigate it further. It is enough
to bound
\begin{equation}
\begin{aligned}\frac{\mathrm{d}}{\mathrm{d}t}\left\langle \widetilde{\psi},(\mathcal{N}^{\mathfrak{J}})\widetilde{\psi}\right\rangle _{\mathcal{F}^{\otimes p}} & =\left\langle \widetilde{\psi},[\mathrm{i}(\mathcal{L}_{0}(t)+\mathcal{L}_{2}(t)+\mathcal{L}_{4}(t)),(\mathcal{N}^{\mathfrak{J}})]\widetilde{\psi}\right\rangle _{\mathcal{F}^{\otimes p}}\\
 & =\sum_{q=1}^{p}\left\langle \widetilde{\psi},[\mathrm{i}(\mathcal{L}_{2}^{q}(t)+\mathcal{L}_{4}^{q}(t)),(\mathcal{N}^{\mathfrak{J}})]\widetilde{\psi}\right\rangle _{\mathcal{F}^{\otimes p}}\\
 & \qquad+\left\langle \widetilde{\psi},[\mathrm{i}(\mathcal{L}_{2}^{\mathrm{cross}}(t)+\mathcal{L}_{4}^{\mathrm{cross}}(t)),(\mathcal{N}^{\mathfrak{J}})]\widetilde{\psi}\right\rangle _{\mathcal{F}^{\otimes p}}.
\end{aligned}
\label{eq:NjNk}
\end{equation}
Due to the symmetric structure, it is enough to show for the $q$-th
term for $q\leq p$ and the last term of (\ref{eq:NjNk}).

Then since $\left\langle \widetilde{\psi},[\mathrm{i}(\mathcal{L}_{2}^{q}(t)+\mathcal{L}_{4}^{q}(t)),\left(\mathcal{E}^{(r)}\left(\mathcal{N}+1\right)^{j}\right)]\widetilde{\psi}\right\rangle =0$
and $\mathcal{L}_{2}^{q}(t)+\mathcal{L}_{4}^{q}(t)$ does not change
the number of the particles of the second component,
\begin{align*}
\left\langle \widetilde{\psi},[\mathrm{i}(\mathcal{L}_{2}^{q}(t)+\mathcal{L}_{4}^{q}(t)),((\mathcal{N}+1)^{\mathfrak{J}})]\widetilde{\psi}\right\rangle _{\mathcal{F}^{\otimes p}} & =\left\langle \widetilde{\psi},[\mathrm{i}(\mathcal{L}_{2}^{q}(t)+\mathcal{L}_{4}^{q}(t)),\prod_{q=1}^{p}\left(\mathcal{E}^{(q)}\left(\mathcal{N}+1\right)^{j_{q}}\right)]\widetilde{\psi}\right\rangle _{\mathcal{F}^{\otimes p}}\\
 & \leq C\left\langle \widetilde{\psi},\prod_{q=1}^{p}\left(\mathcal{E}^{(q)}\left(\mathcal{N}+1\right)^{j_{q}}\right)\widetilde{\psi}\right\rangle _{\mathcal{F}^{\otimes p}}
\end{align*}
For the last term,
\[
\left\langle \widetilde{\psi},[\mathrm{i}(\mathcal{L}_{2}^{\mathrm{cross}}(t)+\mathcal{L}_{4}^{\mathrm{cross}}(t)),\left(\prod_{q=1}^{p}\left(\mathcal{E}^{(q)}\left(\mathcal{N}+1\right)^{j_{q}}\right)\right)]\widetilde{\psi}\right\rangle _{\mathcal{F}^{\otimes p}}\leq C\left\langle \widetilde{\psi},\left(\prod_{q=1}^{p}\left(\mathcal{E}^{(q)}\left(\mathcal{N}+1\right)^{j_{q}}\right)\right)\widetilde{\psi}\right\rangle _{\mathcal{F}^{\otimes p}}\quad
\]
Applying Grönwall lemma, we conclude that 
\[
\left\langle \widetilde{\mathcal{U}}\left(t;s\right)\psi,\left(\prod_{q=1}^{p}\left(\mathcal{E}^{(q)}\left(\mathcal{N}+1\right)^{j_{q}}\right)\right)\widetilde{\mathcal{U}}\left(t;s\right)\psi\right\rangle _{\mathcal{F}^{\otimes p}}\leq Ce^{K|t-s|}\left\langle \psi,\left(\prod_{q=1}^{p}\left(\mathcal{E}^{(q)}\left(\mathcal{N}+1\right)^{j_{q}}\right)\right)\psi\right\rangle _{\mathcal{F}^{\otimes p}}.
\]
Since the proof was symmetric up to the component of particle,%
\[
\left\langle \widetilde{\mathcal{U}}\left(t;s\right)\psi,\left(\mathcal{N}_{\mathrm{total}}+p\mathcal{I}\right)^{j}\widetilde{\mathcal{U}}\left(t;s\right)\psi\right\rangle _{\mathcal{F}^{\otimes p}}\leq Cp^{p}e^{K|t-s|}\left\langle \psi,\left(\mathcal{N}_{\mathrm{total}}+p\mathcal{I}\right)^{j}\psi\right\rangle _{\mathcal{F}^{\otimes p}}.
\]
{} which proves the desired lemma.
\end{proof}
The main difference between the unitary operators $\mathcal{U}$ and
$\widetilde{\mathcal{U}}$ comes from the generator $\mathcal{L}_{3}$.
In the following lemma, we find an estimate on $\mathcal{L}_{3}$.
\begin{lem}
\label{lem:N_1_L3} Suppose that the assumptions in Theorem \ref{thm:Main_Theorem}
hold. Then, for any $\psi\in\mathcal{F}$ and $j\in\mathbb{N}$, there
exist a constant $C\equiv C(j)$ such that 
\[
\left\Vert \left(\mathcal{N}_{\mathrm{total}}+p\mathcal{I}\right)^{j/2}\mathcal{L}_{3}(t)\psi\right\Vert _{\mathcal{F}^{\otimes p}}\leq Cp^{2}(\min_{q}c_{q})^{-1/2}\frac{1}{\sqrt{N}}\left\Vert \left(\mathcal{N}_{\mathrm{total}}+p\mathcal{I}\right)^{(j+3)/2}\psi\right\Vert _{\mathcal{F}^{\otimes p}}
\]
\end{lem}

\begin{proof}
We basically follow the proof in \cite[Lemma 5.3]{Lee2013}. For $q,r=1,\dots,p$,
let 
\[
A_{3}^{(q,r)}(t)=\int\mathrm{d}x\mathrm{d}y\,V_{qr}(x-y)\overline{\varphi_{t}(y)}a_{x}^{(q)*}a_{y}^{(r)}a_{x}^{(q)}.
\]
Then, by definition, we have
\[
\begin{split}\mathcal{L}_{3}(t) & =\sum_{q=1}^{p}\Big(\frac{1}{\sqrt{N_{q}}}\int\mathrm{d}x\mathrm{d}z\,V_{qq}(x-z)a_{x}^{(q)*}(u_{q,t}(z)a_{z}^{(q)*}+\overline{u_{q,t}}(z)a_{z}^{(q)})a_{x}^{(q)}\Big)\\
 & \quad+\sum_{q\neq r}\left(\frac{\sqrt{N_{q}}}{N}\int\mathrm{d}x\mathrm{d}y\,V_{qr}(x-y)a_{y}^{(q)*}(u_{q,t}(x)a_{x}^{(q)*}+\overline{u_{q,t}}(x)a_{x}^{(q)})a_{y}^{(r)}\right)\\
 & \quad=:\sum_{q=1}^{p}\mathcal{L}_{3}^{q}(t)+\mathcal{L}_{3}^{\text{cross}}(t).
\end{split}
\]
Then we get
\begin{align}
\left(\mathcal{N}_{\mathrm{total}}+p\mathcal{I}\right)^{j/2}\mathcal{L}_{3}^{q}(t) & =\frac{1}{\sqrt{N_{q}}}\left(\left(\mathcal{N}_{\mathrm{total}}+p\mathcal{I}\right)^{j/2}A_{3}^{(q,q)}(t)+\left(\mathcal{N}_{\mathrm{total}}+p\mathcal{I}\right)^{j/2}A_{3}^{(q,q)*}(t)\right),\label{eq:L3}\\
\intertext{and}\left(\mathcal{N}_{\mathrm{total}}+p\mathcal{I}\right)^{j/2}\mathcal{L}_{3}^{\mathrm{cross}}(t) & =\sum_{q\neq r}\frac{\sqrt{N_{q}}}{N}\left(\left(\mathcal{N}_{\mathrm{total}}+p\mathcal{I}\right)^{j/2}A_{3}^{(q,r)}(t)+\left(\mathcal{N}_{\mathrm{total}}+p\mathcal{I}\right)^{j/2}A_{3}^{(q,r)*}(t)\right)\label{eq:L3-cross}
\end{align}
The first term in (\ref{eq:L3}), $\left(\mathcal{N}_{\mathrm{total}}+p\mathcal{I}\right)^{j/2}A_{3}^{(q,q)}(t)$,
satisfies for any $\xi\in\mathcal{F}^{\otimes p}$ that 
\begin{align*}
 & \left|\langle\xi,\left(\mathcal{N}_{\mathrm{total}}+p\mathcal{I}\right)^{j/2}A_{3}^{(q,q)}(t)\psi\rangle_{\mathcal{F}^{\otimes p}}\right|\\
 & \quad=\left|\int\mathrm{d}x\mathrm{d}y\,V_{qq}(x-y)\overline{u_{q,t}(y)}\langle\xi,\left(\mathcal{N}_{\mathrm{total}}+p\mathcal{I}\right)^{j/2}a_{x}^{(q)*}a_{y}^{(q)}a_{x}^{(q)}\psi\rangle_{\mathcal{F}^{\otimes p}}\right|\\
 & \quad=\left|\int\mathrm{d}x\mathrm{d}y\,V_{qq}(x-y)\overline{u_{q,t}(y)}\langle\left(\mathcal{N}_{\mathrm{total}}+p\mathcal{I}\right)^{-1/2}\xi,(\mathcal{N}_{\mathrm{total}}+p\mathcal{I})^{(j+1)/2}a_{x}^{(q)*}a_{y}^{(q)}a_{x}^{(q)}\psi\rangle_{\mathcal{F}^{\otimes p}}\right|.
\end{align*}
This leads that
\begin{align*}
\left|\langle\xi,\left(\mathcal{N}_{\mathrm{total}}+p\mathcal{I}\right)^{j/2}A_{3}^{(q,q)}(t)\psi\rangle_{\mathcal{F}^{\otimes p}}\right| & \leq\left(\int\mathrm{d}x\mathrm{d}y\,|V(x-y)|^{2}|u_{q,t}(y)|^{2}\|a_{x}^{(q)}\left((\mathcal{N}_{\mathrm{total}}+p\mathcal{I}\right)^{-1/2}\xi\|_{\mathcal{F}^{\otimes p}}^{2}\right)^{1/2}\\
 & \qquad\times\left(\int\mathrm{d}x\mathrm{d}y\|a_{y}^{(q)}a_{x}^{(q)}\left(\mathcal{N}_{\mathrm{total}}+p\mathcal{I}\right)^{(j+1)/2}\psi\|_{\mathcal{F}^{\otimes p}}^{2}\right)^{1/2}.
\end{align*}
Hence, 
\[
\left|\langle\xi,\left(\mathcal{N}_{\mathrm{total}}+p\mathcal{I}\right)^{j/2}A_{3}^{(q,q)}(t)\psi\rangle_{\mathcal{F}^{\otimes p}}\right|\leq C\sup_{x}\|\xi\|_{\mathcal{F}^{\otimes p}}\|\left(\mathcal{N}_{\mathrm{total}}+p\mathcal{I}\right)^{(j+3)/2}\psi\|_{\mathcal{F}^{\otimes p}}.\quad
\]
Since $\xi$ was arbitrary, by applying the assumption on interaction
potentials $V_{qr}$ in (\ref{eq:assumption_V}), we obtain that 
\begin{equation}
\|\left(\mathcal{N}_{\mathrm{total}}+p\mathcal{I}\right)^{j/2}A_{3}^{(q,q)}(t)\psi\|_{\mathcal{F}^{\otimes p}}\leq C\|\left(\mathcal{N}_{\mathrm{total}}+p\mathcal{I}\right)^{(j+3)/2}\psi\|_{\mathcal{F}^{\otimes p}}.\label{eq:(N+pI)A3qq}
\end{equation}
For the second term of (\ref{eq:L3}), using similar argument,
\begin{equation}
\|\left(\mathcal{N}_{\mathrm{total}}+p\mathcal{I}\right)^{j/2}A_{3}^{(q,q)*}(t)\psi\|_{\mathcal{F}^{\otimes p}}\leq C\|\left(\mathcal{N}_{\mathrm{total}}+p\mathcal{I}\right)^{(j+3)/2}\psi\|_{\mathcal{F}^{\otimes p}}.\label{eq:(N+pI)A3qqstar}
\end{equation}

For (\ref{eq:L3-cross}), similarly we get
\begin{equation}
\|\left(\mathcal{N}_{\mathrm{total}}+p\mathcal{I}\right)^{j/2}A_{3}^{(q,r)}\psi\|_{\mathcal{F}^{\otimes p}}\leq C\|\left(\mathcal{N}_{\mathrm{total}}+p\mathcal{I}\right)^{(j+3)/2}\psi\|_{\mathcal{F}^{\otimes p}},\label{eq:(N+pI)A3qr}
\end{equation}
and
\begin{equation}
\|\left(\mathcal{N}_{\mathrm{total}}+p\mathcal{I}\right)^{j/2}A_{3}^{(q,r)*}\psi\|_{\mathcal{F}^{\otimes p}}\leq C\|\left(\mathcal{N}_{\mathrm{total}}+p\mathcal{I}\right)^{(j+3)/2}\psi\|_{\mathcal{F}^{\otimes p}}.\label{eq:(N+pI)A3qrstar}
\end{equation}
Hence, from (\ref{eq:(N+pI)A3qq}), (\ref{eq:(N+pI)A3qqstar}), (\ref{eq:(N+pI)A3qr}),
and (\ref{eq:(N+pI)A3qrstar}), we get 
\begin{align*}
\|\left(\mathcal{N}_{\mathrm{total}}+p\mathcal{I}^{j/2}\right)\mathcal{L}_{3}(t)\psi\|_{\mathcal{F}^{\otimes p}} & \quad\leq C\left(\sum_{q=1}^{p}\frac{1}{\sqrt{N_{q}}}+\sum_{q\neq r}\frac{\sqrt{N_{q}}}{N}\right)\|\left(\mathcal{N}_{\mathrm{total}}+p\mathcal{I}\right)^{(j+3)/2}\psi\|_{\mathcal{F}^{\otimes p}}\\
 & \quad\leq C\left(\sum_{q,r}\frac{1}{\sqrt{N_{q}}}\right)\|\left(\mathcal{N}_{\mathrm{total}}+p\mathcal{I}\right)^{(j+3)/2}\psi\|_{\mathcal{F}^{\otimes p}}\\
 & \quad\leq C\left(\frac{p^{2}}{(\min_{q}c_{q})^{-1/2}\sqrt{N}}\right)\|\left(\mathcal{N}_{\mathrm{total}}+p\mathcal{I}\right)^{(j+3)/2}\psi\|_{\mathcal{F}^{\otimes p}}
\end{align*}
which was to be proved.
\end{proof}
Finally, we prove the following lemma on the difference between $\mathcal{U}$
and $\widetilde{\mathcal{U}}$.
\begin{lem}
\label{lem:NjUphiUtildeUphitildeU} Suppose that the assumptions in
Theorem \ref{thm:Main_Theorem} hold. Then, for all $j\in\mathbb{N}$,
there exist constants $C\equiv C(j)$ and $K\equiv K(t)$ such that,
for any $\mathbf{f}=(f_{1},\dots,f_{p})$ with $f_{q}\in L^{2}\left(\mathbb{R}^{3}\right)$
for all $q=1,\dots,p$, 
\begin{align*}
 & \left\Vert \left(\mathcal{N}_{\mathrm{total}}+p\mathcal{I}\right)^{j/2}\left(\mathcal{U}^{*}(t)\Phi(\mathbf{f})\mathcal{U}(t)-\widetilde{\mathcal{U}}^{*}(t)\Phi(\mathbf{f})\widetilde{\mathcal{U}}(t)\right)\omega\right\Vert _{\mathcal{F}^{\otimes p}}\\
 & \quad\leq Ce^{Kt}(\min_{q}c_{q})^{-1/2}p^{j+Cp}\frac{1}{\sqrt{N}}.
\end{align*}
\end{lem}

\begin{proof}
We follow the proof in \cite[Lemma 5.4]{Lee2013}. Let 
\[
\mathcal{R}_{1}(\mathbf{f}):=\left(\mathcal{U}^{*}(t)-\widetilde{\mathcal{U}}^{*}(t)\right)\Phi(\mathbf{f})\widetilde{\mathcal{U}}(t)
\]
and 
\[
\mathcal{R}_{2}(\mathbf{f}):=\mathcal{U}^{*}(t)\Phi(\mathbf{f})\left(\mathcal{U}(t)-\widetilde{\mathcal{U}}(t)\right)
\]
so that 
\begin{equation}
\mathcal{U}^{*}(t)\Phi(\mathbf{f})\mathcal{U}(t)-\widetilde{\mathcal{U}}^{*}(t)\Phi(\mathbf{f})\widetilde{\mathcal{U}}(t)=\mathcal{R}_{1}(\mathbf{f})+\mathcal{R}_{2}(\mathbf{f}).\label{eq:R_1}
\end{equation}
We begin by estimating the first term in the right-hand side of (\ref{eq:R_1}).
From Lemma \ref{lem:NjU}, 
\begin{align*}
\left\Vert \left(\mathcal{N}_{\mathrm{total}}+p\mathcal{I}\right)^{j/2}\mathcal{R}_{1}(\mathbf{f})\omega\right\Vert _{\mathcal{F}^{\otimes p}} & =\left\Vert \int_{0}^{t}\mathrm{d}s\left(\mathcal{N}_{\mathrm{total}}+p\mathcal{I}\right)^{j/2}\mathcal{U}^{*}(s;0)\mathcal{L}_{3}(s)\widetilde{\mathcal{U}}^{*}(t;s)\Phi(\mathbf{f})\widetilde{\mathcal{U}}(t)\omega\right\Vert _{\mathcal{F}^{\otimes p}}\\
 & \leq\int_{0}^{t}\mathrm{d}s\left\Vert \left(\mathcal{N}_{\mathrm{total}}+p\mathcal{I}\right)^{j/2}\mathcal{U}^{*}(s;0)\mathcal{L}_{3}(s)\widetilde{\mathcal{U}}^{*}(t;s)\Phi(\mathbf{f})\widetilde{\mathcal{U}}(t)\omega\right\Vert _{\mathcal{F}^{\otimes p}}\\
 & \leq Cp^{Cp}e^{Kt}\int_{0}^{t}\mathrm{d}s\left\Vert \left(\mathcal{N}_{\mathrm{total}}+p\mathcal{I}\right)^{j+1}\mathcal{L}_{3}(s)\widetilde{\mathcal{U}}^{*}(t;s)\Phi(\mathbf{f})\widetilde{\mathcal{U}}(t)\omega\right\Vert _{\mathcal{F}^{\otimes p}}.
\end{align*}
From Lemma \ref{lem:N_1_L3} and the assumption on interaction potentials
$V_{qr}$ in (\ref{eq:assumption_V}), we get
\begin{align*}
\left\Vert \left(\mathcal{N}_{\mathrm{total}}+p\mathcal{I}\right)^{j/2}\mathcal{R}_{1}(\mathbf{f})\omega\right\Vert _{\mathcal{F}\otimes\mathcal{F}} & \leq Cp^{Cp}e^{Kt}\left(\sum_{q=1}^{p}\left(\frac{1}{\sqrt{N_{q}}}+\frac{\sqrt{N_{q}}}{N}\right)\right)\int_{0}^{t}\mathrm{d}s\left\Vert \left(\mathcal{N}_{\mathrm{total}}+p\mathcal{I}\right)^{j+(5/2)}\widetilde{\mathcal{U}}^{*}(t;s)\Phi(\mathbf{f})\widetilde{\mathcal{U}}(t)\omega\right\Vert _{\mathcal{F}^{\otimes p}}
\end{align*}
By Lemma \ref{lem:tildeNj},
\begin{align*}
\left\Vert \left(\mathcal{N}_{\mathrm{total}}+p\mathcal{I}\right)^{j/2}\mathcal{R}_{1}(f)\omega\right\Vert _{\mathcal{F}^{\otimes p}} & \leq Cp^{Cp}e^{Kt}\left(\sum_{q=1}^{p}\left(\frac{1}{\sqrt{N_{q}}}+\frac{\sqrt{N_{q}}}{N}\right)\right)\int_{0}^{t}\mathrm{d}s\left\Vert \left(\mathcal{N}_{\mathrm{total}}+p\mathcal{I}\right)^{j+(5/2)}\Phi(\mathbf{f})\widetilde{\mathcal{U}}(t)\omega\right\Vert _{\mathcal{F}^{\otimes p}}.
\end{align*}
Using the fact that the right hand side does not depend on time $s$
and Lemma \ref{lem:bound-a-a*-Phi-dGamma},
\begin{align*}
\left\Vert \left(\mathcal{N}_{\mathrm{total}}+p\mathcal{I}\right)^{j/2}\mathcal{R}_{1}(\mathbf{f})\omega\right\Vert _{\mathcal{F}^{\otimes p}} & \leq Cp^{Cp}e^{Kt}\left(\sum_{q=1}^{p}\left(\frac{1}{\sqrt{N_{q}}}+\frac{\sqrt{N_{q}}}{N}\right)\right)\left\Vert \left(\mathcal{N}_{\mathrm{total}}+p\mathcal{I}\right)^{j+(5/2)}\Phi(\mathbf{f})\widetilde{\mathcal{U}}(t)\omega\right\Vert _{\mathcal{F}^{\otimes p}}\\
 & \leq Cp^{Cp}e^{Kt}\left(\sum_{q=1}^{p}\left(\frac{1}{\sqrt{N_{q}}}+\frac{\sqrt{N_{q}}}{N}\right)\right)\left\Vert \left(\mathcal{N}_{\mathrm{total}}+p\mathcal{I}\right)^{j+(5/2)+(p/2)}\widetilde{\mathcal{U}}(t)\omega\right\Vert _{\mathcal{F}^{\otimes p}}
\end{align*}
Thus, from Lemma \ref{lem:tildeNj}, we obtain for $\mathcal{R}_{1}(\mathbf{f})$
that 
\begin{align*}
\left\Vert \left(\mathcal{N}_{\mathrm{total}}+p\mathcal{I}\right)^{j/2}\mathcal{R}_{1}(\mathbf{f})\omega\right\Vert _{\mathcal{F}^{\otimes p}} & \leq Cp^{Cp}e^{Kt}\left(\sum_{q=1}^{p}\left(\frac{1}{\sqrt{N_{q}}}+\frac{\sqrt{N_{q}}}{N}\right)\right)\left\Vert \left(\mathcal{N}_{\mathrm{total}}+p\mathcal{I}\right)^{j+(5/2)+(p/2)}\omega\right\Vert _{\mathcal{F}^{\otimes p}}
\end{align*}
The study of $\mathcal{R}_{2}(\mathbf{f})$ is similar and we can
again obtain that 
\begin{align*}
\left\Vert \left(\mathcal{N}_{\mathrm{total}}+p\mathcal{I}\right)^{j/2}\mathcal{R}_{2}(\mathbf{f})\omega\right\Vert _{\mathcal{F}^{\otimes p}} & \leq Cp^{Cp}e^{Kt}\left(\sum_{q=1}^{p}\left(\frac{1}{\sqrt{N_{q}}}+\frac{\sqrt{N_{q}}}{N}\right)\right)\left\Vert \left(\mathcal{N}_{\mathrm{total}}+p\mathcal{I}\right)^{j+(5/2)+(p/2)}\omega\right\Vert _{\mathcal{F}^{\otimes p}}
\end{align*}
Noting that $N_{q}$ are of order $N$, 
\[
\frac{c}{\sqrt{N_{q}}}\leq\frac{\sqrt{N_{q}}}{N}\leq\frac{C}{\sqrt{N_{q}}}
\]
for some $c,C>0$. Hence we have both
\[
\left\Vert \left(\mathcal{N}_{\mathrm{total}}+p\mathcal{I}\right)^{j/2}\mathcal{R}_{1}(\mathbf{f})\omega\right\Vert _{\mathcal{F}^{\otimes p}}\leq C(\min_{q}c_{q})^{-1/2}e^{Kt}p^{j+Cp}\frac{1}{\sqrt{N}}
\]
and
\[
\left\Vert \left(\mathcal{N}_{\mathrm{total}}+p\mathcal{I}\right)^{j/2}\mathcal{R}_{2}(\mathbf{f})\omega\right\Vert _{\mathcal{F}^{\otimes p}}\leq C(\min_{q}c_{q})^{-1/2}e^{Kt}p^{j+Cp}\frac{1}{\sqrt{N}}
\]
This completes the proof of the desired lemma.
\end{proof}

\section{Proof of Propositions \label{sec:Pf-of-Props}}

\subsection{Proof of Propositions \ref{prop:Et1}}

Noting that for disjoint subsets $Q$ and $R$ of $\{1,\dots,p\}$,

\begin{align*}
F_{t}^{Q,R}(\mathbf{x};\mathbf{x}') & :=\sqrt{\boldsymbol{N}_{Q^{c}}}\sqrt{\boldsymbol{N}_{R^{c}}}\,\frac{\boldsymbol{d}_{\boldsymbol{N}}}{\boldsymbol{N}}\,u_{Q^{c},t}(x_{Q^{c}})\,\overline{u_{R^{c},t}(x_{R^{c}}')}\left\langle \frac{\left(\boldsymbol{a}^{*}(\mathbf{u})\right)^{\boldsymbol{N}}}{\sqrt{\boldsymbol{N}!}}\omega,\mathcal{W}(\mathbf{N}^{\odot1/2}\odot\mathbf{u}_{s})\left(a_{\mathbf{x}}^{(Q)*}a_{\mathbf{x}'}^{(R)}\right)\mathcal{U}(t)\omega\right\rangle _{\mathcal{F}^{\otimes p}}\\
\intertext{and}E_{t}^{Q,R}(J) & :=\left\langle u_{(Q\cup R)^{c},t}\right|J_{(Q\cup R)^{c}}\left|u_{(Q\cup R)^{c},t}\right\rangle \sqrt{\boldsymbol{N}_{Q^{c}}}\sqrt{\boldsymbol{N}_{R^{c}}}\frac{\boldsymbol{d}_{\boldsymbol{N}}}{\boldsymbol{N}}\\
 & \qquad\qquad\times\left\langle \frac{\left(\boldsymbol{a}^{*}(\mathbf{u})\right)^{\boldsymbol{N}}}{\sqrt{\boldsymbol{N}!}}\omega,\mathcal{W}(\mathbf{N}^{\odot1/2}\odot\mathbf{u}_{s})\left(d\boldsymbol{\Gamma}^{(Q\cap R)}(J_{Q\cap R})\cdot\boldsymbol{\phi}^{(Q\Delta R)}(J_{Q\Delta R}u_{Q\Delta R,t})\right)\mathcal{U}(t)\omega\right\rangle _{\mathcal{F}^{\otimes p}},
\end{align*}
in this section, we prove proposition \ref{prop:Et1} by applying
the lemmas proved in the following section.
\begin{proof}[Proof of Proposition \ref{prop:Et1}]
For the proof of Proposition \ref{prop:Et1}, we use strategy provided
in \cite[Lemma 4.2]{Lee2013}. Recall that 
\begin{align*}
 & E_{t}^{Q,R}(J)\\
 & =\left\langle u_{(Q\cup R)^{c},t}\right|J_{(Q\cup R)^{c}}\left|u_{(Q\cup R)^{c},t}\right\rangle \sqrt{\boldsymbol{N}_{Q^{c}}}\sqrt{\boldsymbol{N}_{R^{c}}}\frac{\boldsymbol{d}_{\boldsymbol{N}}}{\boldsymbol{N}}\\
 & \qquad\times\left\langle \frac{\left(\boldsymbol{a}^{*}(\mathbf{u})\right)^{\boldsymbol{N}}}{\sqrt{\boldsymbol{N}!}}\omega,\mathcal{W}(\mathbf{N}^{\odot1/2}\odot\mathbf{u}_{s})\mathcal{U}^{*}(t)\left(d\boldsymbol{\Gamma}^{(Q\cap R)}(J_{Q\cap R})\cdot\boldsymbol{\phi}^{(Q\Delta R)}(J_{Q\Delta R}u_{Q\Delta R,t})\right)\mathcal{U}(t)\omega\right\rangle _{\mathcal{F}^{\otimes p}}.
\end{align*}
Then also recall the definitions of $\mathcal{R}_{1}$ and $\mathcal{R}_{2}$
in the proof of Lemma \ref{lem:NjUphiUtildeUphitildeU}. Let $\mathcal{R}(\mathbf{f})=\mathcal{R}_{1}(\mathbf{f})+\mathcal{R}_{2}(\mathbf{f})$
so that 
\[
\mathcal{R}(\mathbf{f})=\mathcal{U}^{*}(t)\boldsymbol{\phi}^{(Q)}(\mathbf{f})\mathcal{U}(t)-\widetilde{\mathcal{U}}^{*}(t)\boldsymbol{\phi}^{(Q)}(\mathbf{f})\widetilde{\mathcal{U}}(t).
\]
From the parity conservation (\ref{eq:Parity_Consevation}), 
\[
P_{2\mathbf{k}}\widetilde{\mathcal{U}}^{*}(t)\boldsymbol{\phi}^{(Q)}(Ju_{Q,t})\widetilde{\mathcal{U}}(t)\omega=0
\]
for all $p$-dimensional vector $\mathbf{k}\in\mathbb{Z}_{\geq0}^{d}$
with nonnegative integer components. (See \cite[Lemma 8.2]{Lee2013}
for more detail.) Thus,
\begin{align*}
 & \left\langle u_{(Q\cup R)^{c},t}\right|J_{(Q\cup R)^{c}}\left|u_{(Q\cup R)^{c},t}\right\rangle \sqrt{\boldsymbol{N}_{Q^{c}}}\sqrt{\boldsymbol{N}_{R^{c}}}\frac{\boldsymbol{d}_{\boldsymbol{N}}}{\boldsymbol{N}}\\
 & \times\left|\left\langle \frac{\left(\boldsymbol{a}^{*}(\mathbf{u})\right)^{\boldsymbol{N}}}{\sqrt{\boldsymbol{N}!}}\omega,\mathcal{W}(\mathbf{N}^{\odot1/2}\odot\mathbf{u}_{s})\mathcal{U}^{*}(t)\left(d\boldsymbol{\Gamma}^{(Q\cap R)}(J_{Q\cap R})\cdot\boldsymbol{\phi}^{(Q\Delta R)}(J_{Q\Delta R}u_{Q\Delta R,t})\right)\mathcal{U}(t)\omega\right\rangle _{\mathcal{F}^{\otimes p}}\right|\\
 & \leq\left\langle u_{(Q\cup R)^{c},t}\right|J_{(Q\cup R)^{c}}\left|u_{(Q\cup R)^{c},t}\right\rangle \sqrt{\boldsymbol{N}_{Q^{c}}}\sqrt{\boldsymbol{N}_{R^{c}}}\frac{\boldsymbol{d}_{\boldsymbol{N}}}{\boldsymbol{N}}\\
 & \qquad\times\left\langle \frac{\left(\boldsymbol{a}^{*}(\mathbf{u})\right)^{\boldsymbol{N}}}{\sqrt{\boldsymbol{N}!}}\omega,\mathcal{W}(\mathbf{N}^{\odot1/2}\odot\mathbf{u}_{s})\mathcal{\widetilde{\mathcal{U}}}^{*}(t)\left(d\boldsymbol{\Gamma}^{(Q\cap R)}(J_{Q\cap R})\cdot\boldsymbol{\phi}^{(Q\Delta R)}(J_{Q\Delta R}u_{Q\Delta R,t})\right)\mathcal{\widetilde{\mathcal{U}}}(t)\omega\right\rangle _{\mathcal{F}^{\otimes p}}\\
 & \qquad+\left\langle u_{(Q\cup R)^{c},t}\right|J_{(Q\cup R)^{c}}\left|u_{(Q\cup R)^{c},t}\right\rangle \sqrt{\boldsymbol{N}_{Q^{c}}}\sqrt{\boldsymbol{N}_{R^{c}}}\frac{\boldsymbol{d}_{\boldsymbol{N}}}{\boldsymbol{N}}\\
 & \qquad\qquad\times\left\langle \frac{\left(\boldsymbol{a}^{*}(\mathbf{u})\right)^{\boldsymbol{N}}}{\sqrt{\boldsymbol{N}!}}\omega,\mathcal{W}(\mathbf{N}^{\odot1/2}\odot\mathbf{u}_{s})\left(d\boldsymbol{\Gamma}^{(Q\cap R)}(J_{Q\cap R})\cdot\mathcal{R}^{(Q\Delta R)}(J_{Q\Delta R}u_{Q\Delta R,t})\right)\omega\right\rangle _{\mathcal{F}^{\otimes p}}\\
 & =:I+II
\end{align*}
Then
\begin{align*}
I & \leq\left\langle u_{(Q\cup R)^{c},t}\right|J_{(Q\cup R)^{c}}\left|u_{(Q\cup R)^{c},t}\right\rangle \sqrt{\boldsymbol{N}_{Q^{c}}}\sqrt{\boldsymbol{N}_{R^{c}}}\frac{\boldsymbol{d}_{\boldsymbol{N}}}{\boldsymbol{N}}\\
 & \qquad\times\left\Vert \left(\mathcal{N}_{\mathrm{total}}+p\mathcal{I}\right)^{-1/2}\mathcal{W}^{*}(\mathbf{N}^{\odot1/2}\odot\mathbf{u})\frac{\left(\boldsymbol{a}^{*}(\mathbf{u})\right)^{\boldsymbol{N}}}{\sqrt{\boldsymbol{N}!}}\omega\right\Vert _{\mathcal{F}^{\otimes p}}\\
 & \qquad\qquad\qquad\times\left\Vert \left(\mathcal{N}_{\mathrm{total}}+p\mathcal{I}\right)^{-1/2}\mathcal{\widetilde{\mathcal{U}}}^{*}(t)\left(d\boldsymbol{\Gamma}^{(Q\cap R)}(J_{Q\cap R})\cdot\boldsymbol{\phi}^{(Q\Delta R)}(J_{Q\Delta R}u_{Q\Delta R,t})\right)\mathcal{\widetilde{\mathcal{U}}}(t)\omega\right\Vert _{\mathcal{F}^{\otimes p}}
\end{align*}
The first factor is easily bounded by using
\begin{align*}
\frac{\left(\boldsymbol{a}^{*}(\mathbf{u})\right)^{\boldsymbol{N}}}{\sqrt{\boldsymbol{N}!}}\omega & =\left(\bigotimes_{q=1}^{p}d_{N_{q}}\mathcal{P}_{N_{q}}\right)\mathcal{W}(\mathbf{N}^{\odot1/2}\odot\mathbf{u}_{s})\,\omega
\end{align*}
that
\begin{equation}
\left\Vert \left(\mathcal{N}_{\mathrm{total}}+p\mathcal{I}\right)^{-1/2}\mathcal{W}^{*}(\mathbf{N}^{\odot1/2}\odot\mathbf{u})\frac{\left(\boldsymbol{a}^{*}(\mathbf{u})\right)^{\boldsymbol{N}}}{\sqrt{\boldsymbol{N}!}}\omega\right\Vert _{\mathcal{F}^{\otimes p}}\leq(N+p)^{-p/2}\boldsymbol{d}_{\boldsymbol{N}}\leq CN^{-p/4}.\label{eq:(Ntot+pI)WaNoemega-bound}
\end{equation}

\textbf{Case 1. $|Q|+|R|\geq2$}

The second factor is bounded by using Lemma \ref{lem:dGammaQphiQ},
\begin{align*}
 & \left\Vert \left(\mathcal{N}_{\mathrm{total}}+p\mathcal{I}\right)^{\frac{1}{2}}\widetilde{\mathcal{U}}^{*}(t)\left(d\boldsymbol{\Gamma}^{(Q\cap R)}(J_{Q\cap R})\cdot\boldsymbol{\phi}^{(Q\Delta R)}(J_{Q\Delta R}u_{Q\Delta R,t})\right)\widetilde{\mathcal{U}}(t)\omega\right\Vert _{\mathcal{F}^{\otimes p}}\\
 & \leq Ce^{Kt}\left\Vert \left(\mathcal{N}_{\mathrm{total}}+p\mathcal{I}\right)^{2}\left(d\boldsymbol{\Gamma}^{(Q\cap R)}(J_{Q\cap R})\cdot\boldsymbol{\phi}^{(Q\Delta R)}(J_{Q\Delta R}u_{Q\Delta R,t})\right)\widetilde{\mathcal{U}}(t)\omega\right\Vert _{\mathcal{F}^{\otimes p}}\\
 & \leq C\left\Vert J^{(Q\cap R)}\right\Vert _{\mathrm{HS}}e^{Kt}\left\Vert \left(\mathcal{N}_{\mathrm{total}}+p\mathcal{I}\right)^{2+|Q\cap R|}\boldsymbol{\phi}^{(Q\cap R)}(J_{Q\Delta R}u_{Q\Delta R,t})\widetilde{\mathcal{U}}(t)\omega\right\Vert _{\mathcal{F}^{\otimes p}}
\end{align*}
By Lemma \ref{lem:dGammaQphiQ},
\begin{align*}
 & \left\Vert \left(\mathcal{N}_{\mathrm{total}}+p\mathcal{I}\right)^{\frac{1}{2}}\widetilde{\mathcal{U}}^{*}(t)\left(d\boldsymbol{\Gamma}^{(Q\cap R)}(J_{Q\cap R})\cdot\boldsymbol{\phi}^{(Q\Delta R)}(J_{Q\Delta R}u_{Q\Delta R,t})\right)\widetilde{\mathcal{U}}(t)\omega\right\Vert _{\mathcal{F}^{\otimes p}}\\
 & \leq C\left\Vert J^{(Q\cup R)}\right\Vert _{\mathrm{HS}}e^{Kt}\left\Vert \left(\mathcal{N}_{\mathrm{total}}+p\mathcal{I}\right)^{2+|Q\cap R|+|Q\Delta R|}\widetilde{\mathcal{U}}(t)\omega\right\Vert _{\mathcal{F}^{\otimes p}}\\
 & \leq C\left\Vert J^{(Q\cup R)}\right\Vert _{\mathrm{HS}}e^{Kt}\left\Vert \left(\mathcal{N}_{\mathrm{total}}+p\mathcal{I}\right)^{4+2|Q\cup R|}\omega\right\Vert _{\mathcal{F}^{\otimes p}}\\
 & \leq C\left\Vert J^{(Q\cup R)}\right\Vert _{\mathrm{HS}}e^{Kt}p^{4+2|Q\cup R|}.
\end{align*}
Thus, using \textbf{$|Q|+|R|\geq2$,}
\[
I\leq CN^{-(|Q|+|R|)/2}\left\Vert J\right\Vert _{\mathrm{HS}}e^{Kt}p^{4+2|Q\cup R|}\leq CN^{-1}\left\Vert J\right\Vert _{\mathrm{HS}}e^{Kt}p^{4+Cp}
\]

This case can be easily bounded. First, using Cauchy-Schwarz inequality,
\begin{align*}
II & \leq\left\langle u_{(Q\cup R)^{c},t}\right|J_{(Q\cup R)^{c}}\left|u_{(Q\cup R)^{c},t}\right\rangle \sqrt{\boldsymbol{N}_{Q^{c}}}\sqrt{\boldsymbol{N}_{R^{c}}}\frac{\boldsymbol{d}_{\boldsymbol{N}}}{\boldsymbol{N}}\\
 & \qquad\qquad\times\left\Vert \left(\mathcal{N}_{\mathrm{total}}+p\mathcal{I}\right)^{-1/2}\mathcal{W}^{*}(\mathbf{N}^{\odot1/2}\odot\mathbf{u})\frac{\left(\boldsymbol{a}^{*}(\mathbf{u})\right)^{\boldsymbol{N}}}{\sqrt{\boldsymbol{N}!}}\omega\right\Vert _{\mathcal{F}^{\otimes p}}\\
 & \qquad\qquad\qquad\times\left\Vert \left(\mathcal{N}_{\mathrm{total}}+p\mathcal{I}\right)^{1/2}\left(d\boldsymbol{\Gamma}^{(Q\cap R)}(J_{Q\cap R})\cdot\mathcal{R}^{(Q\Delta R)}(J_{Q\Delta R}\odot u_{Q\Delta R,t})\right)\omega\right\Vert _{\mathcal{F}^{\otimes p}}.
\end{align*}
Using (\ref{eq:(Ntot+pI)WaNoemega-bound}), Lemma \ref{lem:coherent_all},
\ref{lem:dGammaQphiQ}, and \ref{lem:NjUphiUtildeUphitildeU}, we
have
\begin{align*}
II & \leq CN^{-(1+|Q|+|R)/2}\left\Vert J^{(Q\cup R)}\right\Vert _{\mathrm{op}}e^{Kt}p^{4+2|Q\cup R|}\\
 & \leq CN^{-1}\left\Vert J^{(Q\cup R)}\right\Vert _{\mathrm{HS}}e^{Kt}p^{Cp}\leq CN^{-1}\left\Vert J\right\Vert _{\mathrm{HS}}e^{Kt}p^{Cp}.
\end{align*}

It follows from Lemma Lemma \ref{lem:dGammaQphiQ} that
\[
|E_{t}^{Q,R}(J)|\leq C\left\Vert J\right\Vert _{\mathrm{HS}}p^{Cp}N^{-1}
\]
which is the concludes the case for $|Q|+|R|\geq2$.

\textbf{Case 2. }$|Q|+|R|=1$, i.e., one of $Q$ or $R$ is empty
and the other has only one element.

For this case either $Q$ or $R$ is empty and the other set has only
one element. Hence, $Q\cap R=\emptyset$ and $Q\Delta R$ has also
only one element. Let us denote the element $q$, i.e., $\{q\}=Q\Delta R$.
Using Lemma \ref{lem:coherent_even_odd}, letting 
\[
\boldsymbol{\ell}=(\ell_{1},\dots,\ell_{p})
\]
with $\ell_{q}=2k_{q}-1$ and the other sectors $\ell_{r}=0$ for
all $r\neq q$. Due to the parity conservation of $\widetilde{\mathcal{U}}$,
only the $q$-th component particle occupies odd number sectors. The
other components are only occupying the zeroth sector because $\phi^{(q)}(J_{q}u_{q,t})$
is acting as an identity for the $r$-th sector for $r\neq q$ and
$\widetilde{\mathcal{U}}(t)$ is unitary. This leads that

\begin{align*}
 & \left\langle u_{(Q\cup R)^{c},t}\right|J_{(Q\cup R)^{c}}\left|u_{(Q\cup R)^{c},t}\right\rangle \sqrt{\boldsymbol{N}_{Q^{c}}}\sqrt{\boldsymbol{N}_{R^{c}}}\frac{\boldsymbol{d}_{\boldsymbol{N}}}{\boldsymbol{N}}\times I\\
 & =\frac{\boldsymbol{d}_{\boldsymbol{N}}}{\sqrt{N_{q}}}\left\langle \frac{\left(\boldsymbol{a}^{*}(\mathbf{u})\right)^{\boldsymbol{N}}}{\sqrt{\boldsymbol{N}!}}\omega,\mathcal{W}(\mathbf{N}^{\odot1/2}\odot\mathbf{u}_{s})\mathcal{\widetilde{U}}^{*}(t)\phi^{(q)}(J_{q}u_{q,t})\widetilde{\mathcal{U}}(t)\omega\right\rangle _{\mathcal{F}^{\otimes p}}\\
 & \leq\frac{\boldsymbol{d}_{\boldsymbol{N}}}{\sqrt{N_{q}}}\left\Vert \sum_{k_{q}=1}^{\infty}\left(\mathcal{N}_{\mathrm{total}}+p\mathcal{I}\right)^{-\frac{5}{2}}\mathcal{P}_{\boldsymbol{\ell}}\mathcal{W}^{*}(\mathbf{N}^{\odot1/2}\odot\mathbf{u}_{s})\frac{\left(\boldsymbol{a}^{*}(\mathbf{u})\right)^{\boldsymbol{N}}}{\sqrt{\boldsymbol{N}!}}\omega\right\Vert _{\mathcal{F}^{\otimes p}}\\
 & \qquad\qquad\times\left\Vert \left(\mathcal{N}_{\mathrm{total}}+p\mathcal{I}\right)^{\frac{5}{2}}\widetilde{\mathcal{U}}^{*}(t)\phi^{(q)}(J_{q}u_{q,t})\widetilde{\mathcal{U}}(t)\omega\right\Vert _{\mathcal{F}^{\otimes p}}.
\end{align*}
The projection operator has come from parity conservation of $\widetilde{\mathcal{U}}$
for each components.

Let $K=\frac{1}{2}N^{\odot1/3}$so that Lemma \ref{lem:coherent_even_odd}
implies that
\begin{align*}
 & \left\Vert \sum_{k_{q}=1}^{\infty}\left(\mathcal{N}_{\mathrm{total}}+p\mathcal{I}\right)^{-\frac{5}{2}}\mathcal{P}_{\mathbf{\boldsymbol{\ell}}}\mathcal{W}^{*}(\mathbf{N}^{\odot1/2}\odot\mathbf{u}_{s})\frac{\left(\boldsymbol{a}^{*}(\mathbf{u})\right)^{\boldsymbol{N}}}{\sqrt{\boldsymbol{N}!}}\omega\right\Vert _{\mathcal{F}^{\otimes p}}^{2}\\
 & \leq\sum_{k_{q}=1}^{K}\left\Vert \left(\mathcal{N}_{\mathrm{total}}+p\mathcal{I}\right)^{-\frac{5}{2}}\mathcal{P}_{\mathbf{\boldsymbol{\ell}}}\mathcal{W}^{*}(\mathbf{N}^{\odot1/2}\odot\mathbf{u}_{s})\frac{\left(\boldsymbol{a}^{*}(\mathbf{u})\right)^{\boldsymbol{N}}}{\sqrt{\boldsymbol{N}!}}\omega\right\Vert _{\mathcal{F}^{\otimes p}}^{2}\\
 & \qquad\qquad+\frac{1}{K^{4}}\sum_{k_{q}=K}^{\infty}\left\Vert \left(\mathcal{N}_{\mathrm{total}}+p\mathcal{I}\right)^{-\frac{1}{2}}\mathcal{W}^{*}(\mathbf{N}^{\odot1/2}\odot\mathbf{u}_{s})\frac{\left(\boldsymbol{a}^{*}(\mathbf{u})\right)^{\boldsymbol{N}}}{\sqrt{\boldsymbol{N}!}}\omega\right\Vert _{\mathcal{F}^{\otimes p}}^{2}
\end{align*}
Since 
\begin{align*}
 & \left\Vert \left(\mathcal{N}_{\mathrm{total}}+p\mathcal{I}\right)^{-\frac{5}{2}}\mathcal{P}_{\mathbf{\boldsymbol{\ell}}}\mathcal{W}^{*}(\mathbf{N}^{\odot1/2}\odot\mathbf{u}_{s})\frac{\left(\boldsymbol{a}^{*}(\mathbf{u})\right)^{\boldsymbol{N}}}{\sqrt{\boldsymbol{N}!}}\omega\right\Vert _{\mathcal{F}^{\otimes p}}\\
 & \leq\prod_{q=1}^{p}\left\Vert \left(\mathcal{N}_{\mathrm{total}}+p\mathcal{I}\right)^{-\frac{5}{2}}P_{\ell_{q}}W^{*}(N^{1/2}u_{q,s})\frac{\left(a^{*}(u_{q})\right)^{N_{q}}}{\sqrt{N_{q}!}}\Omega\right\Vert _{\mathcal{F}}\\
 & \leq\left\Vert \left(\mathcal{N}_{\mathrm{total}}+p\mathcal{I}\right)^{-\frac{5}{2}}P_{\ell_{q}}W^{*}(N^{1/2}u_{q,s})\frac{\left(a^{*}(u_{q})\right)^{N_{q}}}{\sqrt{N_{q}!}}\Omega\right\Vert _{\mathcal{F}}\\
 & \qquad\times\prod_{r\neq q}\left\Vert \left(\mathcal{N}_{\mathrm{total}}+p\mathcal{I}\right)^{-\frac{1}{2}}P_{\ell_{r}}W^{*}(N^{1/2}u_{r,s})\frac{\left(a^{*}(u_{r})\right)^{N_{r}}}{\sqrt{N_{r}!}}\Omega\right\Vert _{\mathcal{F}}\\
 & \leq\frac{1}{(k_{q}+1)^{5/2}}\frac{4}{d_{N_{q}}}\frac{(k_{q}+1)^{3/2}}{\sqrt{N_{q}}}\left(\prod_{r\neq q}\frac{4}{d_{N_{r}}}\right)\leq\frac{Cp^{Cp}}{(k_{q}+1)d_{N_{q}}\sqrt{N_{q}}},
\end{align*}
we get
\begin{align*}
 & \left\Vert \sum_{k_{q}=1}^{\infty}\left(\mathcal{N}_{\mathrm{total}}+p\mathcal{I}\right)^{-\frac{5}{2}}\mathcal{P}_{\mathbf{\boldsymbol{\ell}}}\mathcal{W}^{*}(\mathbf{N}^{\odot1/2}\odot\mathbf{u}_{s})\frac{\left(\boldsymbol{a}^{*}(\mathbf{u})\right)^{\boldsymbol{N}}}{\sqrt{\boldsymbol{N}!}}\omega\right\Vert _{\mathcal{F}^{\otimes p}}^{2}\\
 & \leq\left(\sum_{k_{q}=1}^{K}\frac{Cp^{Cp}}{k_{q}^{2}d_{\boldsymbol{N}}N_{q}}\right)+\frac{Cp^{Cp}}{N_{q}^{4/3}\boldsymbol{d}_{\boldsymbol{N}}}\leq\frac{Cp^{Cp}}{\boldsymbol{d}_{\boldsymbol{N}}N_{q}}.
\end{align*}
For the last factor, from (\ref{lem:NjUphiUtildeUphitildeU}) and
Lemma \ref{lem:dGammaQphiQ}, we get

\begin{align*}
\left\Vert \left(\mathcal{N}_{\mathrm{total}}+p\mathcal{I}\right)^{\frac{5}{2}}\mathcal{\widetilde{U}}^{*}(t)\phi^{(q)}(J_{q}u_{q,t})\mathcal{\widetilde{U}}(t)\omega\right\Vert _{\mathcal{F}^{\otimes p}} & \leq Cp^{Cp}e^{Kt}\left\Vert \left(\mathcal{N}_{\mathrm{total}}+p\mathcal{I}\right)^{\frac{5}{2}}\phi^{(q)}(J_{q}u_{q,t})\mathcal{\widetilde{U}}(t)\omega\right\Vert _{\mathcal{F}^{\otimes p}}\\
\leq C\|J_{q}u_{q,t}\|e^{Kt}\left\Vert \left(\mathcal{N}_{\mathrm{total}}+p\mathcal{I}\right)^{3}\mathcal{\widetilde{U}}(t)\omega\right\Vert _{\mathcal{F}^{\otimes p}} & \leq Cp^{Cp}\|J_{q}\|_{\mathrm{HS}}e^{Kt}\left\Vert \left(\mathcal{N}_{\mathrm{total}}+p\mathcal{I}\right)^{3}\omega\right\Vert _{\mathcal{F}^{\otimes p}}\leq Cp^{3}\|J_{q}\|_{\mathrm{HS}}e^{Kt}.
\end{align*}
Then we get
\begin{align*}
I & \leq Cp^{Cp}\|J\|_{\mathrm{HS}}e^{Kt}(\min_{q}c_{q})^{-1}N^{-1}.
\end{align*}
Similarly, we ontain
\[
II\leq Cp^{Cp}N^{-1}\left\Vert J^{(Q\cup R)}\right\Vert _{\mathrm{HS}}e^{Kt}\leq Cp^{Cp}N^{-1}\left\Vert J\right\Vert _{\mathrm{HS}}e^{Kt}.
\]
can be bounded using the same argument for the case $|Q|+|R|\geq2$.

After all, we have
\[
|E_{t}^{Q,R}(J)|\leq Cp^{Cp}\|J\|_{\mathrm{HS}}e^{Kt}(\min_{q}c_{q})^{-1}N^{-1}
\]
which is desired conclusion.
\end{proof}

\section{Lemmas\label{sec:Lemmas}}
\begin{lem}
\label{lem:coherent_all} There exists a constant $C>0$ independent
of $N$ such that, for any $\varphi\in L^{2}(\mathbb{R}^{3})$ with
$\|\varphi\|=1$, we have 
\[
\left\Vert \left(\mathcal{N}_{\mathrm{total}}+p\mathcal{I}\right)^{-1/2}\mathcal{W}^{*}(\mathbf{N}^{\odot1/2}\odot\mathbf{u})\frac{\left(\boldsymbol{a}^{*}(\mathbf{u})\right)^{\boldsymbol{N-I}}}{\sqrt{\boldsymbol{N}!}}\omega\right\Vert _{\mathcal{F}^{\otimes p}}\leq\frac{C}{\boldsymbol{d}_{\boldsymbol{N}}}.
\]
\end{lem}

\begin{proof}
We follow the one particle result in \cite[Lemma 6.3]{ChenLee2011}.
Let
\[
\mathcal{A}_{\mathbf{n}}=\left\Vert \mathcal{P}_{\mathbf{n}}\mathcal{W}^{*}(\mathbf{N}^{\odot1/2}\odot\mathbf{u})\frac{\left(\boldsymbol{a}^{*}(\mathbf{u})\right)^{\boldsymbol{N-I}}}{\sqrt{\boldsymbol{N}!}}\omega\right\Vert _{\mathcal{F}^{\otimes p}}
\]
for $n_{q}<N_{q}$ for $q=1,\dots,p$, then noting that $\mathbf{n}=(n_{1},\dots,n_{p})$
and $\mathbf{N}=(N_{1},\dots,N_{p}$), we get
\begin{align*}
\mathcal{A}_{\mathbf{n}} & =\left\Vert \mathcal{P}_{\mathbf{n}}\mathcal{W}^{*}(\mathbf{N}^{\odot1/2}\odot\mathbf{u})\frac{\left(\boldsymbol{a}^{*}(\mathbf{u})\right)^{\boldsymbol{N-I}}}{\sqrt{\boldsymbol{N}!}}\omega\right\Vert _{\mathcal{F}^{\otimes p}}\\
 & =\prod_{q=1}^{p}\left\Vert P_{n_{q}}W^{*}(\sqrt{N_{1}}u)\frac{\left(a^{(q)*}(u_{q})\right)^{N_{q}-1}}{\sqrt{N_{q}!}}\Omega\right\Vert _{\mathcal{F}}.
\end{align*}
According to \cite[Lemma 6.3]{ChenLee2011}, we have
\[
\left\Vert P_{n_{q}}W^{*}(\sqrt{N_{1}}u)\frac{\left(a^{*}(u_{q})\right)^{N_{q}-1}}{\sqrt{N_{q}!}}\Omega\right\Vert _{\mathcal{F}}<CN_{q}^{-1/4}n_{q}^{-1/4}.
\]
Moreover, for $n_{q}\geq N_{q}$, using (6.16) of \cite[Lemma 6.3]{ChenLee2011}
the sum of its squares from $N_{q}$ to infinity is bounded by $1$,
i.e.,
\[
\sum_{n_{q}=N_{q}}^{\infty}\left\Vert P_{n_{q}}W^{*}(\sqrt{N_{q}}u)\frac{(a^{*}(u_{q}))^{N_{1}-1}}{\sqrt{N_{q}!}}\Omega\right\Vert _{\mathcal{F}}^{2}<1.
\]
Thus,
\begin{align*}
 & \left\Vert \left(\mathcal{N}_{\mathrm{total}}+p\mathcal{I}\right)^{-1/2}\mathcal{W}^{*}(\mathbf{N}^{\odot1/2}\odot\mathbf{u})\frac{\left(\boldsymbol{a}^{*}(\mathbf{u})\right)^{\boldsymbol{N-I}}}{\sqrt{\boldsymbol{N}!}}\omega\right\Vert _{\mathcal{F}^{\otimes p}}^{2}\\
 & =\sum_{\mathbf{n}\in\mathbb{Z}_{\geq0}^{d}}\frac{|A_{\mathbf{n}}|^{2}}{(n_{\mathbf{p}}+1)^{2}}\\
 & =\prod_{q=1}^{p}\left(\sum_{n_{q}=0}^{\infty}\frac{1}{(n_{q}+1)^{2}}\left\Vert \mathcal{P}_{n_{q}}W^{*}(\sqrt{N_{q}}u_{q})\frac{(a^{*}(u_{q}))^{N_{q}-1}}{\sqrt{N_{q}!}}\Omega\right\Vert _{\mathcal{F}}^{2}\right)\\
 & \leq\prod_{q=1}^{p}\frac{C}{N_{q}^{1/2}}\leq\frac{C}{d_{\boldsymbol{N}}^{2}}.
\end{align*}
This proves the desired lemma.
\end{proof}
\begin{lem}
\label{lem:coherent_even_odd} Let $P_{m}$ be the projection onto
the $m$-particle sector of the Fock space $\mathcal{F}$ for a non-negative
integer $m$. Then, for any non-negative integer $k_{q}\leq(1/2)N^{1/3}$,
\[
\left\Vert P_{2k_{q}}W^{*}(\sqrt{N_{1}}u)\frac{\left(a^{*}(u_{q})\right)^{N_{q}-1}}{\sqrt{N_{q}!}}\Omega\right\Vert _{\mathcal{F}}\leq\frac{4}{d_{N_{q}}}
\]
\[
\left\Vert P_{2k_{q}+1}W^{*}(\sqrt{N_{1}}u)\frac{\left(a^{*}(u_{q})\right)^{N_{q}-1}}{\sqrt{N_{q}!}}\Omega\right\Vert _{\mathcal{F}}\leq\frac{4}{d_{N_{q}}}\cdot\frac{(k_{q}+1)^{3/2}}{\sqrt{N_{q}}}
\]
\end{lem}

\begin{proof}
Proof can be found in \cite[Lemma 7.2]{Lee2013}.
\end{proof}

\section*{Acknowledgments}

The author would like to thank gratefully Ji Oon Lee for many motivational
and helpful discussions.

\bibliographystyle{abbrv}
\bibliography{refs}

\end{document}